\documentclass[aps,prl,twocolumn,superscriptaddress,10pt,article,showpacs,longbibliography]{revtex4-2}

\usepackage{blindtext}
\usepackage{lipsum}
\usepackage{graphics}
\usepackage{amsmath}
\usepackage{graphicx}
\usepackage{graphics}
\usepackage{amssymb}
\usepackage{pifont}
\usepackage{verbatim}
\usepackage{physics}
\usepackage{float}
\usepackage[normalem]{ulem}
\usepackage[dvipsnames]{xcolor}
\usepackage{bbm}
\usepackage{enumitem} 
\usepackage{multirow}

\usepackage{graphicx}
\usepackage{bm}
\usepackage{amsmath}
\usepackage{physics}
\usepackage{amssymb}
\usepackage{amsfonts}
\usepackage{amsthm}
\usepackage{bbm}
\usepackage{mathtools}
\usepackage{braket}
\usepackage[normalem]{ulem}
\usepackage{wrapfig}
\usepackage{tikz}
\usepackage{dsfont}
\usepackage{comment}
\usepackage{thmtools,thm-restate}
\usepackage[export]{adjustbox}

\definecolor{blueviolet}{rgb}{0.2, 0.2, 0.6}
\definecolor{webgreen}{rgb}{0,.5,0}
\definecolor{webbrown}{rgb}{.6,0,0}
\usepackage[bookmarks=false,
	colorlinks=true, 
	urlcolor=webbrown, 
	linkcolor=blueviolet, 
	citecolor=webgreen,
	pdfstartpage=1,
	pdfstartview={FitH},  
	bookmarksopen=false
	]{hyperref}
    
\newtheorem{theorem}{Theorem}
\newtheorem{definition}{Definition}
\newtheorem{corollary}{Corollary}
\newtheorem{lemma}{Lemma}
\newtheorem{fact}{Fact}

\newtheorem*{theorem*}{Theorem}

\newtheorem*{task*}{Task}
\newtheorem*{proposition*}{Proposition}

\DeclareMathOperator{\poly}{poly}

\newcommand{\bs}{\boldsymbol}
\newcommand{\be}{\begin{equation}}
\newcommand{\ee}{\end{equation}}

\newcommand{\eps}{\varepsilon}

\newcommand{\PD}{\Pi_{\text{dist}}}

\newcommand{\expect}{\mathop{\mathbb{E}}}

\newcommand{\dist}{\mathrm{dist}}

\DeclareMathOperator*{\E}{{\mathbb{E}}}

\begin{document}

\title{Unitary designs in nearly optimal depth}

\author{Laura Cui}
\thanks{These authors contributed equally to this work.}
\affiliation{Institute for Quantum Information and Matter and Department of Physics, California Institute of Technology,
Pasadena, California 91125, USA}

\author{Thomas Schuster}
\thanks{These authors contributed equally to this work.}
\affiliation{Walter Burke Institute for Theoretical Physics, California Institute of Technology, Pasadena, California 91125, USA}
\affiliation{Institute for Quantum Information and Matter and Department of Physics, California Institute of Technology,
Pasadena, California 91125, USA}
\affiliation{Google Quantum AI, Venice, California 90291, USA}

\author{Fernando Brand{\~a}o}
\affiliation{AWS Center for Quantum Computing, Pasadena, California 91125, USA}
\affiliation{Institute for Quantum Information and Matter and Department of Physics, California Institute of Technology,
Pasadena, California 91125, USA}

\author{Hsin-Yuan Huang}
\affiliation{Institute for Quantum Information and Matter and Department of Physics, California Institute of Technology,
Pasadena, California 91125, USA}
\affiliation{Google Quantum AI, Venice, California 90291, USA}

\begin{abstract}
We construct $\varepsilon$-approximate unitary $k$-designs on $n$ qubits in circuit depth $\mathcal{O}(\log k  \log \log n k / \varepsilon)$.
The depth is exponentially improved over all known results in all three parameters $n$, $k$, $\varepsilon$.
We further show that each dependence is optimal up to exponentially smaller factors.
Our construction uses $\tilde{\mathcal{O}}(nk)$ ancilla qubits and $\mathcal{O}(nk)$ bits of randomness, which are also optimal up to $\log(n k)$ factors.
An alternative construction achieves a smaller ancilla count $\tilde{\mathcal{O}}(n)$ with circuit depth $\mathcal{O}(k  \log \log nk/\varepsilon)$.
To achieve these efficient unitary designs, we introduce a highly-structured random unitary ensemble that leverages long-range two-qubit gates and low-depth implementations of random classical hash functions.
We also develop a new analytical framework for bounding errors in quantum experiments involving many queries to random unitaries.
As an illustration of this framework's versatility, we provide a succinct alternative proof of the existence of pseudorandom unitaries.
\end{abstract}

\maketitle

Random unitaries are ubiquitous across quantum science, serving both as fundamental theoretical tools and practical building blocks for quantum technologies.
They provide useful models for understanding chaotic many-body dynamics~\cite{fisher2023random,nahum2017entgrowth,cotler2022fluctuations}, quantum gravity phenomena~\cite{sekino2008fast,hayden2007black,brown2023quantum}, and thermalization in isolated quantum systems~\cite{deutsch1991quantum,srednicki1994chaos,rigol2008thermalization}.
Beyond their theoretical significance, random unitaries have been essential for device benchmarking~\cite{emerson2005scalable,knill2008randomized,elben2023randomized}, state tomography~\cite{guta2020fast, huang2020predicting,zhao2021fermionic}, quantum advantage demonstrations~\cite{arute2019quantum, morvan2023phase, abanin2025constructive}, and quantum cryptography~\cite{ji2018pseudorandom,ananth2022cryptography,kretschmer2023quantum}.
From an analytical perspective, the uniform Haar measure over unitaries enables tractable mathematical investigations through its elegant structure and a wealth of established results in random matrix theory~\cite{weingarten1978asymptotic,collins2003moments,collins2006integration, mehta2004random,tao2012topics,anderson2010introduction}.

These fundamental roles have sparked an immense effort to understand in what situations and time-scales quantum systems can realize random unitaries.
This effort is spearheaded by two notions of approximate random unitaries: \emph{unitary $k$-designs} \cite{emerson2003pseudo,gross2007evenly,dankert2005efficient,dankert2009exact,brandao2016local, haah2024efficient, chen2024incompressibility} and \emph{pseudorandom unitaries} (PRUs) \cite{metger2024simple, chen2024efficient, ma2024construct}.
Intriguingly, recent work has revealed that both objects can form in drastically lower times, or \emph{circuit depths}, than previously thought.
In particular, unitary designs on $n$ qubits can form in depth $\tilde{\mathcal{O}}(k) \cdot \log (n/\varepsilon)$ on any circuit geometry, where $\varepsilon$ quantifies the approximation error~\cite{laracuente2024approximate,schuster2024random}.
Meanwhile, PRUs can form in depth $\text{poly}(\log n)$ in one-dimensional systems, and depth $\text{poly}(\log \log n)$ in systems with long-range two-qubit gates~\cite{schuster2024random}.
At the heart of these findings is a surprising realization: that quantum systems can appear Haar-random even when their light-cones are only a small fraction of the total system size.

Despite this progress, key open questions remain, especially regarding unitary designs.
While existing circuit depths for designs are optimal for \emph{one-dimensional} systems, lower bounds suggest they may be \emph{exponentially worse} than optimal for more general circuit architectures~\cite{schuster2024random}.
Resolving this question could have substantial practical importance, as many leading quantum computing platforms feature long-range circuit geometries~\cite{bluvstein2024logical,iqbal2024non,slussarenko2019photonic}.

\begin{figure}[t]
    \centering
    \includegraphics[width=\linewidth]{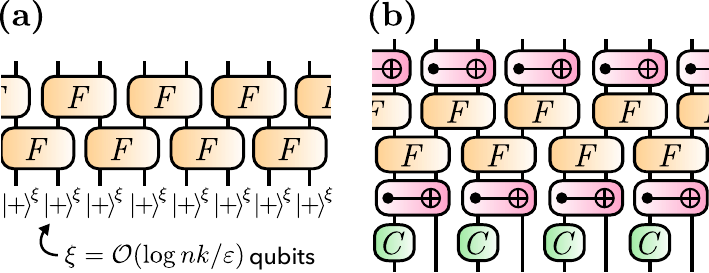}
    \caption{Schematic of our low-depth constructions of state and unitary designs. Black vertical lines denote local patches of $\xi = \mathcal{O}(\log nk/\varepsilon)$ qubits. 
    \textbf{(a)} Our random state designs apply a two-layer circuit of random phase gates $F$ to the plus state. The random phases are drawn from $k$-wise independent hash functions on $2\xi$ bits.  
    \textbf{(b)} Our random unitary designs sandwich the two-layer phase circuit
    between local 2-designs $C$ and conditional-shuffle gates. The latter  shuffle bitstrings on even patches conditional on odd patches (and vice versa), and are also instantiated using $k$-wise independent  functions.} 
    \label{fig: 1}
\end{figure}

In this work, we show that, with long-range two-qubit gates, $\varepsilon$-approximate unitary $k$-designs over $n$ qubits can be realized in a remarkably low circuit depth of $\mathcal{O}(\log k \cdot \log \log nk/\varepsilon)$.
This represents an exponential improvement in the dependence on all three parameters $n, k$, and $\varepsilon$ compared to all known results; the exponentially reduced dependence on $k$ is particularly striking.
Our construction uses $\tilde{\mathcal{O}}(nk)$ ancilla qubits (which begin and end in the zero state), $\mathcal{O}(nk)$ bits of randomness, and two-qubit gates of range $\mathcal{O}(\log nk/\varepsilon)$.
We achieve these designs by introducing a novel random unitary ensemble (Fig.~\ref{fig: 1}) that combines existing classical hash function constructions with inherently quantum ingredients to enable ultra-low circuit depths for random unitary designs.

\begin{table*}
\begin{ruledtabular}
\begin{tabular}{lccc}
 Unitary $k$-designs & Circuit depth & Ancilla count & Error \\
 \hline
 Random 1D circuit~\cite{brandao2016local,chen2024incompressibility} & $\mathcal{O}((nk+\log1/\varepsilon) \cdot \poly\log k)$ & None & Relative \\
 Blocked random 1D circuit~\cite{schuster2024random} & $\mathcal{O}(k \cdot \poly\log k \cdot \log n/\varepsilon)$ & None & Relative \\
 PFC circuit~\cite{metger2024simple,ma2024construct} & $\mathcal{O}(nk \cdot \poly\log(n))$ & None & Measurable \\
 Blocked LRFC circuit \emph{(this work)} & $\mathcal{O}(\log k \cdot \log \log  nk/\varepsilon)$ & $k n \cdot \tilde{\mathcal{O}}(\log \log n / \varepsilon)$ & Measurable \\
  & $\mathcal{O}(k \cdot \log \log nk/\varepsilon)$ & $n \cdot \tilde{\mathcal{O}}(\log \log n k / \varepsilon)$ & Measurable \vspace{1mm} \\
 \hline \hline
 State $k$-designs & Circuit depth & Ancilla count & Error \\
 \hline
 Random phase state~\cite{ji2018pseudorandom,brakerski2019pseudo} & $\mathcal{O}(\log k \cdot \log n)$ & $k n \cdot \tilde{\mathcal{O}}(\log n)$ & Additive \\
 Blocked random phase state \emph{(this work)} & $\mathcal{O}(\log k \cdot \log \log nk/\varepsilon)$ & $k n \cdot \tilde{\mathcal{O}}(\log \log n / \varepsilon)$ & Additive \\
  & $\mathcal{O}(k \cdot \log \log nk/\varepsilon)$ & $n \cdot \tilde{\mathcal{O}}(\log \log n k / \varepsilon)$ & Additive \\
 Lower bound \emph{(this work)} & $\Omega(\log k + \log \log n/\varepsilon)$ & Any & Additive  \\
  & $\Omega(k + \log \log n/\varepsilon)$ & $\mathcal{O}(n)$ & Additive  \\
\end{tabular}
\end{ruledtabular}
\caption{\label{tab:design depths} Comparison of the circuit depth and ancilla count of several known constructions of designs.
The Permutation-Function-Clifford (PFC) ensemble and random phase state have fixed approximation error $\varepsilon = \mathcal{O}(k^2/2^n)$. The PFC and blocked LRFC circuits can be improved to relative error at the cost of an additional factor of $k$ in the circuit depth~\cite{metger2024simple,supp}.
}
\end{table*}

Our main contributions are detailed as follows.
First, we introduce our new random unitary ensemble and prove that it forms a design in the ultra-low circuit depths described above. We also provide a streamlined construction of low-depth state designs.
Second, we introduce a new notion of approximation error for unitary $k$-designs that captures the strongest experimentally-achievable distinguishability.
This \emph{measurable error} quantifies how well a unitary ensemble can be distinguished from Haar-random by any quantum experiment that makes $k$ queries to the unitary.
Our results hold under this strong form of approximation error \footnote{The expert reader may recall an even stronger notion of approximation error in the design literature, the \emph{relative error}. Our designs achieve relative error at the cost of an additional factor of $k$ in the circuit depth (Table~I). As discussed later in the main text, the relative error cannot be measured in any (sub-exponential-time) quantum experiment~\cite{supp}.}.
Third, we prove a complementary lower bound, showing that unitary designs, with any reasonable notion of error, require circuit depth at least $\Omega(\log k + \log \log n/\varepsilon )$.
To establish this bound, we provide a simple and efficient test that can distinguish any lower-depth quantum state or unitary from Haar-random. 
Fourth, a core technical contribution of our work is the development of a new analytic approach for proving that random unitary ensembles are indistinguishable from Haar-random.
To demonstrate the versatility of our approach, we apply it to give a succinct alternative proof of the existence of pseudorandom unitaries (following recent work by Ma and Huang~\cite{ma2024construct}).
We anticipate that this analytical framework will facilitate continued developments in the field.

\emph{Background.}---We begin with a short review of state and unitary designs. 
A random state ensemble $\mathcal{S}$ is a state $k$-design with additive error $\varepsilon$ if its $k$-th moment,
\begin{equation}
    \chi^{}_{\mathcal{S}} \equiv \E_{\psi \in \mathcal{S}} \big[ \dyad{\psi}^{\otimes k}\big],
\end{equation}
is equal to the $k$-th moment $\chi_H$ of the Haar ensemble on the unit sphere of $\mathbbm{C}^{2^n}$ up to small trace-norm error, $\lVert \chi_\mathcal{S} - \chi_H \rVert_1 \leq \varepsilon$.
This trace norm condition ensures that $\mathcal S$ is indistinguishable from Haar-random in any experiment that queries a random state $k$ times.

A random unitary ensemble $\mathcal E$ is a unitary $k$-design with additive error $\varepsilon$ if its $k$-th moment, i.e.~the channel
\begin{equation} \label{eq: def Phi}
    \Phi_\mathcal{E}(\cdot) \equiv \E_{U \in \mathcal{E}} \left[ U^{\otimes k} (\cdot) U^{\dagger,\otimes k} \right],
\end{equation}
is equal to the $k$-th moment $\Phi_H$ of the Haar ensemble on $U(2^n)$ up to small diamond-norm error, $\lVert \Phi_\mathcal{E} - \Phi_H \rVert_\diamond \equiv \max_\rho \lVert \Phi_\mathcal{E}(\rho) - \Phi_H(\rho) \rVert_1 \leq \varepsilon$. Here, $\rho$ could be supported on a larger system containing both the $n$ qubits and an arbitrary number of ancilla qubits.
As mentioned earlier, our unitary design results in fact hold under a much stronger notion of approximation error, the measurable error, which captures the maximum distinguishability achievable by any experiment making $k$ queries to $U$. We defer a formal definition to the technical sections below.

At present, the best-scaling constructions of state and unitary $k$-designs are one-dimensional random circuits.
Recent works have established that such circuits, in which each gate is drawn independently and randomly from the two-qubit unitary group, create designs in depth $\mathcal{O}(k \cdot \text{poly} \log k \cdot \log(n/\varepsilon))$~\cite{brandao2016local,chen2024incompressibility,schuster2024random}.
Two key challenges have prevented improvements over this scaling, even when allowing long-range two-qubit gates.
First, standard methods to amplify designs~\cite{brandao2016local} decrease the design error  exponentially in the circuit depth; achieving a sub-logarithmic depth scaling is thus difficult.
Second, the dependencies $\Omega(k)$ and $\Omega(\log n/\varepsilon)$ \emph{are} in fact optimal for random circuits~\cite{brandao2016local,aharonov2023polynomial,fefferman2024anti}, whose lack of structure prevents faster design formation.
Hence, any further improvements in design depths must involve highly-structured  ensembles.

\emph{Nearly optimal state designs.}---We now introduce our construction of state designs.
Our construction utilizes a binary phase state, $\ket{\psi} = \frac{1}{\sqrt{2^n}} \sum_{x \in \{0,1\}^n} (-1)^{f(x)} \ket{x}$.
When $f: \{0,1\}^n \rightarrow \{0,1\}$ is a random function, the phase state forms a $k$-design with exponentially small additive error, $k^2/2^n$~\cite{brakerski2019pseudo}.
However, this state has very high circuit depth, since a random function on $n$ qubits requires exponentially many gates to compute.

To design a more resource-efficient state ensemble, we make two modifications.
First, and most crucially, we introduce a blocked variant of the random phase state [Fig.~\ref{fig: 1}(a)].
Blocks of random phase gates are applied to an initial product state using a two-layer brickwork structure.
Each block acts on a local patch of $2\xi$ qubits as $F_{i} = \sum_{x_{i} \in \{0,1\}^{2\xi}} (-1)^{f_i(x_{i})} \dyad{x_{i}}$.
Blocks in the first and second layer overlap on $\xi$ qubits each.
We prove that this blocked random phase state forms a state $k$-design with additive error $\varepsilon = nk^2/2^\xi$.
Hence, the size of each local patch scales only logarithmically in all three parameters, $\xi = \log_2(nk^2/\varepsilon)$.

Second, following earlier works~\cite{brakerski2019pseudo,metger2024simple}, we replace each random function $f_{i}(x_{i})$ with a \emph{$2k$-wise independent random hash function}~\cite{wegman1981new} on $2\xi$ bits.
By definition, a $2k$-wise independent hash function reproduces the $2k$-th moment of a random function.
Hence, the resulting $2k$-wise blocked random phase state achieves a state $k$-design with the same error as before.
The factor of two accounts for each function appearing in both the bra and the ket~\cite{zhandry2021PRF}.

To implement each $2k$-wise random function, we borrow a standard construction from the classical  literature~\cite{wegman1981new,brakerski2019pseudo}.
Namely, we define $f_{i}(x_{i}) = \sum_{j=0}^{k-1} \alpha_{j} (x_i)^j$ where $\alpha_{j} \in \text{GF}(2^{2\xi})$ are randomly drawn from the Galois field of dimension $2^{2\xi}$.
Each bitstring $x_{i}$ is uniquely associated with an element of the Galois field $\text{GF}(2^{2\xi})$.
All arithmetic operations occur within the field $\text{GF}(2^{2\xi})$.
In the supplemental material~\cite{supp}, we provide a detailed accounting showing that such functions can be computed in quantum circuit depth $\mathcal{O}(\log k \cdot \log \xi)$ using $\mathcal{O}(\xi k  \log \xi  \log \log \xi)$ ancilla qubits, and depth $\mathcal{O}(k  \cdot \log \xi)$ using $\mathcal{O}(\xi  \log \xi  \log \log \xi)$ ancilla qubits.
From this construction, we immediately achieve state $k$-designs in quantum circuit depth $\mathcal{O}(\log k \cdot \log \log nk/\varepsilon)$ with $nk\cdot\tilde{\mathcal{O}}(\log \log nk/\varepsilon)$ ancilla qubits, and depth $\mathcal{O}(k \cdot \log \log nk/\varepsilon)$ using $n \cdot \tilde{\mathcal{O}}(\log \log nk/\varepsilon)$ ancilla qubits.

Our proof that the blocked random phase state forms a state $k$-design proceeds via a surprisingly simple argument [Fig.~\ref{fig: proof}(a-b)].
We consider $k$ copies of the state, decomposed in the computational basis $x^{(1)}, x^{(2)}, \ldots, x^{(k)}$ on each copy.
We then project onto a subspace where the local bitstrings $x_{i}$ at each patch are distinct across the $k$ copies, i.e.~$x_{i}^{(j)} \neq x_{i}^{(j')}$ for all patches $i$ and copies $j \neq j'$.
This projection incurs only a small trace-norm error, since the original state is a superposition over all bitstrings, and most bitstrings are locally distinct.
Finally, we show that \emph{on} the local distinct subspace, the two-layer blocked structure $\prod_i F_{i}$ acts identically to a global random phase operator $F$, which in turn acts identically to the Haar twirl.
These equivalences follow from straightforward computations.

\begin{figure*}[t]
    \centering
    \includegraphics[width=0.95\linewidth]{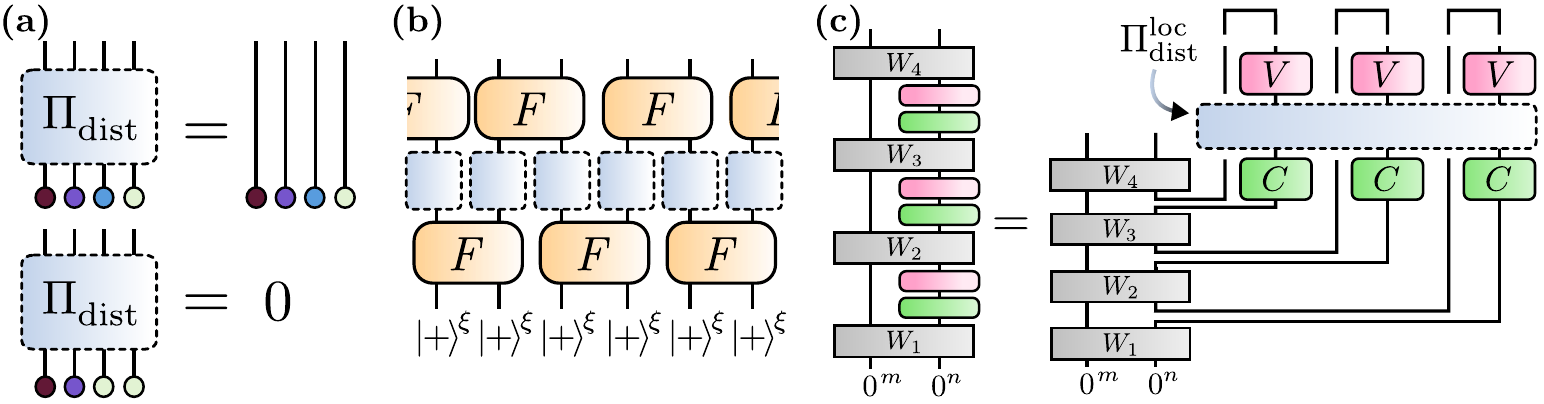}
    \caption{Illustration of key ideas from our design proofs. \textbf{(a)} The distinct subspace projector~\cite{metger2024simple} $\Pi_{\text{dist}}$ eliminates bitstrings that are equal on different copies of the state or unitary (bottom), while leaving ``distinct'' sets of bitstrings invariant (top). Here, each leg represents one of $k$ copies, and the colored circles indicate the $n$-bit string $x^{(j)}$ of a computational basis state on each copy ($j=1,\ldots,k$). \textbf{(b)} To prove that the blocked random phase state is a state $k$-design, we insert projectors onto the local distinct subspace between each random phase gate. On the local distinct subspace, the blocked random phase circuit acts identically to a Haar-random unitary, which completes the proof. \textbf{(c)} To prove that the blocked LRFC circuit forms a unitary $k$-design with small measurable error, we reformulate any quantum experiment involving sequential applications of a unitary $U = VC$ (with $V \equiv S_o \prod_i F_{i}$ and  $C \equiv S_e C_o$ for the blocked LRFC circuit) into an equivalent experiment involving parallel applications and post-selection on EPR states (top bars). We then insert projectors onto the local distinct subspace, $\Pi_{\text{dist}}^{\text{loc}}$,  between $V$ and $C$ in the parallel applications. On the local distinct subspace, the blocked LRFC circuit acts identically to a Haar-random unitary, which completes the proof.}
    \label{fig: proof}
\end{figure*}

\emph{Nearly optimal unitary designs.}---To construct low-depth random unitary designs, we combine the two-layer brickwork structure consisting of blocks of random phase gates with several new ingredients. 
These new ingredients are necessary because a quantum  experiment can apply a unitary to any initial state, including those outside the local distinct subspace.

Our unitary ensemble is inspired by the Luby-Rackoff-Function-Clifford (LRFC) ensemble, $U = S_L  S_R FC$, introduced in~\cite{SRU2025}.
Here, $S_L$ randomly shuffles the left $n/2$ qubits based on the right $n/2$ qubits, $S_L \ket{x_L , x_R} = \ket{x_L + h_L(x_R) , x_R}$, where $h_L: \{0,1\}^{n/2} \rightarrow \{0,1\}^{n/2}$ is a random function.
$S_R$ is the analogous operation for shuffling the right $n/2$ qubits, $F$ is a random phase unitary on all $n$ qubits, and $C$ is a random Clifford unitary.
The LRFC ensemble forms a unitary design with an exponentially small additive error, $6k^2/2^{n/2}$~\cite{SRU2025}.

Before presenting our simplest and most resource-efficient unitary ensemble, let us first describe an alternative approach.
Consider a two-layer blocked circuit similar to Fig.~\ref{fig: 1}, but where each block of gates is a random LRFC circuit on $2\xi$ qubits.
We know that the LRFC ensemble forms a design with additive error $6k^2/2^{\xi}$~\cite{SRU2025}.
Hence, the two-layer blocked circuit consisting of small LRFC circuits is indistinguishable from a two-layer blocked circuit consisting of small Haar-random unitaries.
From the gluing lemma of~\cite{schuster2024random}, the two-layer Haar-random blocked circuit forms a unitary design with error $nk^2/2^\xi$.
%
%
The shuffle operations $S_L$, $S_R$ and phase unitary $F$ can be implemented in short depth using $2k$-wise independent functions~\cite{wegman1981new,supp}.
The random Clifford $C$ can be replaced with an exact unitary 2-design, with circuit depth $\mathcal{O}(\log \xi)$ and ancilla count $\tilde{\mathcal{O}}(\xi)$~\cite{cleve2015near}.
Together, this yields a unitary design with circuit depth $\mathcal{O}(\log k \log \log nk/\varepsilon)$ and ancilla count $nk \cdot \tilde{\mathcal{O}}(\log \log nk/\varepsilon)$.

We can further simplify this unitary ensemble via the construction in Fig.~\ref{fig: 1}(b).
We consider the circuit, $U = S_o (\prod_i F_{i}) S_e C_o$, where $S_o$ is a tensor product of random shuffles of odd patches conditional on even patches, $\prod_i F_{i}$ is our two-layer blocked random phase circuit, $S_e$ shuffles even patches conditional on odd patches, and $C_o$ is a tensor product of random Clifford unitaries on each odd patch.
This \emph{blocked LRFC ensemble} requires fewer shuffle operations, and both fewer and smaller Clifford unitaries, compared to the two-layer blocked LRFC circuit.
Its simplicity will also enable a stand-alone design proof, independent of previous results.
As before, the ensemble has circuit depth $\mathcal{O}(\log k \log \log nk/\varepsilon)$ and ancilla count $\tilde{\mathcal{O}}(nk)$, or depth $\mathcal{O}( k \log \log nk/\varepsilon)$ and ancilla count $\tilde{\mathcal{O}}(n)$.
It uses $3nk + 5n/2$ bits of randomness~\cite{supp}, which is near the optimal value of $2nk$~\cite{gross2007evenly,roy2009unitary}.

\emph{Measurable error of unitary designs.}---Before describing our proof of unitary designs, let us briefly digress to discuss notions of approximation error in the design literature.
Unfortunately, the additive error for unitary designs provides only a weak guarantee, that a unitary appears Haar-random in experiments that query it $k$ times in \emph{parallel}~\cite{schuster2024random}.
To address this, one common approach uses the so-called \emph{relative error},  $(1-\varepsilon)\Phi_H \preceq \Phi_\mathcal{E} \preceq (1+\varepsilon) \Phi_H$, where $\preceq$ denotes the channel inequality~\footnote{Here, $\Phi \preceq \Phi'$ indicates that $\tr(M \Phi(\rho)) \leq \tr(M \Phi'(\rho))$ for all positive operators $\rho$ and $M$.}.
The relative error is remarkably strong: it can bound any property that can be efficiently measured in any experiment, as well as many exponentially small quantities that can never be measured.
While this strength can be convenient, it often represents an overly stringent and un-physical requirement.
Indeed, recent works have established unitary ensembles that have large relative error, yet cannot be distinguished from Haar-random unitaries by any efficient experiment~\cite{ma2024construct,SRU2025}.

To address this, we define a new, natural notion of error for unitary designs: the \emph{measurable error}.
A unitary ensemble has small measurable error if it cannot be distinguished from Haar-random by any experiment that queries it $k$ times.
To be precise, if we denote the output of a general quantum experiment as $\ket*{\psi_W^U} = W_{k+1} U W_k U \cdots U W_1 \ket{0}$ where $W_j$ are arbitrary quantum operations (involving arbitrary ancilla qubits) applied between successive queries to $U$~\cite{supp}, then the measurable error demands
\begin{equation}
    \left\lVert \E_{U \sim \mathcal{E}} \Big[ \dyad*{\psi^U_W} \Big] - \E_{U \sim H} \Big[ \dyad*{\psi^U_W} \Big] \right\rVert_1 \leq \varepsilon
\end{equation}
for all $W_j$.
This resembles the notion of adaptive security in the context of PRUs \footnote{We remark that the adjectives ``adaptive'' and ``non-adaptive'' do not perfectly capture the distinction between experiments that query a unitary in parallel versus general experiments. For example, an experiment could query a unitary on many copies in parallel and then perform adaptive measurements on the many copies.}.
The measurable error is substantially stronger than the additive error, which only concerns experiments in which $U^{\otimes k}$ is applied once in parallel.
Moreover, if a unitary ensemble fails to have small measurable error, there always exists an experiment that can measure that the  ensemble is not Haar-random~\cite{supp}.

To prove that the blocked LRFC ensemble forms a unitary design with small measurable error, we develop a new general approach for bounding the measurable error of random unitary ensembles.
Our approach combines and refines recent ideas from the study of PRUs.
At a high-level, it proceeds in three steps.
First, we reformulate any experiment involving sequential applications of a random unitary as an equivalent experiment involving parallel applications and post-selection [see Fig.~\ref{fig: proof}(c)].
Then, we prove that one can insert a projector onto the local distinct subspace between key components of the parallel random unitaries in the post-selected experiment.
Finally, we show that the ensemble of interest, i.e.~the blocked LRFC ensemble, acts identically to a Haar-random unitary on the local distinct subspace, thereby completing the proof.
Following these steps, we prove that the blocked LRFC ensemble has measurable error $\varepsilon = 3(n/\xi)k^2/2^\xi$, which leads to the stated depths of our unitary designs.

\emph{Minimum depth of designs.}---Given rapid recent  improvements in design circuit depth, a natural question to ask is: Can one do any better? Unfortunately, while a simple counting argument shows that a circuit depth $\Omega(\log k)$ is indeed optimal~\cite{brandao2016local,supp} (and thus our $k$-dependence is optimal up to an exponentially smaller $\log \log k$ factor), existing lower bounds leave gaps in the dependence on $n$ and $\varepsilon$. Specifically, lower bounds for $n$ apply only to the relative error, while bounds for $\varepsilon$ apply only to unitaries and not states~\cite{schuster2024random}.

Here, we resolve this by introducing a simple and efficient test to distinguish any sufficiently low-depth quantum state from Haar-random.
Our test implies that $n$-qubit state and unitary designs with additive error $\varepsilon$ require circuit depth $\Omega(\log \log n/\varepsilon)$, as we achieve.
%
%
Our test proceeds as follows: We choose a random product basis, and measure two identical copies of the state of interest in the chosen basis.
We then divide the measurement outcomes, represented as $n$-bit strings, into local patches of $\sim \! L$ bits each, where $L$ is the size of the light-cone of the state preparation circuit.
Then, we count the number of patches where the two measurement outcomes are equal (i.e.~collide).
We prove that for any low-depth state, the number of collisions is high, whereas in a Haar-random state it is very low.
A detailed analysis yields the precise scaling with $n$ and $\varepsilon$~\cite{supp}.

Physically, our test builds upon the intuition from~\cite{schuster2024random} that large numbers of non-commuting observables are essential for quantum circuits to appear random.
We show that for circuits of insufficient depth, an extensive number of observables both have high expectation value in the state, and commute with the random product measurement.
This causes the measurement distribution to retain large amounts of information about the state.
This manifests in some local bitstrings appearing (and hence, colliding) more often than others.

\begin{figure}[t]
\includegraphics[width=0.95\linewidth]{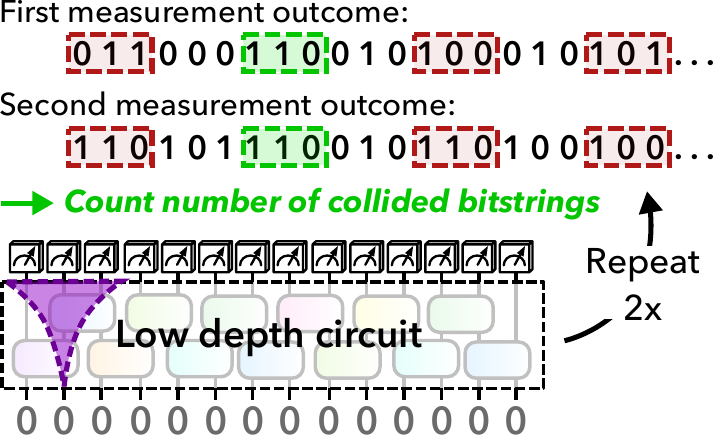}
 \caption{Illustration of our lower bound on state and unitary designs.
 We consider the output distribution when the state of interest is measured in a random product basis.
 We show that with two measurements, one can efficiently distinguish any state with circuit depth less than $\log \log n/\varepsilon$ from Haar-random, by counting local collisions (green) in the   measurement outcomes. The collisions are considered on regions (red/green) of the same size as the circuit light-cone (purple).}
\label{fig:lower-bound}
\end{figure}

\emph{Application to pseudorandom unitaries.}---Finally, although our work focuses primarily on unitary designs, the techniques we develop also have immediate applications to the study of PRUs.
A PRU is any efficient unitary ensemble indistinguishable from Haar-random unitaries in all computationally-bounded quantum experiments~\cite{ji2018pseudorandom}.
The standard approach to proving the security of PRUs involves, as an intermediary step, proving that a more-random version of the ensemble is a design with small measurable error~\cite{brakerski2019pseudo,metger2024simple,ma2024construct,SRU2025}.

Thus far, the only existing method for proving the security of PRUs is the path-recording framework proposed in~\cite{ma2024construct}, which relies on analyzing sums of computational paths through computational basis states.
Our proof that our unitary ensemble [Fig.~\ref{fig: 1}(b)] has small measurable error introduces a new methodology that complements  existing techniques.
To illustrate the versatility of our approach, we apply it in the supplemental material~\cite{supp} to give two short and self-contained proofs that the PFC~\cite{metger2024simple,ma2024construct} and LRFC~\cite{SRU2025} ensembles form PRUs.

\emph{Outlook.}---Our work comes near to closing a long line of research on the circuit depths of random unitary designs.
Nonetheless, several questions remain open.
First, constructing truly optimal designs, which requires decoupling the $\log (k)$ and $\log \log(n /\varepsilon)$ dependences, remains open.
Second, applications of $k$-designs with $k > 3$ remain comparatively sparse. 
One theoretical application concerns the growth of quantum circuit complexity~\cite{brandao2016local,chen2024incompressibility}; here, our results prove that the complexity of unitaries without ancilla qubits can grow exponentially in the depth of the unitary with ancilla qubits~\cite{sun2023asymptotically}.
A second application involves using higher-order designs to lower bound the probability that individual random unitaries will mimic lower-order average behavior~\cite{nietner2023average,schuster2024random}.
Finally, while our construction of unitary $k$-designs is nearly optimal in terms of scaling, it is not readily implementable on near-term quantum computing architectures.
The development of practical implementations of unitary and state designs remains an important open direction.

\textit{Acknowledgements}---We are tremendously grateful to Gregory D. Kahanamoku-Meyer for several  discussions on the resource requirements for quantum arithmetic. T.S. acknowledges support from the Walter Burke Institute for Theoretical Physics at Caltech.
T.S. and H.H. acknowledge support from the U.S. Department of Energy, Office of Science, National Quantum Information Science Research Centers, Quantum Systems Accelerator.
The Institute for Quantum Information and Matter is an NSF Physics Frontiers Center.

\noindent \textit{Note Added 1.} During the writing of this work, we learned of Ref.~\cite{zhang2025designs}, which constructs  unitary designs in depth $\mathcal{O}(\log \log n/\varepsilon + 2^{\mathcal{O}(k \log k)})$ and  state designs in depth $\mathcal{O}(\log \log n/\varepsilon + \log k)$ for $k < n$. Their unitary and state ensembles are very different than our own. Our lower bound proves that their construction of state designs is precisely optimal.
We thank the authors for sharing their results and several enlightening discussions.

\noindent \textit{Note Added 2.} After posting this work, we learned of Ref.~\cite{Deep2025}, which independently introduces the blocked random phase state and proves that it forms a pseudorandom state when each $f_i$ is instantiated pseudorandomly.


\let\oldaddcontentsline\addcontentsline
\renewcommand{\addcontentsline}[3]{}
\bibliography{refs}
\let\addcontentsline\oldaddcontentsline


\onecolumngrid
\newpage

{\centering
\large\bfseries
Supplementary Material: Unitary designs in nearly optimal depth
\par}

\tableofcontents


\section{Preliminaries}

In this section, we establish the foundational concepts necessary for our analysis of designs and their distinguishability properties. We begin by reviewing the mathematical framework for averaging over the unitary group via the Haar measure, then proceed to define state and unitary designs with their respective error measures.

\subsection{Haar averaging over the unitary group}

We first review properties of the moments of Haar-random unitaries, which serve as the reference point against which we compare other unitary ensembles. Let $\ket{x}$ denote a computational basis state on a system of $n$ qubits corresponding to the bitstring $x \in \{0,1\}^n$. The Haar measure provides the unique translation-invariant measure on the unitary group, defining a uniform distribution over all unitaries on $n$ qubits.
\begin{definition}[Averaging over the unitary group]
    Given a linear operator $X$ acting on $nk$ qubits, the $k$-th moment with respect to $U(2^n)$ is defined via the twirl over the unitary group:
    \begin{equation}
    \Phi_{H}(X) = \int dU \, U^{\otimes k} X (U^\dagger)^{\otimes k},
    \end{equation}
    where we have left out the implicit dependence on $k$.
\end{definition}
This definition captures how the operator $X$ transforms when we apply $k$ independent copies of a Haar-random unitary $U$. While individual Haar-random unitaries have no predictable structure, their statistical moments possess well-defined mathematical properties that can be characterized precisely.

The structure of these moments can be determined using representation theory. The key observation is that these moments must be invariant under any unitary change of basis applied to each copy separately. By applying Schur-Weyl duality \cite{fulton2013representation, goodman2009symmetry}, this symmetry constraint determines the form of the moments:
\begin{lemma}[Explicit form in terms of permutations]
    For any linear operator $X$ acting on $nk$ qubits, the $k$-th moment with respect to the unitary group can be written in the form 
    \begin{equation}
    \Phi_{H}(X) = \sum_{\sigma, \tau \in S_k} c_{\sigma, \tau} \cdot \tr(\sigma X) \cdot \tau.
    \end{equation}
    The coefficients $c_{\sigma,\tau}$ depend on $k$ and the Hilbert space dimension $2^n$.
\end{lemma}
In principle, all of the coefficients $c_{\sigma, \tau}$ can be computed exactly using combinatorial formulas derived from the Weingarten calculus of the unitary group \cite{collins2006integration}. However, the main property of the moments that we will use in the following sections is their relationship to the symmetric group on $k$ copies of the physical systems.

\subsection{State $k$-designs}

We now turn to state ensembles and their ability to mimic the statistical properties of Haar-random states under measurements on multiple copies. State designs provide a framework for constructing explicit state ensembles that approximate the behavior of Haar-random states for experiments involving a bounded number of copies. For quantum states, the approximation is typically characterized using trace distance $\norm{\cdot}_1$.
\begin{definition}[State $k$-design]
    An ensemble of states $\mathcal S$ is an $\varepsilon$-approximate quantum state $k$-design if it reproduces the twirl of the unitary group on $k$ copies of any state on $n$ qubits up to additive error $\varepsilon$, or equivalently,
    \begin{equation}
    \norm{\chi_\mathcal{S} - \chi_H}_1 \leq \varepsilon.
    \end{equation}
    where we have used the abbreviated notation
    \begin{equation}
    \chi_\mathcal{S} = \expect_{\ket{\phi} \sim \mathcal{S}} (|\phi \rangle\!\langle \phi|)^{\otimes k}
    \end{equation}
    to denote the $k$-th moment over the state ensemble $S$.
\end{definition}
This definition characterizes when an ensemble $\mathcal{S}$ cannot be distinguished from Haar-random states by any measurement performed on $k$ copies of a random state. The moments of Haar-random states are maximally mixed over the symmetric subspace on the $k$ copies~\cite{harrow2013church}.

\subsubsection{Example: Binary phase state}

To illustrate the construction of state designs, we review binary phase state ensembles, first introduced in \cite{ji2018pseudorandom} with statistical properties analyzed in \cite{brakerski2019pseudo}. These provide a concrete construction of pseudorandom quantum states using classical Boolean functions.
\begin{definition}[Random binary phase state] \label{def:randomphasestate}
    Let $f:\{0,1\}^n \to \{0,1\}$ be a uniformly random Boolean function. The corresponding \emph{random binary phase state} on $n$ qubits is defined as
    \begin{equation}
        \ket{\phi_f} = \frac{1}{2^{n/2}} \sum_{x \in \{0,1\}^n} (-1)^{f(x)} \ket{x}.
    \end{equation}
\end{definition}

The construction begins with the uniform superposition state and applies random phase factors determined by the Boolean function $f$.

\begin{lemma}[Binary phase states are designs \cite{brakerski2019pseudo}] \label{lemma:randomphasemoment}
    The ensemble of random binary phase states from Definition~\ref{def:randomphasestate} forms an $\varepsilon$-approximate state $k$-design with $\varepsilon = 4k^2/2^n.$
\end{lemma}

\subsection{Unitary $k$-designs}

We now establish the framework for unitary designs, which are ensembles of unitary transformations that approximately reproduce the statistics of Haar-random unitaries. Unlike state designs, unitary designs concern dynamical processes that can be applied to arbitrary input states and must withstand more sophisticated experimental probes.

\subsubsection{Additive Error and Parallel Security}

The most commonly studied error measure is additive error, which is based on the differences between the moments in diamond norm. Recall that diamond distance concerns the maximum trace distance on the output states for any input states of a much bigger system size.
\begin{definition}[Unitary $k$-design up to additive error]
    An ensemble of unitaries $\mathcal{E}$ is a \emph{unitary $k$-design} if it reproduces the first $k$ moments of the Haar measure. In addition, given some $\varepsilon > 0$, the ensemble $\mathcal{E}$ is an approximate unitary $k$-design up to \emph{additive error} $\varepsilon$ if
    \begin{equation}
    \norm{\Phi_\mathcal{E} - \Phi_H}_\diamond \leq \varepsilon,
    \end{equation}
    where we have used the abbreviated notation
    \begin{equation}
    \Phi_\mathcal{E}(X) = \E_{U \sim \mathcal{E}} U^{\otimes k} X (U^\dagger)^{\otimes k}
    \end{equation}
    to denote the $k$-th moment over the unitary ensemble $\mathcal E$.
\end{definition}
\noindent Explicitly, the diamond norm is given by $\lVert \Phi - \Phi' \rVert_\diamond = \max_\rho \lVert \Phi(\rho) - \Phi'(\rho) \rVert_1$, where the maximization is over all states $\rho$ on $nk+m$ qubits, where the number $m$ of ancilla qubits may be arbitrarily large.

The diamond norm measures the maximum distinguishability of two quantum channels.
The channels $\Phi_\mathcal{E}$ and $\Phi_H$ apply $k$ copies of a random unitary $U$ in parallel.
It follows that the additive error measures the maximum distinguishability of the ensemble $\mathcal E$ and the Haar-ensemble in any quantum experiment that queries $k$ copies of the unitary $U$ in parallel.
%
%
\begin{lemma}[Additive error is equivalent to parallel security]
    An ensemble $\mathcal{E}$ is an approximate unitary $k$-design up to additive error $\varepsilon$ if and only if for any quantum algorithm making a single query to $U^{\otimes k}$, i.e.~$k$ parallel queries to $U$, the output states when $U$ is sampled from $\mathcal{E}$ versus the Haar ensemble are $\varepsilon$-close in trace distance.
\end{lemma}

\subsubsection{Measurable Error and Adaptive Security}

It is widely known in quantum complexity theory and quantum cryptography \cite{metger2024simple, ma2024construct} that indistinguishability under one query to $U^{\otimes k}$, i.e., $k$ parallel queries to $U$, does not imply indistinguishability under general quantum experiments that can query $k$ copies of $U$. A quantum experiment making $k$ queries to $U$ can apply the unitary in sequence and with arbitrary quantum operations in between each query. This is much more powerful as it enables the quantum experiment to learn some properties about $U$, then adaptively probe the unitary $U$ based on the properties it has learned. The general formulation of any quantum experiment that makes $k$ queries to $U$ is written as follows.
\begin{definition}[Quantum experiments with $k$ queries to $U$] \label{def: qu expt}
    A quantum experiment with $k$ queries consists of:
    \begin{enumerate}
        \item An initial state preparation $\ket{\psi_0}$ on registers $A \otimes A'$, where $A$ has dimension $2^n$ and $A'$ is an auxiliary register of dimension $2^m$. Without loss of generality, we can set $\ket{\psi_0} = \ket{0^n} \otimes \ket{0^m}$.
        \item For $i = 1, \ldots, k$: 
        \begin{itemize}
            \item Apply a unitary $W_i$ to registers $A \otimes A'$.
            \item Apply the unknown unitary $U$ to register $A$.
        \end{itemize}
        \item Apply a final unitary $W_{k+1}$ and measure to obtain classical outcome.
    \end{enumerate}
    %
\end{definition}
To capture indistinguishability under the most powerful quantum experiments that make up to $k$ queries to a unitary $U$, we consider a strong notion of approximation error: the maximum distinguishability between a random unitary ensemble and the Haar ensemble over all possible 
$k$-query quantum experiments.
Whenever this error is large, there exists a quantum experiment that can be constructed to distinguish the two ensembles. 
We therefore refer to this as the \emph{measurable error}. 
We now state the definition of measurable error formally.
\begin{definition}[Unitary $k$-design up to measurable error]
    Let $\varepsilon > 0$. An ensemble of unitaries $\mathcal{E}$ is an approximate unitary $k$-design up to measurable error $\varepsilon$ if for any quantum experiment with $k$ queries to $U$, the output states when $U$ is sampled from $\mathcal{E}$ versus the Haar ensemble are $\varepsilon$-close in trace distance, i.e.~
    \begin{equation}
        \sup_{W_1 \cdots W_{k+1}} \norm{\rho_\mathcal{E} - \rho_H}_1 \leq \varepsilon,
    \end{equation}
    where we have used the notation 
    \begin{equation}
        \rho_\mathcal{E} = \expect_{U \sim \mathcal{E}} \left[ W_{k+1} [U \otimes \mathbbm{1}_m] W_k \cdots W_2 [U \otimes \mathbbm{1}_m] W_1 |0^{n+m} \rangle\! \langle 0^{n+m} | W_1^\dagger [U^\dagger \otimes \mathbbm{1}_m] W_2^\dagger \cdots W_k^\dagger [U^\dagger \otimes \mathbbm{1}_m] W_{k+1}^\dagger 
        \right]
    \end{equation}
    to denote the expected output state of a general quantum experiment that queries $U$ $k$ times.
\end{definition}

\subsubsection{Relative Error and Output Probabilities}

One approach to bounding the distinguishability of two random unitary ensembles under general quantum experiments is to adopt an even more stringent definition of statistical closeness. For instance, we can consider the notion of approximate unitary designs based on relative error as defined in \cite{brandao2016local}:
\begin{definition}[Unitary $k$-design up to relative error]
    Let $\varepsilon > 0$. Then an ensemble of unitaries $\mathcal{E}$ is an approximate unitary $k$-design up to relative error $\varepsilon$ if
    \begin{equation}
    (1 - \varepsilon) \Phi_H \preceq \Phi_{\mathcal{E}} \preceq (1+\varepsilon) \Phi_H,
    \end{equation}
    where $\mathcal{A} \preceq \mathcal{B}$ denotes that $\mathcal{B}-\mathcal{A}$ is completely positive.
\end{definition}
The relative error can be thought of as a quantum analog of the ``almost'' $k$-wise independence condition for classical Boolean functions, which is sufficient to guarantee indistinguishability against a classical adversary \cite{aaronson2010bqp}. In \cite{schuster2024random}, it is shown that a bounded relative error is a sufficient condition for the indistinguishability of two ensembles under quantum experiments with access to $k$ copies of a random unitary sampled from either ensemble.
\begin{lemma}[Relative error implies measurable error \cite{schuster2024random}] \label{lemma:rel error to measurable}
    Suppose $\mathcal{E}$ is an approximate $k$-design up to relative error $\varepsilon$. Then for any quantum experiment with access to $k$ copies of $U \sim \mathcal{E}$, the output states are $2\varepsilon$-close in trace norm. 
\end{lemma}

Since measurable error is an upper bound on additive error, relative error serves as an upper bound for additive error as well (up to a constant factor 2).
However, the converse is not true. In general, converting from additive to relative error incurs a factor which is exponential in $n$ and $k$ \cite{brandao2016local,schuster2024random}:

\begin{lemma}[Comparing additive and relative error \cite{brandao2016local,schuster2024random}] \label{lemma:add to rel error}
    Let $\mathcal{E}$ be a unitary ensemble. If $\mathcal{E}$ is an approximate $k$-design up to relative error $\varepsilon$, then it is always an approximate design up to additive error $\varepsilon' = 2\varepsilon$. Conversely, if $\mathcal{E}$ is an approximate $k$-design up to additive error $\varepsilon$, it is guaranteed to be an approximate design up to relative error 
    \begin{equation} \label{eq: add to rel}
        \eps' = 2^{nk} {2^n+k-1 \choose k} \varepsilon \leq \left(\frac{4^{nk}}{k!} \right) \left( 1+\frac{k^2}{2^n} \right) \eps,
    \end{equation}
    where the first equality holds for any $k$, and the second inequality for any $k^2 \leq 2^n$.
\end{lemma}
\noindent The bound on the relative error in terms of the additive error [Eq.~(\ref{eq: add to rel})] follows from Lemma~30 in~\cite{brandao2016local}, by replacing the lossy inequality ${N+t-1 \choose t} \leq N^t$ in the proof of the lemma with either the exact expression, or the tighter inequality ${N+t-1 \choose t} \leq N^t(1+t^2/N)/t!$ for $t^2 \leq N$~\cite{schuster2024random}.
Here, $t \rightarrow k$ and $N \rightarrow 2^n$ in our notation.

While the relative error provides the strongest possible theoretical guarantee among the three error metrics, it does not correspond directly to physically observable quantities. Both the additive and measurable error have clear operational interpretations: the additive error quantifies the maximum distinguishability under parallel queries, while the measurable error captures the maximum distinguishability achievable by any quantum experiment. In contrast, the relative error characterizes the multiplicative deviation of the ensemble's moments from the Haar moments, which may not manifest as measurable differences in experiments. When the relative error significantly exceeds the measurable error for a given ensemble, the gap represents a theoretical structure that cannot be exploited by physical experiments. Rather, the relative error serves primarily as a mathematical tool for establishing theoretical properties rather than characterizing experimentally observable behavior.

An important exception is quantum experiments assisted by classical simulations.
Classical simulations can enable a quantum experiment to be informed by quantities, such as exponentially small probability amplitudes, that can never be efficiently measured in a quantum experiment alone.
This has enabled novel benchmarking and fidelity estimation strategies, such as classical shadow tomography~\cite{huang2020predicting}, as well as new theoretical understandings of thermalization and chaos in many-body quantum systems via projected ensembles~\cite{cotler2021emergent,choi2023preparing}. 
However, these classical-simulation-assisted quantum experiments require one to be able to exactly simulate the quantum system of interest on a classical computer, including features of the system that would not be detectable in any efficient quantum experiment. Hence, these classical simulations require exponential time in general, but can be polynomial time in some cases, such as when the quantum circuit is Clifford as in classical shadows~\cite{huang2020predicting}.

The relative error is uniquely suited to these quantum experiments assisted by classical simulations.
Most often, the usefulness of relative error stems from its ability to bound exponentially small probability amplitudes.
For example, one can bound moments of probabilities, such as
\begin{equation}
(1-\varepsilon) \E_{U\sim H}[ \bra{x} U \rho U^\dagger \ket{x}^{ k}] \leq \E_{U\sim\mathcal{E}}[ \bra{x} U \rho U^\dagger \ket{x}^{ k}] \leq (1+\varepsilon) \E_{U\sim H}[ \bra{x} U \rho U^\dagger \ket{x}^{ k}]    
\end{equation}
for any $\rho$ and $k$, in any approximate unitary $k$-design $\mathcal E$ with relative error $\varepsilon$.
Such bounds are not possible for unitary designs with additive or measurable error, since the probability amplitudes are exponentially small and not physically observable in any quantum experiment that queries $U$.

In Section~\ref{sec:relerrordesigns}, we demonstrate how repeated composition of designs with small additive error can yield designs with exponentially improved relative error guarantees. Specifically, we show that composing a specific unitary ensemble $p = 8k+1$ times achieves relative error. This provides a systematic approach for constructing relative error designs with only a moderate overhead in circuit depth.

\subsubsection{Example: Exact unitary 2-designs}

The multiqubit Clifford group provides a foundational example of an exact low-order design:
\begin{lemma}[Multiqubit Clifford group is an exact $3$-design \cite{zhu2017multiqubit}]
    The multiqubit Clifford group $\mathcal{C}_n$ on $n$ qubits forms an exact unitary $3$-design. That is, for any operator $X$ acting on $3n$ qubits,
    \begin{equation}
    \Phi^{(3)}_{\mathcal{C}_n}(X) = \Phi^{(3)}_H(X),
    \end{equation}x
    where $\Phi^{(3)}_{\mathcal{C}_n}(X) = \frac{1}{|\mathcal{C}_n|} \sum_{C \in \mathcal{C}_n} C^{\otimes 3} X (C^\dag)^{\otimes 3}$.
\end{lemma}
\noindent Any Clifford unitary can be compiled in circuit depth $\mathcal{O}(\log n)$ using $\mathcal{O}(n^2)$ ancilla qubits~\cite{moore2001parallel,jiang2020optimal}.

For the specific case of $k = 2$, even more resource-efficient constructions of exact designs are known. The ancilla count improves by a nearly quadratic factor compared to a random Clifford unitary.
\begin{fact}
    [Compilation of exact unitary 2-designs~\cite{cleve2015near}]
    Exact unitary 2-designs on $n$ qubits can be constructed in circuit depth $\mathcal{O}(\log n)$ using $\mathcal{O}(n \log n \log \log n)$ gates and $\tilde{\mathcal{O}}(n)$ ancilla qubits.
\end{fact}

It is not known how to construct exact designs for arbitrary $k$. However, there are several examples of ensembles which are known to approximate the moments of the Haar measure \cite{brandao2016local, haferkamp2022random, chen2024incompressibility, schuster2024random}.

\subsubsection{Example: PFC ensemble}

We review the Permutation-Function-Clifford (PFC) construction introduced in \cite{metger2024simple}, which adapts the state design construction by sandwiching a random phase between random permutation and Clifford operators:
\begin{definition}[PFC ensemble \cite{metger2024simple}] \label{def:PFC}
    Suppose $g$ and $f$ are drawn uniformly randomly from  permutations and binary functions on $\{0,1\}^n$, respectively, and $C$ is drawn uniformly from the Clifford group on $n$ qubits. Then the Permutation-Function-Clifford (PFC) ensemble is given by the family of $n$-qubit unitaries:
    \begin{equation}
        U = PFC = \left[ \sum_{x \in \{0,1\}^n} | g(x) \rangle \langle x| \right] \left[ \sum_{x' \in \{0,1\}^n} (-1)^{f(x')} |x'\rangle \langle x'| \right] C.
    \end{equation}
\end{definition}

In \cite{metger2024simple}, it was proven that the PFC ensemble forms an approximate unitary $k$-design up to additive error $\varepsilon = \mathcal{O}(\sqrt{k^2/2^n})$. However, it was not known whether the PFC ensemble forms an approximate design up to small measurable error. This open question was resolved in \cite{ma2024construct}:

\begin{theorem}[PFC are designs \cite{ma2024construct}] \label{thm:PFC-designs}
    The Permutation-Function-Clifford (PFC) ensemble is an approximate unitary $k$-design up to measurable error $\varepsilon = 4k^2/2^n$.
\end{theorem}
\noindent We provide a short alternate proof of Theorem~\ref{thm:PFC-designs} in Section~\ref{sec: pf PFC}.

\subsubsection{Example: LRFC ensemble}

We also review the Luby-Rackoff-Function-Clifford (LRFC) ensemble introduced in~\cite{SRU2025}, which is a permutation-free alternative to the PFC ensemble.
This alternative can be useful, since random functions are much easier to construct and analyze compared to random permutations.
\begin{definition}[LRFC ensemble \cite{SRU2025}] \label{def:LRFC}
    Suppose $h_L$ and $h_R$ are drawn uniformly randomly from functions on $\{0,1\}^{n/2} \rightarrow \{0,1\}^{n/2}$,  $f$ is drawn uniformly randomly from  binary functions on $\{0,1\}^n$, and $C$ is drawn uniformly from the Clifford group on $n$ qubits. 
    Then the Luby-Rackoff-Function-Clifford (LRFC) ensemble is given by the family of $n$-qubit unitaries:
    \begin{equation}
        U = S_L S_R FC \equiv 
        \left[ \sum_{x \in \{0,1\}^{n}} | x_L \oplus h_L(x_R) \lVert x_R \rangle \langle x_L \lVert x_R | \right] 
        \left[ \sum_{x' \in \{0,1\}^{n}} | x_L' \lVert x_R' \oplus h_R(x_L') \rangle \langle x_L' \lVert x_R' | \right]  F C,
    \end{equation}
    where $F$ acts as in Definition~\ref{def:PFC}, and  $\lVert$ denotes the concatenation of two $\frac{n}{2}$-bit strings, $x \equiv x_L \lVert x_R$ and $x' \equiv x'_L \lVert x'_R$.
\end{definition}
\noindent The $S_L$ and $S_R$ operations are random functions that shuffle the left $n/2$ qubits conditional on the value of the right $n/2$ qubits, and vice versa.
The composition $S_L S_R$ is a permutation of the $2^n$ bitstrings, which we can view as replacing the random permutation $P$ in the PFC ensemble.
The permutation $S_L S_R$ is not uniformly random, however.

In~\cite{SRU2025}, it was shown that the LRFC ensemble forms an approximate unitary $k$-design with small measurable error:\begin{theorem}[LRFC are designs \cite{SRU2025}] \label{thm:LRFC-designs}
    The Luby-Rackoff-Function-Clifford (LRFC) ensemble is an approximate unitary $k$-design up to measurable error $\varepsilon =6k^2/\sqrt{2^n}$.
\end{theorem}
\noindent This shows that the permutation $S_L S_R$ is indistinguishable from a random permutation $P$, when applied in composition with $F$ and $C$. 
We provide a short alternate proof of Theorem~\ref{thm:LRFC-designs} in Section~\ref{sec: pf LRFC}.

\subsection{$k$-wise independent random functions}

We now review the properties of classical $k$-wise independent functions. We are interested in families of functions which generate outputs that are close to those of uniformly random Boolean functions:
\begin{definition}[$k$-wise independent function]
    Let $f_\alpha: \{0,1\}^n \to \{0,1\}^n$ be a function on bitstrings of length $n$ which is parameterized by a random seed $\alpha \in \mathcal{A}$. Then $f_\alpha$ 
    is known as a $k$-wise independent function if for all choices of distinct input values $x_1, x_2, ..., x_k \in \{0,1\}^n$ and output values $y_1, y_2, ..., y_k \in \{0,1\}^n$, 
    \begin{equation}
    \Pr_{\alpha \sim \mathcal{A}} [f_\alpha(x_1) = y_1, f_\alpha(x_2) = y_2, ..., f_\alpha(x_k) = y_k] = 2^{-nk}.
    \end{equation}
\end{definition}
\noindent This definition is automatically satisfied by a truly random function, which sends each input bitstring to an independent and uniformly random output bitstring.

The standard construction of $k$-wise independent functions uses polynomial evaluation over finite fields. Let the set of inputs $\{0,1\}^n$ be identified with the finite field $\mathbb{F}_{2^n} = \mathrm{GF}(2^n)$. Each element $a \in \mathbb{F}_{2^n}$ can be represented as an $n$-bit vector $a_{n-1} a_{n-2} \ldots a_1 a_{0}$ associated with the polynomial $a(x) = \sum_{i=0}^{n-1} a_i x^i$ over $\mathbb{Z}_2$. The field operations are defined as follows: addition $a+b$ corresponds to bitwise XOR, while multiplication $a \cdot b$ is computed by multiplying the polynomials $a(x) \cdot b(x)$ modulo an irreducible polynomial $p(x)$ of degree $n$.

\begin{theorem}[Polynomial-based $k$-wise independent functions \cite{carter1977universal, alon1992simple}]
    Consider the function $f_\alpha: \mathbb{F}_{2^n} \to \mathbb{F}_{2^n}$ defined as follows. The seed $\alpha$ consists of $k$ elements $(a_0, a_1, \ldots, a_{k-1})$, where $\alpha_{i} \in \mathbb{F}_{2^n}$. The function $f_\alpha$ is defined by the evaluation of the polynomial whose coefficients are given by the seed:
    \begin{equation}
        f_\alpha(x) = \sum_{i=0}^{k-1} a_i x^i.
    \end{equation}
    When $a_0, \ldots, a_{k-1}$ are chosen uniformly at random, $f_\alpha$ is a $k$-wise independent function.
\end{theorem}

This construction has the desired $k$-wise independence property because any $k$ distinct inputs $x_1, \ldots, x_k$ form a Vandermonde system, and the corresponding outputs are determined by a degree-$(k-1)$ polynomial uniquely determined by its values at these points. Since the polynomial coefficients are chosen uniformly at random, the outputs appear uniformly random for any choice of $k$ inputs.

\subsubsection{Example: Derandomization of the binary phase state}

One can use $k$-wise independent functions to realize practical implementations of state and unitary designs, which avoid the need for truly random Boolean functions.
This can be achieved through bounded independence, by instead using $k$-wise independent functions:
\begin{definition}[$k$-wise independent phase state] \label{def:kwisephasestate}
    Let $f_\alpha: \{0,1\}^n \to \{0,1\}$ be a $(2k)$-wise independent function parameterized by $\alpha$. The corresponding $k$-wise independent phase state is
    \begin{equation}
        \ket{\phi_\alpha} = \frac{1}{2^{n/2}} \sum_{x \in \{0,1\}^n} (-1)^{f_\alpha(x)} \ket{x}.
    \end{equation}
\end{definition}
\begin{lemma}[$k$-wise independent phase states are state $k$-designs \cite{brakerski2019pseudo}] \label{lemma:kwisephasestate}
    For any $k$ and classical construction of a $(2k)$-wise independent function, the ensemble of $k$-wise independent binary phase states given by Def. \ref{def:kwisephasestate} yields an $\varepsilon$-approximate state $k$-design with $\varepsilon = 4k^2/2^n.$
\end{lemma}
\noindent An analogous replacement can be applied to the PFC circuit~\cite{metger2024simple} and LRFC circuit to generate unitary $k$-designs using $k$-wise independent functions and $k$-wise independent permutations.
We refer to~\cite{metger2024simple} for a discussion of $k$-wise independent permutations, as we do not use them in our work. In general, $k$-wise independent permutations are more difficult to implement than $k$-wise independent functions.

\section{Quantum circuit implementation of $k$-wise independent functions} \label{sec: compilation}

To implement $k$-wise independent functions efficiently in quantum circuits, we need to analyze the classical circuit complexity of polynomial evaluation over finite fields. We focus on implementations using classical reversible Boolean circuits with only constant-local Boolean gates and without the ability to perform unbounded fan-in or fan-out. We can directly implement a classical reversible Boolean circuit on a quantum computer via the following correspondence.

\vspace{0.75em}
\paragraph{Quantum gates.} A quantum circuit is inherently reversible as every quantum gate is a unitary transformation that preserves information. When implementing classical functions on quantum computers, we must use reversible classical circuits, which can be directly translated to quantum circuits by replacing each reversible Boolean gate, such as CNOT or Toffoli gates, to the associated quantum gates acting on computational basis states.

\vspace{0.75em}
\paragraph{Ancilla qubits.} Classical irreversible circuits can be made reversible through the use of ancillary bits to store intermediate results. On quantum devices, these ancillary bits become ancillary qubits. The number of classical ancillary bits directly determines the number of quantum ancillary qubits required.

\vspace{0.75em}
\paragraph{Circuit depth.} The depth of a reversible classical circuit is the depth of the associated quantum circuit.

\vspace{0.75em}
This section provides a detailed analysis of the reversible Boolean circuit resources required for such implementations, which directly translate to quantum circuit resources through the above correspondence.

\subsection{Finite field arithmetic}

The efficient implementation of $k$-wise independent functions relies on fast algorithms for arithmetic operations in finite fields. We begin by establishing the circuit complexity of basic field operations.

\subsubsection{Field addition}

Field addition in $\mathbb{F}_{2^n}$ has a remarkably simple implementation:

\begin{lemma}[Field Addition \cite{von2003modern}] \label{lemma:field_add}
Let $a, b \in \mathbb{F}_{2^n}$ represented by $n$ bits. The sum $a+b$ can be computed by a reversible classical Boolean circuit of depth $\mathcal{O}(1)$ with no ancillary bits. The field addition operation can be implemented either as $a, b \mapsto a+b, b$ or $a, b \mapsto a, a+b$ without any ancillary bits.
\end{lemma}
\begin{proof}
Addition in $\mathbb{F}_{2^n}$ corresponds to bitwise XOR of the binary representations. So we can implement $a+b$ using $n$ parallel XOR gates, each acting on corresponding bit positions. This uses $\mathcal{O}(1)$ depth and no ancillary bits.
\end{proof}

\subsubsection{Field multiplication}

Field multiplication is significantly more complex than addition, requiring polynomial multiplication followed by modular reduction. The state-of-the-art approach for large numbers uses the Schönhage-Strassen algorithm \cite{eberly1984very, schonhage1971schnelle, von2003modern}:

\begin{lemma}[Field Multiplication \cite{eberly1984very, schonhage1971schnelle, von2003modern}] \label{lemma:field_mult}
Let $a, b \in \mathbb{F}_{2^n}$ represented by $n$ bits. The product $a \cdot b$ can be computed by a reversible classical Boolean circuit of depth $\mathcal{O}(\log n)$ with $\mathcal{O}(n \log n \log \log n)$ ancillary bits. The field multiplication operation is implemented as $a, b, 0^n, 0^m \mapsto a, b, a \cdot b, 0^m$ with $m$ ancillary bits.
\end{lemma}

We first describe the standard irreversible Boolean circuit for performing field multiplication, then explain how to make the circuit reversible. The algorithm consists of two main phases:

\vspace{0.75em}
\paragraph{Phase 1: Polynomial multiplication.} We multiply $a(x), b(x)$ in $\mathbb{Z}_2[x]$ to obtain $c'(x) = a(x) \cdot b(x)$ with degree at most $(2n-2)$. This is achieved using a divide-and-conquer approach:
\begin{enumerate}
    \item \textbf{Recursive decomposition:} Rewrite the two degree-$(n-1)$ polynomials $a(x)$ and $b(x)$ over $\mathbb{F}_2$ as degree-$(n/\lfloor \sqrt{n} \rfloor)$ polynomials over the smaller field $\mathbb{F}_{2^{\lfloor \sqrt{n} \rfloor}}$.
    \item \textbf{Number-Theoretic Transform (NTT):} Apply the finite field analog of the Fast Fourier Transform to transform the polynomials into evaluation form. This step has circuit depth $\mathcal{O}(\log n)$ and size $\mathcal{O}(n \log n)$.
    \item \textbf{Pointwise multiplication:} Perform parallel pointwise multiplication over $\mathbb{F}_{2^{\lfloor \sqrt{n} \rfloor}}$. This step is implemented recursively, leading to the recurrence relations:
    \begin{align}
        D(n) &= D(\lfloor \sqrt{n} \rfloor) + \mathcal{O}(\log n), & \text{(circuit depth)}\\
        S(n) &= S(\lfloor \sqrt{n} \rfloor) + \mathcal{O}(n \log n),  & \text{(circuit size)}
    \end{align}
    where we have $D(n), S(n) = \mathcal{O}(1)$ for any small constant $n$.
    \item \textbf{Inverse NTT:} Transform back to the coefficient representation with the same complexity as the forward transform.
\end{enumerate}
Solving the recurrence relations yields total circuit depth $\mathcal{O}(\log n)$ and size $\mathcal{O}(n \log n \log \log n)$.

\vspace{0.75em}
\paragraph{Phase 2: Modular reduction.} We reduce $c'(x)$ modulo the irreducible polynomial $p(x)$ to obtain $c(x) = c'(x) \bmod p(x)$ with degree at most $n-1$. This can be implemented using Barrett reduction \cite{barrett1986implementing}. The setup is as follows. We work in the ring $\mathbb{F}_2[x]$ of polynomials with coefficients in $\mathbb{F}_2$. For polynomials $f(x), g(x) \in \mathbb{F}_2[x]$ with $g(x) \neq 0$, polynomial division gives us:
\begin{equation}
f(x) = q(x) \cdot g(x) + r(x)
\end{equation}
where $q(x)$ is the quotient and $r(x)$ is the remainder with $\deg(r) < \deg(g)$. We define $\lfloor f(x) / g(x) \rfloor = q(x)$ as the quotient polynomial, ignoring the remainder. Barrett reduction uses the following identity:
\begin{theorem}[Barrett identity for polynomials]
For $f, g, h \in \mathbb{F}_2[x]$ with $g, h \neq 0$, the following identity holds:
\begin{equation}
\lfloor f / g \rfloor = \lfloor (f \cdot \lfloor h / g \rfloor) / h \rfloor
\end{equation}
provided that $\deg(f) \leq \deg(h)$.
\end{theorem}
\begin{proof}
Let $f = q \cdot g + r$ and $h = q' \cdot g + r'$ where $\deg(r) < \deg(g)$ and $\deg(r') < \deg(g)$. Then $\lfloor f / g \rfloor = q$ and $\lfloor h / g \rfloor = q'$. We have:
\begin{equation}
    f \cdot \lfloor h / g \rfloor = f \cdot q' = q \cdot g \cdot q' + r \cdot q' = h \cdot q + (r \cdot q' - r' \cdot q).
\end{equation}
We can bound the degree of the error term $r \cdot q' - r' \cdot q$ as follows. Since $h = q' \cdot g + r'$, we have $\deg(q') = \deg(h) - \deg(g)$. Similarly, $\deg(q) = \deg(f) - \deg(g)$. Therefore:
\begin{align}
\deg(r \cdot q') &\leq \deg(r) + \deg(q') \leq \deg(q') + \deg(g) - 1 < \deg(h),\\
\deg(r' \cdot q) &\leq \deg(r') + \deg(q) \leq \deg(q) + \deg(g) - 1 < \deg(f) \leq \deg(h).
\end{align}
Hence, $\deg(r \cdot q' - r' \cdot q) < \deg(h)$. This implies that:
\begin{equation}
    \lfloor f \cdot \lfloor h / g \rfloor / h \rfloor = \left\lfloor \frac{h \cdot q + (r \cdot q' - r' \cdot q)}{h} \right\rfloor = q,
\end{equation}
where the last equality follows because $\deg(r \cdot q' - r' \cdot q) < \deg(h)$. Since $\lfloor f / g \rfloor = q$, this concludes the proof.
\end{proof}
Barrett reduction computes $c'(x) \bmod p(x)$ using the above identity with $f = c'(x)$, $g = p(x)$, and $h = x^{2n-2}$. Since $\deg(c'(x)) \leq 2n-2 = \deg(x^{2n-2})$, the condition is satisfied. The modular reduction becomes:
\begin{equation}
    c'(x) \bmod p(x) = c'(x) - \lfloor c'(x) / p(x) \rfloor \cdot p(x) = c'(x) - \lfloor (c'(x) \cdot \lfloor x^{2n-2} / p(x) \rfloor) / x^{2n-2} \rfloor \cdot p(x).
\end{equation}
The computation proceeds as follows:
\begin{enumerate}
    \item \textbf{Precomputation:} For the fixed polynomial $p(x)$ of degree $n$, precompute the Barrett constant:
    \begin{equation}
    \mu(x) = \lfloor x^{2n-2} / p(x) \rfloor
    \end{equation}
    This polynomial has degree $(2n-2) - n = n-2$. Since $p(x)$ is fixed, $\mu(x)$ can be precomputed once (on a classical computer) and reused.
    
    \item \textbf{Polynomial multiplication:} Perform polynomial multiplication $c'(x) \cdot \mu(x)$ to obtain the polynomial $e(x)$ of degree at most $(2n-2) + (n-2) = 3n-4$. This multiplication can be implemented using the Sch\"onhage-Strassen algorithm with circuit depth $\mathcal{O}(\log n)$ and size $\mathcal{O}(n \log n \log \log n)$.
    
    \item \textbf{Quotient extraction:} Compute the quotient:
    \begin{equation}
    q(x) = \lfloor e(x) / x^{2n-2} \rfloor
    \end{equation}
    By the Barrett identity, $q(x)$ equals the true quotient $\lfloor c'(x) / p(x) \rfloor$ and has degree at most $(3n-4) - (2n-2) = n-2$. This operation extracts the coefficients of $x^{2n-2}$ and higher monomials from $e(x)$, which can be implemented by right-shifting the $(3n-3)$ bits describing $e(x)$ by $(2n-2)$ positions.
    
    \item \textbf{Remainder computation:} Compute the remainder:
    \begin{equation}
    r(x) = c'(x) - q(x) \cdot p(x) = c'(x) + q(x) \cdot p(x),
    \end{equation}
    where the second equality holds in $\mathbb{F}_2[x]$. This requires another polynomial multiplication $q(x) \cdot p(x)$, which uses a circuit depth of $\mathcal{O}(\log n)$ and a size of $\mathcal{O}(n \log n \log \log n)$.
\end{enumerate}
Since the Barrett identity guarantees that $q(x)$ is the exact quotient, the remainder $r(x)$ automatically has degree at most $n-1$ and equals $c'(x) \bmod p(x)$. The Barrett reduction requires two polynomial multiplications (steps 2 and 4), both implementable using the Sch\"onhage-Strassen algorithm. The total circuit depth for Barrett reduction is $\mathcal{O}(\log n)$ and the total size is $\mathcal{O}(n \log n \log \log n)$. The field multiplication (polynomial multiplication + Barrett reduction) requires circuit depth $\mathcal{O}(\log n)$ and size $\mathcal{O}(n \log n \log \log n)$ using a standard Boolean circuit.

\vspace{0.75em}
\paragraph{Conversion to reversible circuits.} The above describes a standard Boolean circuit, which is irreversible in general. To obtain a reversible Boolean circuit, we apply the standard conversion procedure:
\begin{enumerate}
    \item Create an ancillary bit initialized to $0$ for every output bit of each gate.
    \item Implement the irreversible Boolean circuit by replacing each Boolean gate with a reversible gate that acts on the associated ancillary bit.
    \item Use bitwise XOR to copy the $n$ ancillary bits associated with the answer $a \cdot b$ to the final $n$ output bits.
    \item Run all reversible gates in reverse to return all ancillary bits to $0$.
\end{enumerate}
Hence, the circuit size of a standard Boolean circuit upper bounds the number of ancillary bits of the reversible Boolean circuit, while the circuit depth only increases by a constant factor.
Together, the product $a \cdot b$ can be computed by a reversible classical Boolean circuit of depth $\mathcal{O}(\log n)$ with $\mathcal{O}(n \log n \log \log n)$ ancillary bits.

\subsection{Polynomial evaluation circuit}

With the field arithmetic primitives established, we can now analyze the evaluation of $k$-wise independent functions:

\begin{lemma}[Low-depth and low-ancilla $k$-wise independent functions] \label{lemma:kwisefns}
A classical reversible Boolean circuit for evaluating $f_\alpha(x) = \sum_{i=0}^{k-1} a_i x^i$ for a given seed $\alpha = (a_0, \ldots, a_{k-1})$ and input $x \in \mathbb{F}_{2^n}$ can be implemented with either:
\begin{enumerate}
\item \textbf{Low depth:} $\mathcal{O}(\log k \cdot \log n)$ depth and $\mathcal{O}(k n \log n \cdot \log \log n)$ ancillary bits.
\item \textbf{Low ancilla:} $\mathcal{O}(k \cdot \log n)$ depth and $\mathcal{O}(n \log n \cdot \log \log n)$ ancillary bits.
\end{enumerate}
\end{lemma}

\noindent We present both low-depth and low-ancilla constructions below.

\vspace{1em}
\noindent \textbf{Low-depth construction:}
\vspace{0.75em}

The computation  proceeds in three stages as follows.

\subsubsection*{Stage 1: Computing powers of $x$}

The first stage computes all powers $x^0, x^1, \ldots, x^{k-1}$ using an approach that minimizes circuit depth:

\vspace{0.75em}
\paragraph{Phase 1a (Sequential power computation):} Compute exponents that are powers of 2 using repeated squaring:
\begin{align}
    x^1 &= x \\
    x^2 &= x^1 \cdot x^1 \\
    x^4 &= x^2 \cdot x^2 \\
    &\vdots \\
    x^{2^{\lfloor \log_2(k-1) \rfloor}} &= \left(x^{2^{\lfloor \log_2(k-1) \rfloor - 1}}\right)^2
\end{align}
This requires $\lfloor \log_2(k-1) \rfloor = \mathcal{O}(\log k)$ field multiplications performed sequentially. The powers are stored in $\lfloor \log_2(k-1) \rfloor \cdot n$ bits. Since each multiplication has depth $\mathcal{O}(\log n)$ by Lemma~\ref{lemma:field_mult}, the total circuit depth is $\mathcal{O}(\log k \cdot \log n)$. By reusing ancillary bits across sequential field multiplications, we need only $\mathcal{O}(n \log n \log \log n)$ ancillary bits.

\vspace{0.75em}
\paragraph{Phase 1b (Parallel power computation):} Use a binary tree to copy the power-of-2 exponents for later use. This requires depth $\mathcal{O}(\log k)$ and $\mathcal{O}(nk \log k)$ ancillary bits. For each remaining exponent $j \in \{3, 5, 6, 7, 9, \ldots, k-1\}$, we express $j$ in binary as $j = \sum_{i} b_i 2^i$ and compute:
\begin{equation}
x^j = \prod_{i: b_i = 1} x^{2^i}.
\end{equation}
All exponents $j$ are computed in parallel. Each product requires at most $\mathcal{O}(\log k)$ field multiplications. All powers of $x$, $x^0, x^1, \ldots, x^{k-1}$, are stored in classical memory, which requires $k \cdot n$ bits. The total circuit depth is $\mathcal{O}(\log k \cdot \log n)$, and the number of ancillary bits is $\mathcal{O}(k \cdot n \log n \log \log n)$.

\subsubsection*{Stage 2: Coefficient multiplication}

Compute $a_i \cdot x^i$ for all $i \in \{0, 1, \ldots, k-1\}$ in parallel. Each of the $k$ parallel multiplications requires the use of $\mathcal{O}(n \log n \log \log n)$ ancillary bits by Lemma~\ref{lemma:field_mult}, for a total of $\mathcal{O}(k \cdot n \log n \log \log n)$ ancillary bits. Since all multiplications are performed in parallel, the circuit depth is $\mathcal{O}(\log n)$.

\subsubsection*{Stage 3: Binary tree summation}

Sum the $k$ terms $a_i \cdot x^i$ using a binary tree with $\lceil \log_2 k \rceil$ levels of parallel additions:
\begin{align}
    \text{Level 1:} &\quad \lceil k/2 \rceil \text{ pairwise sums in parallel} \\
    \text{Level 2:} &\quad \lceil k/4 \rceil \text{ pairwise sums in parallel} \\
    &\vdots \\
    \text{Final level:} &\quad \text{Single sum}
\end{align}
By Lemma~\ref{lemma:field_add}, each addition requires constant depth and no ancillary bits. The total circuit depth is $\mathcal{O}(\log k)$.

\subsubsection*{Final reversible implementation}

We use bitwise XOR to copy the $n$ ancillary bits associated with the answer $f_{\alpha}(x)$ to the final $n$ output bits. By running the entire circuit in reverse, we return all ancillary bits to $0$. This reversible Boolean circuit implements:
\begin{equation}
\alpha, x, 0^n, 0^m \mapsto \alpha, x, f_{\alpha}(x), 0^m,
\end{equation}
where $m = \mathcal{O}(k n \log n \log \log n)$ is the total number of ancillary bits required.

\vspace{1em}
\noindent \textbf{Low-ancilla construction:}
\vspace{0.75em}

The low-ancilla approach trades depth for space by computing the function $f_\alpha(x) = \sum_{i=0}^{k-1} a_i x^i$ sequentially using the following algorithm:
\begin{align}
    &\texttt{result} \leftarrow 0 \\
    &\texttt{current\_power} \leftarrow 1 \quad \text{(representing } x^0 \text{)} \\
    &\texttt{for } i = 0 \texttt{ to } k -1: \\
    &\quad\quad\quad \texttt{result} \leftarrow \texttt{result} + a_i \cdot \texttt{current\_power} \\
    &\quad\quad\quad \texttt{current\_power} \leftarrow \texttt{current\_power} \cdot x
\end{align}
Each iteration requires one field multiplication ($a_i \cdot \texttt{current\_power}$), one field addition (\texttt{result} $+$ $a_i \cdot \texttt{current\_power}$), and one field multiplication (\texttt{current\_power} $\cdot x$). Each field arithmetic operation has depth $\mathcal{O}(\log n)$ by Lemma \ref{lemma:field_add} and Lemma \ref{lemma:field_mult}. With $k$ sequential iterations, the total depth is $\mathcal{O}(k \log n)$.

The space complexity includes storage for \texttt{result}, \texttt{current\_power}, and temporary ancillary bits for field operations. Since we reuse ancillary bits across sequential operations, the total requirement is $\mathcal{O}(n \log n \log \log n)$ ancillary bits. Similar to the low-depth construction, we use the bitwise XOR to copy the $n$ bits associated with the answer $f_{\alpha}(x)$ to the final $n$ output bits. By running the entire circuit in reverse, we return all ancillary bits to $0$. This reversible Boolean circuit implements:
\begin{equation}
\alpha, x, 0^n, 0^m \mapsto \alpha, x, f_{\alpha}(x), 0^m,
\end{equation}
where $m$ is the total number of ancillary bits required for the chosen construction.

\subsection{Summary of circuit resources}

The complete quantum circuit implementation of a $k$-wise independent $n$-bit function uses either:
\begin{itemize}
    \item \textbf{Low depth:} $\mathcal{O}(\log k \cdot \log n)$ depth and $\mathcal{O}(k n \log n \log \log n)$ ancillary qubits.
    \item \textbf{Low ancilla:} $\mathcal{O}(k \cdot \log n)$ depth and $\mathcal{O}(n \log n \log \log n)$ ancillary qubits.
\end{itemize}
These resource requirements provide flexibility for different quantum computing architectures, enabling efficient implementations of $k$-wise independent designs based on available resources.

\section{State designs in nearly optimal depth} \label{sec:state designs}

In this section, we provide details of our construction of approximate state designs.
We first introduce our design construction based on blocked random phase states.
We then discuss the key ingredient of our proof, the local distinct subspace projector, and provide full proof details.
We discuss the application of our results to a simple construction of $\text{poly} \log \log n$-depth pseudorandom states in the final section.


\subsection{Blocked random phase states}

As discussed in the main text (Fig.~1), we consider a blocked variant of the random binary phase state,
\begin{equation} \label{eq: low depth binary phase state}
        \ket{\psi} = \frac{1}{\sqrt{2^{n}}} \sum_{x \in \{0,1\}^n} (-1)^{\sum_a f_{a,a+1}(x_{a,a+1})} \ket{x},
\end{equation}
in which the phase $f(x) \equiv \sum_a f_{a,a+1} (x)$ is computed via a two-layer circuit of random classical functions.
Here, we make a small change in notation compared to the main text, by replacing the index $i$ for each patch of $2\xi$ qubits with the index $a,a+1$ (which labels the two patches of $\xi$ qubits each, $a$ and $a+1$, that combine to form $i$).
This makes it explicit that adjacent functions overlap on patches of $\xi$ qubits each.
Each function $f_{a,a+1}$ acts on $2\xi$ qubits, and only depends on the values $x_{a,a+1}$ of the bitstring $x$ on those qubits.
The value of $\xi$ is tunable and will be set depending on the desired $n,k,\varepsilon$ of the state design.

To construct approximate state $k$-designs, we draw each $f_{a,a+1}$ from a $2k$-wise independent random function ensemble on $2\xi$ bits.
We require $2k$-wise independence instead of $k$-wise independence to account for both the ``ket'' and ``bra'' of the quantum state.
We will refer to the resulting ensemble, $k$-wise independent blocked random phase state.
We prove that they form state designs:
\begin{theorem}[Blocked random phase states are designs] \label{thm:blocked-phase-design}
    The $k$-wise independent blocked random phase state ensemble with parameter $\xi$ is an approximate state $k$-design up to additive error $\varepsilon = 3nk^2/2^\xi \xi + 2k^2/2^n = \mathcal{O}(nk^2/2^\xi \xi)$.
\end{theorem}

The circuit depth of the resulting state ensemble is equal to twice the circuit depth of a $k$-wise independent function on $2\xi$ qubits. From Lemma~\ref{lemma:kwisefns} of Section~\ref{sec: compilation}, we have:
\begin{fact}[Circuit resources for blocked random phase states] \label{fact:blocked-phase-depth}
    The $n$-qubit $k$-wise independent blocked random phase state with parameter $\xi$ can be implemented with either:
    \begin{enumerate}
        \item \textbf{Low depth:} $\mathcal{O}(\log k \cdot \log \xi)$ depth using $\mathcal{O}(k n \log \xi \cdot \log \log \xi)$ ancilla qubits.
        \item \textbf{Low ancilla:} $\mathcal{O}(k \cdot \log \xi)$ depth using $\mathcal{O}(n \log \xi \cdot \log \log \xi)$ ancilla qubits.
    \end{enumerate}
    Both implementations use $2nk$ bits of randomness.
\end{fact}
\noindent The number of bits of randomness follows because each $k$-wise independent function uses $\xi k$ bits of randomness~\cite{wegman1981new}.
The optimal number of bits of randomness for quantum state $k$-designs is $nk$ (up to sub-leading additive factors)~\cite{gross2007evenly,roy2009unitary}; a mere factor of two fewer than our ensemble.
Intuitively, this lower bound follows because the Haar twirl $\chi_H$ has rank $2^{nk}$ (up to sub-leading multiplicative factors, for $k \ll 2^n$).

Combining Fact~\ref{fact:blocked-phase-depth} and Theorem~\ref{thm:blocked-phase-design} yields the following:

\begin{corollary}[Circuit resources for state $k$-designs] \label{cor:state-design-resources}
    For any $k^2 / \varepsilon \leq 2^n / 3$, setting $\xi = \max(\log_2(3 n k^2 / \varepsilon), 3) \leq n$ yields approximate state $k$-designs with additive error $\varepsilon$ using either:
    \begin{itemize}
        \item \textbf{Low depth:} $\mathcal{O}(\log k \cdot \log \log(nk/\varepsilon))$ depth and $\mathcal{O}(k n \log \log(nk/\varepsilon) \log \log \log(nk/\varepsilon))$ ancilla qubits.
        \item \textbf{Low ancilla:} $\mathcal{O}(k \cdot \log \log(nk/\varepsilon))$ depth and $\mathcal{O}(n \log \log(nk/\varepsilon) \log \log \log(nk/\varepsilon))$ ancilla qubits.
    \end{itemize}
    Both use $2nk$ bits of randomness, which is merely a factor of two above the optimal $nk$ bits for $n$-qubit state $k$-designs.
\end{corollary}

\subsection{The local distinct subspace}

Before proceeding to our proof of Theorem~\ref{thm:blocked-phase-design}, let us first introduce a key technical object in our analysis: the local distinct subspace. We begin with the standard (non-local) distinct subspace~\cite{metger2024simple}, and then turn to its local variant.

%
%
Let $x = (x^{(1)},x^{(2)},\ldots,x^{(k)})$ be a length-$k$ list of $n$-bit strings.
The list labels a computational basis state $\ket{x} \equiv \bigotimes_{j=1}^k \ket{x^{(j)}}$ on the $k$-copy Hilbert space.
The \emph{distinct subspace} is spanned by $x$ such that no two $x^{(j)}$ are equal:
\begin{definition}[Distinct subspace~\cite{metger2024simple}]
    Let $\dist = \{x \, | \, x^{(i)} \neq x^{(j)} \, \mathrm{for} \, i \neq j \}$ denote the set of distinct $x$.
    The \emph{distinct subspace} is defined by the projector
    \vspace{-2mm}\begin{equation}
        \Pi_{\text{\emph{dist}}} = \sum_{x \in \text{\emph{dist}}} \dyad{x}.
    \end{equation}
    The subspace has dimension $\mathfrak{D} = (2^n)!/(2^n-k)!$ which obeys $1-k^2/2^n \leq \mathfrak{D}/2^{nk} \leq 1$.
\end{definition}
\noindent 
A crucial property of the distinct subspace is that \emph{most} bitstrings are distinct when $k^2 \ll 2^n$.
This is formalized by the lower bound on the dimension of the distinct subspace $\mathfrak{D}$ stated in the definition. 
The lower bound shows that the ratio of the number of distinct bitstrings, $\mathfrak{D}$, to the number of total bitstrings, $2^{nk}$, is  close to one when $k^2 \ll 2^n$. 

To analyze our blocked state and unitary ensembles, we  introduce a \emph{local} variant of the distinct subspace.
Let $x^{(j)}_a$ denote the $\xi$-bit string obtained by restricting an $n$-bit string $x^{(j)}$ to the $\xi$ qubits on patch $a$.
We let $x_a = (x^{(1)}_a,x^{(2)}_a,\ldots,x^{(k)}_a)$ denote the corresponding length-$k$ list of $\xi$-bit strings.
The \emph{local distinct subspace} is spanned by $x$ such that $x_a$ is distinct for every patch $a$.
%
%
%
\begin{definition}[Local distinct subspace]
    For any system of $n$ qubits divided into patches of $\xi$ qubits each.
    Let $\text{\emph{loc-dist}} = \{ x \,|\, x_a^{(i)} \neq x_a^{(j)} \text{\emph{ for all }} a \text{\emph{ and }} i \neq j \}$ denote the set of locally distinct $x$.
    The \emph{local distinct subspace} is defined by the projector
    \vspace{-2mm}
    \begin{equation}
        \Pi_{\text{\emph{dist}}}^{\text{\emph{loc}}} \equiv  \sum_{x \in \text{\emph{loc-dist}}} \dyad{x}.  \vspace{-1mm}       
    \end{equation}
    The local distinct subspace has dimension $\mathfrak{D}_{\text{\emph{loc}}} = ((2^\xi)!/(2^\xi-k)!)^{n/\xi}$ which obeys $1-nk^2/2^\xi \xi \leq \mathfrak{D}_{\text{\emph{loc}}} \leq 1$.
\end{definition}
\noindent The local distinct subspace is contained within the global distinct subspace, $\PD^{\text{loc}} = \PD^{\text{loc}} \PD$.

One might worry that requiring each bitstring to be locally distinct imposes too strong of a restriction, and that most bitstrings will lie outside the local distinct subspace.
Fortunately, this is not the case, as expressed in the lower bound on $\mathfrak{D}_{\text{loc}}$ stated in the definition.
To derive the formula for $\mathfrak{D}_{\text{loc}}$, we note that the local distinct subspace projector is a tensor product of projectors onto the distinct subspace of each patch $a$, $\Pi_{\text{dist}}^{\text{loc}} = \bigotimes_a \Pi_{\text{dist}}^a$, where $\PD^a \equiv \sum_{x_a \in \text{dist}_a} \dyad{x_a}$.
The subspace dimension is therefore given by $\mathfrak{D}_{\text{loc}} = \prod_{a=1}^m \mathfrak{D}_a$ where $m \equiv n/\xi$ is the total number of patches and $\mathfrak{D}_a = (2^\xi)!/(2^\xi-k)!$ is the distinct subspace dimension of each patch.
From the bound
\begin{equation} \label{eq:dimension-bound-Da}
1-k^2/2^\xi \leq \mathfrak{D}_a \leq 1,
\end{equation}
we have $1-m k^2/2^\xi \leq \mathfrak{D}_{\text{loc}} \leq 1$.
Hence, most bitstrings are locally distinct when $(n/\xi)k^2 < nk^2 \ll 2^\xi$, which only requires $\xi$ be logarithmic in $n$ and $k$.

\subsection{Proof of Theorem~\ref{thm:blocked-phase-design}: The blocked random phase state has small additive error}

Let us now turn to the proof of Theorem~\ref{thm:blocked-phase-design} and show that the blocked random phase state is a state $k$-design.
Our proof consists of only a small number of very straightforward steps, as discussed in the main text.
We describe each step in detail to facilitate understanding.

We consider three state ensembles.
Let 
\begin{equation}
    \chi_B \equiv \E_{F_{a,a+1}}\Big[ \Big( \prod_a F_{a,a+1} \Big)^{\otimes k} \cdot \dyad{+^n}^{\otimes k} \cdot \Big( \prod_a F_{a,a+1} \Big)^{\otimes k} \Big]
\end{equation}
denote the twirl of the blocked random phase state,
\begin{equation}
    \chi_F \equiv \E_{F} \Big[ \, F^{\otimes k} \cdot \dyad{+^n}^{\otimes k}  \cdot F^{\otimes k}\Big]
\end{equation}
denote the twirl of the (global) random phase state, and 
\begin{equation} \label{eq: chiH}
    \chi_H \equiv \E_U \Big[ U^{\otimes k} \cdot \dyad{+^n}^{\otimes k} \cdot (U^\dagger)^{\otimes k} \Big] = \frac{(2^n -1)!}{(2^n+k-1)!} \sum_{\pi \in S_k} \pi,
\end{equation}
denote the Haar twirl.
Our consideration of the random phase state $\chi_F$ is purely for pedagogical sake.
In the formula for $\chi_H$, we apply a standard equation for the Haar twirl\footnote{One can also write $\chi_H = P_{\text{sym}}/D_{\text{sym}}$ where $P_{\text{sym}} = \frac{1}{k!}\sum_\pi \pi$ is the projector onto the symmetric subspace and $D_{\text{sym}} = {2^n+k-1 \choose 2^n-1}$ is the symmetric subspace dimension~\cite{harrow2013church}.}.

The twirls of the blocked random phase state and global random phase state are easy to compute.
Let us decompose $\dyad{+^n}^{\otimes k} = \frac{1}{2^{nk}} \sum_{x,\tilde{x}} \dyad{\tilde x}{x}$ in the computational basis on $k$ copies.
When averaged, the global random phase $F$ enforces that $\tilde x$ and $x$ are \emph{stabilizations} of one another~\cite{brakerski2019pseudo},
\begin{equation}\nonumber
    \chi_F = \E_{f} \bigg[ \frac{1}{2^{nk}} \sum_{x,\tilde x} (-1)^{\sum_j f( x^{(j)}) + \sum_j f(\tilde x^{(j)})} \dyad{\tilde x}{x} \bigg] = \frac{1}{2^{nk}} \sum_{x,\tilde x} \delta_{(\tilde x, x \text{ are stabilizations})} \cdot \dyad{\tilde x}{x}.
\end{equation}
Here, $\tilde x$ and $x$ are stabilizations if every bitstring that appears in $\tilde x$ an odd number of times appears in $x$ an odd number of times, and vice versa.
This guarantees that the phase vanishes, $\sum_j f(\tilde x^{(j)}) + \sum_j f( x^{(j)}) = 0$, for every random binary function $f$.
If this condition does not hold, the average over $f$ will cause the  term to vanish.
The twirl of the blocked random phase state is similarly easy to compute, albeit with several more indices,
\begin{equation}\nonumber
    \chi_B = \!\!\E_{f_{a,a+1}} \!\bigg[ \frac{1}{2^{nk}} \sum_{x,\tilde x} (-1)^{\sum_{j,a} f^{}_{a,a+1}( x^{(j)}_{a,a+1}) + \sum_{j,a} f^{}_{a,a+1}(\tilde x^{(j)}_{a,a+1})} \dyad{\tilde x}{x} \! \bigg] \! =\! \frac{1}{2^{nk}} \sum_{x,\tilde x} \Big( \prod_a \delta_{(\tilde x_{a,a+1}, x_{a,a+1} \text{ are stabilizations})} \Big) \dyad{\tilde x}{x}.
\end{equation}
The total phase is a sum of the local phases from each patch $a,a+1$ of each copy $j$.
For the term to not vanish after averaging over every $f_{a,a+1}$, we require that $\tilde x_{a,a+1}$ and $x_{a,a+1}$ are stabilizations for each patch $a$. 
%
%

With formulas for all three states in hand, we can now  project onto the local distinct subspace.
Note that all three states commute with the local distinct subspace projector.
Hence, we have
\begin{equation}
    \left\lVert \chi_B - \PD^{\text{loc}} \, \chi_B \, \PD^{\text{loc}} \right\rVert_1 \leq 1 - \tr(\PD^{\text{loc}} \, \chi_B)
\end{equation}
and similar for $\chi_F$ and $\chi_H$.
Since $\chi_B$ and $\chi_F$ have uniform support over the computational basis states, we have
\begin{equation} \label{eq: norm B}
    1-\tr(\PD^{\text{loc}} \, \chi_B) = 1-\tr(\PD^{\text{loc}} \, \chi_F) =
    1 - \mathfrak{D}_{\text{loc}}/2^{nk} \leq nk^2/2^\xi \xi.
\end{equation}
This implies that the projection to the local distinct subspace incurs only a small trace-norm error, as discussed in the main text.
For the Haar-random state $\chi_H$, we have
\begin{equation} \label{eq: norm H}
    1-\tr(\PD^{\text{loc}} \, \chi_H) = 1-\tr(\PD^{\text{loc}}) \frac{(2^n-1)!}{(2^n+k-1)!} =  1 - \mathfrak{D}_{\text{loc}} \frac{(2^n-1)!}{(2^n+k-1)!} \leq nk^2/2^\xi \xi + k^2/2^n.
\end{equation}
This follows immediately from Eq.~(\ref{eq: chiH}) because the trace of all $\pi \neq \mathbbm{1}$ with the distinct projector are zero.
The inequality holds because $\mathfrak{D}_{\text{loc}}$ is only slightly less than $2^{nk}$ and the ratio of factorials is only slightly less than $1/2^{nk}$ (for $k^2 \leq 2^n$).
Hence, the Haar-random state can also be projected to the local distinct subspace with only a small trace-norm error.

To complete the proof, we just observe that the three states are equal on the local distinct subspace.
Consider $\chi_B$ and $\chi_F$.
For $\chi_F$, the restriction to the local distinct subspace enforces $\tilde x = \pi x$ for some permutation $\pi \in S_k$.
This follows because the stabilization of a distinct bitstring must be a permutation of the bitstring.
%
%
%
%
For $\chi_B$, we similarly have $\tilde x_{a,a+1} = \pi_a x_{a,a+1}$ on each patch $a$, for $\pi_a \in S_k$.
This is equivalent to $\tilde x_a = \pi_a x_a$ and $\tilde x_{a+1} = \pi_a x_{a+1}$, if we consider the support of $x$ on patches $a$ and $a+1$ individually.
However, since $x_a$ is distinct, there is a unique $\pi_a$ that can satisfy each condition.
This implies that the first permutation $\pi_1$ (which obeys $\tilde x_1 = \pi_1 x_1$ and $\tilde x_2 = \pi_1 x_2$) must equal the second permutation $\pi_2$ (which obeys $\tilde x_2 = \pi_2 x_2$ and $\tilde x_3 = \pi_2 x_3$).
Iterating $m = n/\xi$ times, one finds that \emph{all} permutations are in fact equal, $\pi_a = \pi_1 \equiv \pi$.
This produces an identical state as  $\chi_F$, 
\begin{equation}
    \PD^{\text{loc}} \, \chi_B \, \PD^{\text{loc}} = \frac{1}{2^{nk}} \sum_{\pi \in S_k} \sum_{x \in \text{loc-dist}}  \dyad{\pi x}{x} = \PD^{\text{loc}} \, \chi_F \, \PD^{\text{loc}}.
\end{equation}
Here, we replace the sum over $\tilde x$ with a sum over $\pi$, since for every distinct $x$, each $\pi \in S_k$ contributes a unique $\tilde x = \pi x$.
%
%
Turning to $\chi_H$, we insert the definition of $\PD^{\text{loc}}$ into the formula for $\chi_H$ [Eq.~(\ref{eq: chiH})], which gives
\begin{equation}
    \PD^{\text{loc}} \, \chi_H \, \PD^{\text{loc}} = \frac{(2^n-1)!}{(2^n+k-1)!} \sum_{\pi \in S_k} \sum_{x \in \text{loc-dist}}  \dyad{\pi x}{x}.
\end{equation}
As promised, this is equal to the projected $\chi_B$ and $\chi_F$ up to only a small difference in normalization.
Observing Eq.~(\ref{eq: norm B}) and Eq.~(\ref{eq: norm H}), the difference in normalization can be at most $nk^2/2^\xi \xi + k^2 /2^n$. This implies that the trace-norm error is bounded by the same amount, $\lVert \PD^{\text{loc}} \, \chi_B \, \PD^{\text{loc}} -\PD^{\text{loc}} \, \chi_H \, \PD^{\text{loc}} \rVert_1 \leq nk^2/2^\xi \xi + k^2 2^n$.
Combining with Eq.~(\ref{eq: norm B}) and Eq.~(\ref{eq: norm H}), we have $\lVert \chi_B - \chi_H \rVert_1 \leq 3nk^2/2^\xi \xi + 2k^2/2^n = \mathcal{O}(nk^2/2^\xi \xi)$ as desired. \qed

\subsection{Second approach to nearly optimal state designs via state-function gluing lemma} \label{sec:state-design-gluing}

In this section, we provide a second, alternative approach to proving Theorem~\ref{thm:blocked-phase-design}.
Instead of analyzing the entire blocked random phase circuit at once, our alternative approach analyzes only three adjacent blocks of the circuit at a time.
This yields moderately simpler notation.
It also provides a more versatile intermediary result, showing that any two random states can be ``glued'' together by a random function acting on small subsets of both systems.

In order to analyze the moments of this ensemble, we formulate a procedure to ``glue'' together state designs on neighboring blocks using random phase operators. The main technical ingredient is an adaptation of the distinct subspace analysis in which we consider a projection onto the locally distinct subspace of the region acted on by the phase operators. We then bound the statistical closeness to a Haar-random state after introducing each new block. This procedure yields an error that is exponentially small in the block size, so that plugging in known implementations of classical $k$-wise independent functions yields a depth which is doubly logarithmically in the system size. 
\begin{lemma}[Gluing state designs with phase operators] \label{lemma:gluingstates}
    Suppose $\mathcal{S}_A, \mathcal{S}_B$ are $\varepsilon_A$- and $\varepsilon_B$-approximate state $k$-designs on $A$ and $B$, respectively, where $\abs{A}, \abs{B} \geq \xi$ and $k = o(2^{\xi/2})$. Then the ensemble formed by applying a random phase operator $F_\alpha$ on $\xi$ qubits each of $A$ and of $B$ is an $\varepsilon$-approximate state $k$-design with
    \begin{equation}
    \varepsilon = \varepsilon_A + \varepsilon_B + \mathcal{O}(k^2/2^\xi) .
    \end{equation}
\end{lemma}
\proof Let $\mathcal{S}_{FAB}$ and $\mathcal{S}_{F}$ denote the state ensembles generated by applying random phase operators on $\xi$ qubits each of states from $\mathcal{S}_A, \mathcal{S}_B$, and on Haar-random states on the two regions, respectively. Recall that we have
\begin{equation} \label{eq:gluingstatestatement}
\begin{aligned}
    \norm{\chi_{\mathcal{S}_{FAB}} - \chi_H}_1 &\leq \norm{\chi_{\mathcal{S}_{FAB}} - \chi_{\mathcal{S}_{F}}}_1 + \norm{\chi_{\mathcal{S}_{F}} - \chi_H}_1 \\
    &\leq \varepsilon_A + \varepsilon_B + \norm{\chi_{\mathcal{S}_{F}} - \chi_H}_1,
\end{aligned}
\end{equation}
where we have used the fact that the trace norm is non-increasing under the application of any quantum channel to replace the first term. To analyze the second term, we apply projections onto the locally distinct subspaces corresponding to $\Xi_A$ and $\Xi_B$, the respective regions of $A, B$ that the phase operators act on. We have:
\begin{equation} \label{eq:gluingstatehaarcompare}
\begin{aligned}
    \norm{\chi_{\mathcal{S}_{F}} - \chi_H}_1 \leq &\norm{\chi_{\mathcal{S}_{F}} - \left[\Pi_{\dist({\Xi_A})} \otimes \Pi_{\dist({\Xi_B})} \right]\chi_{\mathcal{S}_{F}} \left[\Pi_{\dist({\Xi_A})} \otimes \Pi_{\dist({\Xi_B})}\right]}_1 \\&+ \norm{\left[\Pi_{\dist({\Xi_A})} \otimes \Pi_{\dist({\Xi_B})} \right]\chi_{\mathcal{S}_{F}} \left[\Pi_{\dist({\Xi_A})} \otimes \Pi_{\dist({\Xi_B})}\right] - \chi_H}_1 \\
    \leq &\, \norm{\left[\Pi_{\dist({\Xi_A})} \otimes \Pi_{\dist({\Xi_B})} \right]\chi_{\mathcal{S}_{F}} \left[\Pi_{\dist({\Xi_A})} \otimes \Pi_{\dist({\Xi_B})}\right] - \chi_H}_1 + \mathcal{O}(k^2/2^\xi).
\end{aligned}
\end{equation}
In addition, since the projection operators are diagonal in the computational basis, they commute with the action of the phase operators. We can therefore write the twirl over the phases via
\begin{equation} \label{eq:projectedgluedstate}
\begin{aligned}
    & \, \left[\Pi_{\dist({\Xi_A})} \otimes \Pi_{\dist({\Xi_B})} \right]\chi_{\mathcal{S}_{F}} \left[\Pi_{\dist({\Xi_A})} \otimes \Pi_{\dist({\Xi_B})}\right] \\ = & \, \expect_{\alpha} F_{\alpha}^{\otimes k} \left[\Pi_{\dist({\Xi_A})} \otimes \Pi_{\dist({\Xi_B})} \right] \left[\chi_{H(A)} \otimes \chi_{H(B)}\right] \left[\Pi_{\dist({\Xi_A})} \otimes \Pi_{\dist({\Xi_B})}\right] F_{\alpha}^{\otimes k} \\
    = & \, \expect_{\alpha} F_{\alpha}^{\otimes k} \left[\frac{1}{(k!)^2 Z_A Z_B}  \sum_{\substack{{a \in \{0,1\}^{k(\abs{A}-\xi)}} \\ {\vec x \in \dist(2^\xi,k)}}} \sum_{\pi \in S_k} |a, \vec x \rangle \langle \pi(a, \vec x)| \otimes \sum_{\substack{{b \in \{0,1\}^{k(\abs{B}-\xi)}} \\ {\vec x' \in \dist(2^\xi,k)}}} \sum_{\pi' \in S_k} |\vec x', b \rangle \langle \pi'(\vec x', b)| \right] F_{\alpha}^{\otimes k},
\end{aligned}
\end{equation}
where the normalization factors are given by the dimensions of the corresponding symmetric subspaces
\begin{equation}
    Z_A = {{2^\abs{A} + k - 1} \choose {k}}, \quad Z_B = {{2^\abs{B} + k - 1} \choose {k}}.
\end{equation}
We now observe that the twirl over the random phases eliminates contributions from terms where $\pi \neq \pi'$. The entire projected component is therefore maximally mixed on a subspace of the global symmetric subspace, so the final term in the RHS of Eq. \ref{eq:gluingstatehaarcompare} is bounded by
\begin{equation} \label{eq:gluing state local dist project}
    \begin{aligned}
        \text{RHS} &= 1 - \Tr \left[\Pi_{\dist({\Xi_A})} \otimes \Pi_{\dist({\Xi_B})} \right]\chi_{\mathcal{S}_{F}} \left[\Pi_{\dist({\Xi_A})} \otimes \Pi_{\dist({\Xi_B})}\right]\\
        &= 1 - 2^{k(\abs{A}-\xi)} 2^{k(\abs{B}-\xi)} {{2^\xi} \choose {k}}^2 {{2^\abs{A} + k - 1} \choose {k}}^{-1} {{2^\abs{B} + k - 1} \choose {k}}^{-1} \\
        &= 1 - \frac{2^{k(\abs{A}-\xi)} 2^\xi (2^\xi - 1) \cdots (2^\xi - k + 1)}{(2^\abs{A} + k - 1) (2^\abs{A} + k - 2) \cdots 2^\abs{A}} \frac{2^{k(\abs{B}-\xi)} 2^\xi (2^\xi - 1) \cdots (2^\xi - k + 1)}{(2^\abs{B} + k - 1) (2^\abs{B} + k - 2) \cdots 2^\abs{B}} \\
        &\leq \mathcal{O}(k^2/2^\xi) + \mathcal{O}(k^2/2^\abs{A}) + \mathcal{O}(k^2/2^\abs{B}),
    \end{aligned}
\end{equation}
where in the last step we have used the restriction on $k$ to obtain the asymptotic expansion in terms of $|A|, |B|,$ and $\xi$. Plugging back into Eq. \ref{eq:gluingstatestatement} and collecting the error terms from each step yields
\begin{equation}
    \norm{\chi_{\mathcal{S}_{FAB}} - \chi_H}_1 \leq \varepsilon_A + \varepsilon_B + \mathcal{O}(k^2/2^\xi), 
\end{equation}
up to leading order in $k$ and $\xi$. \qed
\\

We can then plug the approximation error given in Lemma \ref{lemma:kwisephasestate} into the above result repeatedly in order to obtain a bound on the blocked phase state construction. From Lemma \ref{lemma:kwisephasestate}, applying the first layer of phase gates yields approximate state designs on each block, up to error $\mathcal{O}(k^2/2^\xi)$. The result follows from starting with the leftmost block, then applying Lemma \ref{lemma:gluingstates} for each additional block, up to $n/2\xi - 1$. This gives an alternative proof of Theorem~\ref{thm:blocked-phase-design}.

\subsection{Application to pseudorandom states}

Although the main focus of our work concerns state and unitary designs, our results also yield simple constructions and proofs of pseudorandom states and unitaries.
A pseudorandom state (PRS) is an ensemble of quantum states that is indistinguishable from the Haar-random state ensemble by any bounded-time quantum experiment.
Here, a bounded-time quantum experiment is any experiment that involves a bounded number of queries to the unknown state as well as bounded-time quantum circuits in between queries.
We refer to earlier works for a comprehensive and pedagogical introduction~\cite{ji2018pseudorandom,schuster2024random,ma2024construct}.
To be more specific, we say that a PRS on $n$ qubits has security against any $t(n)$-time adversary if the PRS cannot be distinguished from Haar-random by any $t(n)$-time experiment, where $t(n)$ is any function of $n$.

The canonical example of a PRS is the pseudorandom binary phase state.
This is identical to the binary phase state in Eq.~(\ref{eq: low depth binary phase state}), except that $f(x)$ is taken to be a quantum-secure pseudorandom function (PRF) instead of a classical $k$-wise independent function.
By definition, a quantum-secure PRF is indistinguishable from a truly random function by any bounded-time quantum experiment.
Hence, the proof that the pseudorandom binary phase state is a PRS amounts to showing that the truly random binary phase state is indistinguishable from Haar-random.
This is identical to the proof that the random binary phase states form an approximate state $k$-design for any bounded $k = \mathcal{O}(t(n))$ and additive error $1/\varepsilon = \mathcal{O}(t(n))$. 
Therefore, the proof of Lemma~\ref{lemma:kwisephasestate} (originally proven in~\cite{brakerski2019pseudo} using substantially different techniques) also implies that the pseudorandom binary phase state is a PRS with security against any sub-exponential time adversary.

We can now apply an identical logic to the blocked random phase state.
Replacing each local function $f_{a,a+1}$ with a PRF on $2\xi$ qubits yields the simplest construction yet of low-depth pseudorandom states.
%
%
We take each PRF to be indistinguishable from random by any $o(\exp(\xi))$-time adversary~\cite{zhandry2021PRF}.
Our proof of Theorem~\ref{thm:blocked-phase-design} shows that the random ensemble is an approximate state $k$-design for any $k, 1/\varepsilon = o(\exp(\xi))$.
Hence, the original pseudorandom ensemble is indistiguishable from a Haar-random state by any $o(\exp(\xi))$-time adversary.
We set $\xi = \omega(\log n)$, in which case $o(\exp(\xi))$-time security implies security against any polynomial-time adversary since $\poly n = o(\exp(\xi))$.
Our ensemble can be compiled in circuit depth $\poly \log n$ in 1D circuits and $\poly \log \log n$ in all-to-all connected circuits~\cite{schuster2024random}.
These scalings are similar to existing extremely low-depth PRS constructions~\cite{schuster2024random}; the primary benefit is that our construction and proof are much simpler.


\section{Unitary designs in nearly optimal depth}

In this section, we provide the full details of our construction of approximate unitary $k$-designs in circuit depth $\mathcal{O}(\log k \log \log nk/\varepsilon)$. 
We first introduce our design construction based on the blocked LRFC circuit ensemble.
We then describe our novel approach to bounding the measurable error of random unitary ensembles.
We apply this approach to provide short proofs that the PFC ensemble, LRFC ensemble, and blocked LRFC ensemble have small measurable error.
We review the application of our proofs to pseudorandom unitaries in the final section.

\subsection{Blocked Luby-Rackoff-Function-Clifford (LRFC) circuits} \label{sec: low depth LRFC}

As introduced in the main text (Fig.~1), we consider a blocked variant of the LRFC circuit ensemble:
\begin{definition}[Blocked LRFC circuit] \label{def:low-depth-LRFC}
    Consider the circuit shown in Fig.~1(b) of the main text, in which:
    \begin{itemize}
        \item For every patch $a$, $h_{a,a+1}$ is drawn uniformly randomly from functions on $\{0,1\}^{\xi} \rightarrow \{0,1\}^{\xi}$, and instantiates the shuffle gate $S_{a,a+1} \ket{x_a \lVert x_{a+1}} = \ket{x_a \lVert x_{a+1} \oplus h_{a,a+1}(x_a)}$.
        \item For every patch $a$, $f_{a,a+1}$ is drawn uniformly randomly from  binary functions on $\{0,1\}^{2\xi} \rightarrow \{0,1\}$, and instantiates the random function $F_{a,a+1} \ket{x_{a,a+1}} = (-1)^{f_{a,a+1}(x_{a,a+1})} \ket{x_{a,a+1}}$.
        \item For every odd patch $a$, $C_a$ is drawn uniformly randomly from the Clifford group on $\xi$ qubits.
    \end{itemize}.
    \noindent The blocked Luby-Rackoff-Function-Clifford (LRFC) ensemble is given by the family of $n$-qubit random unitaries:
    \begin{equation}
        U = S_o \cdot F_e \cdot F_o \cdot S_e \cdot C_o,
    \end{equation}
    where $S_{o} = \bigotimes_{a \in \text{even}} S_{a,a+1}$ is a tensor product of random shuffles of odd patches conditional on even patches, $F_{e} = \bigotimes_{a \in \text{even}} F_{a,a+1}$ is a tensor product of random functions on each even pairs of patches, $F_{o} = \bigotimes_{a \in \text{odd}} F_{a,a+1}$ is similar for odd pairs of patches, $S_{e} = \bigotimes_{a \in \text{odd}} S_{a,a+1}$ shuffles even patches conditional on odd patches, and $C_o = \bigotimes_{a\in \text{odd}} C_a$ is a tensor product of  Clifford unitaries on odd patches. 
\end{definition}
\noindent As in our discussion of state designs, we make a small change in notation compared to the main text, by replacing the index $i$ for each patch of $2\xi$ qubits with the index $a,a+1$ labeling two patches of $\xi$ qubits each.
The value of $\xi$ is tunable and will depend on the desired $n, k, \varepsilon$ of the final unitary design.

In the following subsections, we will prove that the blocked LRFC ensemble forms an approximate unitary design with small measurable error:
\begin{theorem}[Blocked LRFC circuits are designs] \label{thm:low-depth-LRFC-designs}
    The blocked LRFC ensemble is an approximate unitary $k$-design up to measurable error $\varepsilon = 3nk^2/2^\xi \xi$.
\end{theorem}
\noindent  We describe in detail our approach to proving this result in the remainder of this section; the final proof is found in Section~\ref{sec: pf blocked LRFC}.

From Theorem~\ref{thm:low-depth-LRFC-designs}, we can construct approximate unitary designs in nearly optimal depth by drawing each small random function $h_{a,a+1}$ and $f_{a,a+1}$ from a $2k$-wise independent random function on $\xi$ or $2\xi$ bits, respectively.
We will also draw each random unitary $C_a$ from any exact unitary 2-design on $\xi$ qubits.
The circuit depth of the resulting unitary ensemble is equal to twice the depth of a $2\xi$-bit random function plus twice the depth of a $\xi$-bit random function plus the depth of a $\xi$-bit exact unitary 2-design. 
From our earlier review and analysis (Sections~\ref{sec: compilation}), this yields designs in the following circuit depths:

\begin{fact}[Circuit resources for blocked LRFC circuits] \label{fact:blocked-LRFC-resources}
    The $k$-wise independent blocked LRFC circuit can be implemented with either:
    \begin{enumerate}
        \item \textbf{Low depth:} $\mathcal{O}(\log k \cdot \log \xi)$ depth using $\mathcal{O}(nk \cdot \log \xi \cdot \log \log \xi)$ ancilla qubits.
        \item \textbf{Low ancilla:} $\mathcal{O}(k \cdot \log \xi)$ depth using $\mathcal{O}(n \cdot \log \xi \cdot \log \log \xi)$ ancilla qubits.
    \end{enumerate}
    Both implementations use $3nk+(5/2)n$ bits of randomness.
\end{fact}

Combining Fact~\ref{fact:blocked-LRFC-resources} and Theorem~\ref{thm:low-depth-LRFC-designs} yields the following:

\begin{corollary}[Circuit resources for unitary $k$-designs] \label{cor:unitary-design-resources}
    For any $k^2/\varepsilon \leq 2^n/3$, setting $\xi = \log_2(3nk^2/\varepsilon) \leq n$ yields approximate unitary $k$-designs with measurable error $\varepsilon$ using either:
    \begin{itemize}
        \item \textbf{Low depth:} $\mathcal{O}(\log k \cdot \log \log(nk/\varepsilon))$ depth and $\mathcal{O}(k n \log(nk/\varepsilon) \log \log(nk/\varepsilon))$ ancilla qubits.
        \item \textbf{Low ancilla:} $\mathcal{O}(k \log \log(nk/\varepsilon))$ depth and $\mathcal{O}(n \log(nk/\varepsilon) \log \log(nk/\varepsilon))$ ancilla qubits.
    \end{itemize}
    The randomness requirement is $3nk + (5/2)n$ bits.
\end{corollary}

The number of bits of randomness follows because $F_e$ and $F_o$ together use $2nk$ bits of randomness~\cite{wegman1981new}.
The shuffles $S_e$ and $S_o$ use $nk$ bits of randomness, since their functions act on half as many bits.
Finally, the tensor product of exact unitary 2-designs $C$ uses $(n/2\xi)5\xi = (5/2)n$ bits of randomness~\cite{cleve2015near}. 
The optimal number of bits of randomness for a unitary $k$-design is $2nk$ (up to sub-leading additive factors), since applying the Haar twirl to one side of the EPR state yields a mixed state with rank $2^{2nk}$ (up to sub-leading multiplicative factors).
Hence, our ensemble is a mere factor of $3/2$ away from optimal in this regard, when $k$ is large.

\subsection{A simple and general approach for bounding the measurable error} \label{sec:bounding measurable error}

\begin{figure}
\centering
\includegraphics[width=0.5\columnwidth]{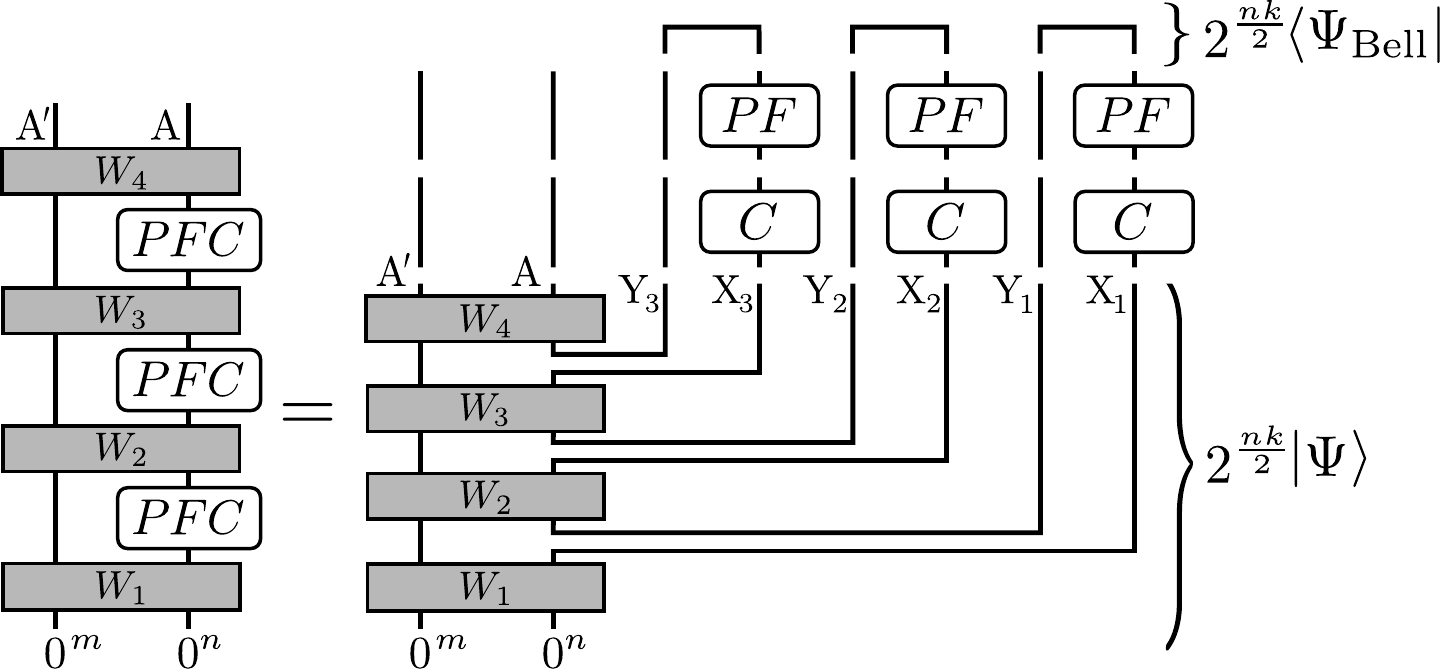}
\caption{A visual depiction of our expression Eq.~(\ref{eq: 1}) for the output $\ket{\psi^U_W}$ of any quantum experiment querying $U = PFC$.
On the left, the unitary acts $k$ times in sequence on an $n$-qubit register $A$, interleaved with arbitrary quantum operations, $W_j$, involving ancilla qubits $A'$.
On the right, $X_j$ denotes the register that the $j$-th unitary acts on, for $j = 1,\ldots,k$, and $Y_j$ denotes the register that $X_j$ is entangled with in $\ket{\Psi_{\text{Bell}}}$. 
}
\label{fig: PRU}
\end{figure}

As discussed in the main text, to prove Theorem~\ref{thm:low-depth-LRFC-designs}, we develop a new and general approach for bounding the measurable error of random unitary ensembles.
Our approach provides a convenient alternative to the path-recording framework introduced by~\cite{ma2024construct} (see Section~\ref{sec: PRU} for an expanded comparison).
In this subsection, we provide a high-level overview of our approach.
We then demonstrate it in full detail in the following subsection, by giving a short self-contained proof that the PFC ensemble has small measurable error.
We then apply the same approach to prove that the LRFC and blocked LRFC ensembles have small measurable error as well.
Our proofs for the three ensembles are nearly identical to one another, illustrating the versatility of our approach.

Our approach melds and refines ideas from several recent works on pseudorandom states and unitaries~\cite{giurgica2023pseudorandomness,jeronimo2023pseudorandom,metger2024simple,ma2024construct,schuster2024random,SRU2025}.
It consists of three key steps: 
\begin{enumerate}
    \item We first reformulate any quantum experiment involving sequential applications of a random unitary as an experiment involving parallel applications and post-selection~\cite{schuster2024random} (Fig.~\ref{fig: PRU}).

    \item We then prove that one can insert a projector onto the distinct subspace~\cite{giurgica2023pseudorandomness,jeronimo2023pseudorandom,metger2024simple} (or, a local variant of it) in between key components of the parallel random unitaries~\cite{metger2024simple,ma2024construct,SRU2025} in the post-selected experiment.

    \item We complete the proof by showing that the random unitary ensemble of interest acts identically to a Haar-random unitary on the (local) distinct subspace~\cite{giurgica2023pseudorandomness,jeronimo2023pseudorandom,metger2024simple,ma2024construct,SRU2025}.
\end{enumerate}
\noindent Several of these ideas are implicit in the path-recording framework of~\cite{ma2024construct}.
We adopt similar notation to~\cite{ma2024construct} to help highlight these connections.
At the same time, we deviate from~\cite{ma2024construct} in  important ways, including our use of post-selection and the introduction of several simple operator identities that enable a direct analysis.

Let us first discuss Step 1.
Following Definition~\ref{def: qu expt}, the output state of any experiment querying $U$ can be expressed $W_{k+1} U W_{k} U \cdots W_2 U W_1 \ket{0^{n+m}}$, where $U$ acts on an $n$-qubit subsystem $A$ and each $W_i$ acts on an arbitrarily large system $AA'$.
As shown in Fig.~\ref{fig: PRU}, we can re-arrange this expression so that the $k$ sequential applications of $U$ act instead as $k$ parallel applications on ancilla qubits.
This requires $2k$ new ancilla registers of $n$ qubits each, with one register $X_i$ for each input to a copy of $U$ and one register $Y_i$ for each output.
The registers are connected by an initial EPR state (between each $Y_i$ and $X_{i+1}$) and a final EPR projection (between each $X_i$ and $Y_i$).
The entire reformulation is best viewed in tensor network notation; nonetheless, in equations, we have
\begin{equation} \label{eq: 1}
	\ket{\psi_U} = 2^{nk} (\mathbbm{1}_{\color{gray}AA'} \otimes \bra{ \Psi_{\text{Bell}}}_{\color{gray}XY } ) (\mathbbm{1}_{\color{gray}ABY} \otimes \boldsymbol{U}_{\color{gray}X})\ket{\Psi}_{\color{gray}ABXY},
\end{equation}
where $\ket{\Psi}$ is an initial state determined by the choice of $W_i$ (see Fig.~\ref{fig: PRU}),  $\bra{\Psi_{\text{EPR}}}$ is the tensor product of $k$ EPR projectors on each $X_i$ and $Y_i$ (with $X \equiv \cup_i X_i$, $Y \equiv \cup_i Y_i$), and $\boldsymbol{U} \equiv U^{\otimes k}$.
We suppress the dependence of $\ket{\psi_U}$ and $\ket{\Psi}$ on $W$ for brevity.
The reformulation incurs a post-selection overhead of $2^{nk}$ and hence is not physical. Nonetheless, it is convenient for analysis due to the parallelism in $U$.

We leave the details of Steps 2 and 3 to  the following subsection.
For now, we simply recall several definitions.
The operator $\PD = \sum_{x \in \text{dist}} \dyad{x}$ projects onto the distinct subspace, where $x = (x_1,\ldots,x_k)$ with all $x_j$ distinct.
The number $\mathfrak{D} = (2^n)!/(2^n-k+1)!$ denotes the distinct subspace dimension, $1-k^2/2^n \leq \mathfrak{D}/2^{nk} \leq 1$.

\subsection{Proof of Theorem~\ref{thm:PFC-designs}: The PFC ensemble has small measurable error} \label{sec: pf PFC}

Let $\delta = k^2/2^n$, $U = PFC$ and $\bs{P} \equiv P^{\otimes k}$, $\bs{F} \equiv F^{\otimes k}$, $\bs{C} \equiv C^{\otimes k}$.
The expected output state is given by
\begin{equation}
	\rho = \E_{U \sim PFC} \Big[ \dyad{\psi_U}_{\color{gray}AA'} \Big] = \E_{C} \left[ \tr^{}_{\color{gray}XY} \! \big( B'_{\color{gray}XY} \bs{C}^{}_{\color{gray}X} \dyad*{\Psi}^{}_{\color{gray}AA'XY}\bs{C}^\dagger_{\color{gray}X} \big) \right]
\end{equation}
where $B'_{\color{gray}XY} \equiv 4^{nk} \E_{P, F} \big[ (\bs{P} \bs{F})^\dagger_{\color{gray}X} \dyad{\Psi_{\text{Bell}}}_{\color{gray}XY} (\bs{P} \bs{F})_{\color{gray}X} \big]$ is the twirled Bell state multiplied by $4^{nk}$.
We also define an analogous operator $B_{\color{gray} XY} \equiv \sum_{\pi \in S_k} \pi_{\color{gray} X} \otimes \pi_{\color{gray} Y}$, where $\pi$ permutes the $k$ registers of $X$ or $Y$.
($B$ is equal to the twirl of the Bell state over a Haar-random $U$ up to small relative error~\cite{schuster2024random}; however, we do not  use this in our proof.)

Our proof uses five facts. The first two are fundamental and correspond to Steps 3 and 2 in the proof approach overview, respectively. The final three are for technical convenience.

\vspace{3mm}
\noindent \text{(1.)}
    $B'$ and $(2^{nk}/\mathfrak{D}) B$ are equal on the distinct subspace, $ B'( \PD \otimes \PD ) = (2^{nk}/\mathfrak{D}) B( \PD \otimes \PD )$.
    This follows by expanding the EPR projector as $2^{nk} \dyad{\Psi_{\text{Bell}}} (\PD \otimes \PD) = \sum_{\tilde{x}} \sum_{x\in \text{dist}} \dyad{\tilde x, \tilde x}{x, x}$.
    The twirl over $F$ yields, 
    \begin{equation} \label{eq: twirl F}
        \E_F \Big[ \bs F \dyad{\tilde x, \tilde x}{x, x} \bs F^\dagger \Big]
        = \E_f \Big[ (-1)^{\sum_i f(\tilde{x}_i)+\sum_i f(x_i)} \dyad{\tilde x, \tilde x}{x, x} \Big] = \sum_\pi \delta_{\tilde{x} = \pi x} \cdot \dyad{\pi x, \pi x}{x, x},
    \end{equation}
    where $\pi \in S_k$ is any permutation, and we abbreviate $\pi x \equiv (x_{\pi(1)},\ldots,x_{\pi(k)})$. 
    Note that we can exchange the order of $F$ and $P$ since the average over a random function commutes with any permutation.
    The condition $\tilde{x} = \pi x$ ensures that $f(\tilde x) + f(x) = 0$; otherwise the sign vanishes upon averaging.
    The remaining average over $P$ yields
    \begin{equation}
        \E_P \Big[ \bs P \dyad{\pi x, \pi x}{x, x} \bs P^\dagger \Big] = (\pi \otimes \pi) \cdot \E_P \Big[ \bs P \dyad{x, x}{x, x} \bs P^\dagger \Big] = (1/\mathfrak{D}) \sum_{y \in \text{dist}} \dyad{\pi y, \pi x}{y, x},
    \end{equation}
    since a random permutation maps any distinct bitstring $x$ to a uniformly random distinct bitstring $y$.
    Combining these expressions, we have $B'( \PD \otimes \PD ) = (2^{nk}/\mathfrak{D}) \sum_{x,y \in \text{dist}} \dyad{\pi y,\pi x}{y,x} = (2^{nk}/\mathfrak{D}) B( \PD \otimes \PD )$ as desired.
    
\vspace{4mm}
\noindent \text{(2.)}
    The twirl of $\PD$ over any 2-design is near the identity, $\big\lVert \mathbbm{1} - \E_{C } \big[ \bs{C}^\dagger \PD \bs{C} \big] \big\rVert_\infty \leq \delta$~\cite{metger2024simple}.
    Intuitively, this follows because $C$ scrambles states into random superpositions of bitstrings, and most bitstrings are distinct.
    In  detail, we can upper bound $\mathbbm{1} - \Pi_{\text{dist}} \leq \sum_{1\leq i<j \leq k} \Pi_{ij}$ where $\Pi_{ij} = \sum_{z \in \{0,1\}^n} \dyad{z, z}$.
    The twirl over $C$ then yields $\E_{C} [\bs{C} \Pi_{ij} \bs{C}^\dagger  ] = (\mathbbm{1}+\mathcal{S}_{ij})(2^n+1)$, where $\mathcal{S}_{ij}$ is the swap operator.
    Hence, the norm is less $\sum_{i<j} \frac{2}{2^n+1} \leq k^2/2^n = \delta$.

\vspace{4mm}
\noindent \text{(3.)}
$B$ and $B'$ commute with $\PD \otimes \mathbbm{1}$ and $\mathbbm{1} \otimes \PD$. 
$B$ also commutes with $\bs{C}$.

\vspace{4mm}
\noindent \text{(4.)} For any positive semi-definite matrices $P,Q,R$ where $P$ and $Q$ commute, we have $\tr(PQR)  \leq \lVert Q \rVert_\infty \tr(P R)$.
This follows from Holder's inequality, $\tr(PQR) = \text{tr}(Q \sqrt{P} R \sqrt{P}) \leq \lVert Q \rVert_\infty \lVert \sqrt{P} R \sqrt{P} \rVert_1 =  \lVert Q \rVert_\infty \tr(P R)$. 

\vspace{4mm}
\noindent \text{(5.)} $\tr(B (\mathbbm{1} \otimes \PD) \dyad*{\Psi}) = \mathfrak{D}/2^{nk}$. 
This is shown below the current proof.

\vspace{4mm}
We can now prove the claim.
We insert $\mathbbm{1} = \PD \otimes \PD + (\mathbbm{1}-\PD \otimes \PD)$ to decompose $\rho$ into two terms,
\begin{align*}
	\rho = \rho_{\text{dist}} + \delta\rho \equiv \E_{C} \Big[ \! \tr_{\color{gray} XY}( B' (\PD \otimes \PD) \bs{C} \dyad*{\Psi} \bs{C}^\dagger ) \Big] + \E_{C} \Big[ \! \tr_{\color{gray} XY}( B' (\mathbbm{1}-\PD \otimes \PD) \bs{C} \dyad*{\Psi} \bs{C}^\dagger ) \Big],
\end{align*}
\noindent where each operator acts on the same subsystems as before. We bound the second term as follows,
\begin{align*}
    \lVert \delta \rho \rVert_1 = \tr( \delta \rho ) & = \E_{C} \big[ \text{tr}( B' (\mathbbm{1}-\PD \otimes \PD) \bs{C} \dyad*{\Psi} \bs{C}^\dagger ) \big] && \text{(since $\delta \rho$ is positive)} \\
    & = 1 - (2^{nk}/\mathfrak{D}) \tr \big( B \cdot \E_{C} \big[ \bs{C}^\dagger \PD \bs{C} \big] \otimes \PD \cdot \dyad*{\Psi} \big) && \text{(since $\tr(\rho) = 1$ and (1.))} \\
    & \leq 1 - (2^{nk}/\mathfrak{D})(1-\delta) \tr \big( B (\mathbbm{1} \otimes \PD) \dyad*{\Psi} \big) && \text{(from (2.), (3.), (4.))} \\
    & = 1 - (2^{nk}/\mathfrak{D})(1-\delta)(\mathfrak{D}/2^{nk}) = \delta. && \text{(from (5.))}
\end{align*}
Meanwhile, the first term is close to a fixed density matrix $\rho_a$,
\begin{align*}
    \rho_{\text{dist}} & = (2^{nk}/\mathfrak{D}) \tr_{\color{gray} XY} \!  \big( B \cdot \E_{C} \big[ \bs{C}^\dagger \PD \bs{C} \big] \otimes \PD \cdot \dyad*{\Psi} \big) && \text{(from (1.))} \\
    & = (2^{nk}/\mathfrak{D}) \tr_{\color{gray} XY} \! \big( B (\mathbbm{1} \otimes \PD) \dyad*{\Psi} \big) + \Delta \equiv \rho_a + \Delta, && \text{(from (2.), (3.), (4.))}
\end{align*} 
where $\lVert \Delta \rVert_1 \leq \delta (2^{nk}/\mathfrak{D}) \tr( B (\mathbbm{1} \otimes \PD) \dyad*{\Psi}) = \delta$ from (2.), (3.), (4.), (5.).

This completes the proof: We have $\lVert \rho - \rho_a \rVert_1 \leq 2\delta$ when $C$ is drawn from any 2-design.
Since both the Clifford and Haar ensembles are 2-designs, we have $\lVert \rho - \rho_H \rVert_1 \leq \lVert \rho_H - \rho_a \rVert_1 + \lVert \rho - \rho_a \rVert_1 \leq 4\delta = 4k^2/2^n$. \qed

\vspace{3mm}
\noindent \textbf{Proof of (5.):} 
We proceed by induction. 
%
Let $\Pi_{\text{dist}}^j$ be the distinct subspace projector on $(XY)_{\leq j}$, and $\rho_j = \tr_{{\color{gray}AA'(XY)_{>j}}}\!( \dyad*{\Psi} )$ the reduced density matrix.
The key property we use is that $\rho_j$ is maximally mixed on $Y_j$. 

\vspace{2mm}
\noindent \textbf{Base case ($k=1$): }Here, $B = \mathbbm{1} \otimes \mathbbm{1}$ and $\PD = \mathbbm{1}$. Hence, $\tr(B(\mathbbm{1}\otimes \PD)\dyad*{\Psi}) = 1$.

\vspace{2mm}
\noindent \textbf{Inductive step: }We expand $B = \sum_\pi \pi \otimes \pi$ and use the following identity: $\tr_j(\pi \Pi^j_{\text{dist}}) = (2^n-j+1) \Pi^{j-1}_{\text{dist}}$ if $\pi(j)=j$, and $\tr_j(\pi \Pi^j_{\text{dist}}) = 0$ otherwise.
This allows us to trace over $Y_j$, and, since this enforces $\pi(j)=j$, $X_j$ as well, 
\begin{equation}
    \sum_{\pi \in S_j} \tr_{{\color{gray}(XY)_{\leq j}}} \!\! \big((\pi \otimes \pi)(\mathbbm{1}\otimes \Pi_{\text{dist}}^j) \rho_j \big) = \frac{2^n-j+1}{2^n} \sum_{\pi' \in S_{j-1}} \tr_{{\color{gray}(XY)_{\leq j-1}}} \!\! \big((\pi' \otimes \pi')(\mathbbm{1}\otimes \Pi_{\text{dist}}^{j-1} ) \rho_{j-1} \big),
\end{equation}
where $\pi'$ is $\pi$ restricted to the first $j-1$ elements.
Iterating $k$ times yields $((2^n)!/(2^n-k+1)!)/2^{nk} = \mathfrak{D}/2^{nk}$. \qed

\subsection{Proof of Theorem~\ref{thm:LRFC-designs}: The LRFC ensemble has small measurable error}  \label{sec: pf LRFC}

%


To adapt our proof to the LRFC ensemble, we will project to slightly smaller subspaces than in the PFC ensemble.
We let $\PD^L$ project to the distinct subspace on the left $n/2$ qubits, and $\PD^R$ on the right $n/2$ qubits.
The left and right distinct subspace have dimensions $\mathfrak{D}_{L,R} = 2^{n/2} \cdot (2^{n/2})!/(2^{n/2}-k+1)! \geq 2^{nk}(1-\delta')$ where $\delta' = k^2/\sqrt{2^n}$.


We now proceed similar to PFC.
Let $U = S_L S_R FC$ and $\bs{L} \equiv S_L^{\otimes k}$, $\bs{R} \equiv S_R^{\otimes k}$,  $\bs{F} \equiv F^{\otimes k}$, $\bs{C} \equiv C^{\otimes k}$.
The expected output state is $\rho = \E_{C} \tr^{}_{\color{gray}XY} \! \big( B'' \bs{C} \dyad*{\Psi} \bs{C}^\dagger \big) $,
where
\begin{equation}
    B''_{\color{gray}XY} \equiv 4^{nk} \E_{LRF} \big[ (\bs{L}\bs{R}\bs{F})^\dagger_{\color{gray}X} \dyad{\Psi_{\text{Bell}}}_{\color{gray}XY} (\bs{L}\bs{R}\bs{F})_{\color{gray}X} \big]
    = 4^{nk} \E_{LRF} \big[ (\bs{F}\bs{R})^\dagger_{\color{gray}X}
    (\bs{L})^\dagger_{\color{gray}Y}
    \dyad{\Psi_{\text{Bell}}}_{\color{gray}XY} (\bs{F}\bs{R})_{\color{gray}X} (\bs{L})_{\color{gray}Y} \big].
\end{equation}
In the second step, we commute $F$ and $R$, and pull $L$ to the opposite side of the EPR pair using $\bs{L}_{\color{gray}X} \ket{\Psi_{\text{Bell}}} = \bs{L}^T_{\color{gray}Y} \ket{\Psi_{\text{Bell}}}$ and $L^T = L$, for future convenience.
%
We also define $B_{\color{gray} XY} \equiv \sum_{\pi \in S_k} \pi \otimes \pi$ as before.
We again use five facts.

\vspace{3mm}
\noindent \text{(1'.)}
    $B''$ and $B$ are equal on the modified distinct subspace, $(\PD^{LR})_{\color{gray}XY} \equiv (\PD^L)_{\color{gray}X} \otimes (\PD^R)_{\color{gray}Y}$.
    We expand $2^{nk} \dyad{\Psi_{\text{Bell}}} (\PD^{LR})_{\color{gray}XY} = \sum_{\tilde{x}} \sum_{x_L, x_R \in \text{dist}} \dyad{\tilde x, \tilde x}{x, x}$, abbreviating $x \equiv x_L \lVert x_R$.
    The twirl over $F$ sets $\tilde x = \pi x$ for $\pi \in S_k$ as before.
    Meanwhile, the twirl over $L$ yields 
    \begin{equation} \label{eq: left twirl}
        \E_L \Big[ \bs{L}^\dagger_{\color{gray}Y} \dyad{\pi x, \pi x}{x, x} \bs{L}_{\color{gray}Y} \Big] = (\pi \otimes \pi) \cdot \E_L \Big[ \bs{L}^\dagger_{\color{gray}Y} \dyad{x, x}{x, x} \bs{L}_{\color{gray}Y} \Big] = (1/2^{(n/2)k}) \sum_{y_L} \dyad{\pi x, \pi (y_L \lVert x_R)}{x, (y_L \lVert x_R)},
    \end{equation}
    since $y_L \equiv x_L + f_L(x_R)$ is a uniformly random set of $k$ $\frac{n}{2}$-bit strings when $x_R$ is distinct.
    We can then twirl over $R$ in an identical manner,
    \begin{equation} \label{eq: right twirl}
        (\pi \otimes \pi) \cdot \E_R \Big[ \bs{R}^\dagger_{\color{gray}X} \dyad{x, (y_L \lVert x_R)}{x, (y_L \lVert x_R)} \bs{R}_{\color{gray}X} \Big] 
        = (1/2^{(n/2)k}) \sum_{y_R} \dyad{\pi (x_L \lVert y_R), \pi (y_L \lVert x_R)}{(x_L \lVert y_R), (y_L \lVert x_R)}.
    \end{equation}
    In total, this shows $B''( \PD^{LR})_{\color{gray}XY} = \sum_\pi \sum_{x_L,x_R \in \text{dist}} \sum_{y_L, y_R}  \dyad{\pi (x_L \lVert y_R), \pi (y_L \lVert x_R)}{(x_L \lVert y_R), (y_L \lVert x_R)} = B( \PD^{LR})_{\color{gray}XY}$.
    
\vspace{3mm}
\noindent \text{(2'.)}
    The twirl of $(\PD^L)_{\color{gray}X}$ over any 2-design is near the identity, $\big\lVert \mathbbm{1} - \E_{C } \big[ \bs{C}^\dagger \PD^L \bs{C} \big] \big\rVert_\infty \leq \delta'$.
    This follows from a near-identical computation as before.

\vspace{3mm}
\noindent \text{(3'.)}
$B$ and $B''$ commute with $(\PD^L)_{\color{gray}Y}$ and $(\PD^R)_{\color{gray}X}$. 
$B$ also commutes with $\bs{C}$.

\vspace{3mm}
\noindent \text{(4.)} For any positive semi-definite matrices $P,Q,R$ where $P$ and $Q$ commute, we have $\tr(PQR)  \leq \lVert Q \rVert_\infty \tr(P R)$.

\vspace{3mm}
\noindent \text{(5'.)} $\tr(B (\mathbbm{1}_{\color{gray}X} \otimes (\PD^R)_{\color{gray}Y}) \dyad*{\Psi}) = \mathfrak{D}_R/2^{nk}$. 
This follows from an identical computation as before.

\vspace{4mm}
We can now prove the claim.
We insert $\mathbbm{1} = (\PD^L)_{\color{gray}X} \otimes (\PD^R)_{\color{gray}Y} + (\mathbbm{1}-(\PD^L)_{\color{gray}X} \otimes (\PD^R)_{\color{gray}Y})$ to decompose $\rho$,
\begin{align*}
	\rho = \rho_{\text{dist}} + \delta\rho \equiv \E_{C} \Big[ \! \tr_{\color{gray} XY}( B'' (\PD^L \otimes \PD^R) \bs{C} \dyad*{\Psi} \bs{C}^\dagger ) \Big] + \E_{C} \Big[ \! \tr_{\color{gray} XY}( B'' (\mathbbm{1}-\PD^L \otimes \PD^R) \bs{C} \dyad*{\Psi} \bs{C}^\dagger ) \Big],
\end{align*}
\noindent where each operator acts on the same subsystems as before. We bound the second term as follows,
\begin{align*}
    \lVert \delta \rho \rVert_1 = \tr( \delta \rho ) & = \E_{C} \big[ \text{tr}( B'' (\mathbbm{1}-\PD^L \otimes \PD^R) \bs{C} \dyad*{\Psi} \bs{C}^\dagger ) \big] && \text{(since $\delta \rho$ is positive)} \\
    & = 1 - \tr \big( B \cdot \E_{C} \big[ \bs{C}^\dagger \PD^L \bs{C} \big] \otimes \PD^R \cdot \dyad*{\Psi} \big) && \text{(since $\tr(\rho) = 1$ and (1'.))} \\
    & \leq 1 - (1-\delta') \tr \big( B (\mathbbm{1} \otimes \PD^R) \dyad*{\Psi} \big) && \text{(from (2'.), (3'.), (4.))} \\
    & = 1 - (1-\delta')(\mathfrak{D}_R/2^{nk}) \leq 2 \delta'. && \text{(from (5'.))}
\end{align*}
Meanwhile, the first term is close to a fixed density matrix $\rho_a$,
\begin{align*}
    \rho_{\text{dist}} & = \tr_{\color{gray} XY} \!  \big( B \cdot \E_{C} \big[ \bs{C}^\dagger \PD^L \bs{C} \big] \otimes \PD^R \cdot \dyad*{\Psi} \big) && \text{(from (1'.))} \\
    & = \tr_{\color{gray} XY} \! \big( B (\mathbbm{1} \otimes \PD^R) \dyad*{\Psi} \big) + \Delta \equiv \rho_a + \Delta, && \text{(from (2'.), (3'.), (4.))}
\end{align*} 
where $\lVert \Delta \rVert_1 \leq \delta' \tr( B (\mathbbm{1} \otimes \PD^R) \dyad*{\Psi}) = \delta'(\mathfrak{D}_R/2^{nk}) \leq \delta'$ from (2'.), (3'.), (4.), (5'.).

This completes the proof: We have $\lVert \rho - \rho_a \rVert_1 \leq 3\delta'$ when $C$ is drawn from any 2-design.
Since both the Clifford and Haar ensembles are 2-designs, we have $\lVert \rho - \rho_H \rVert_1 \leq \lVert \rho_H - \rho_a \rVert_1 + \lVert \rho - \rho_a \rVert_1 \leq 6\delta' = 6k^2/\sqrt{2^n}$. \qed

\subsection{Proof of Theorem~\ref{thm:low-depth-LRFC-designs}: The blocked LRFC circuit has small measurable error}  \label{sec: pf blocked LRFC}

Finally, we can address the blocked LRFC circuit.
Similar to our proof of nearly optimal state designs, we will project onto a tensor product of \emph{local} distinct subspaces $\PD^a$ on each patch $a$ of $\xi$ qubits.
%
%
Motivated by our analysis of the LRFC ensemble,
we will divide these patches into two alternating subsets: even patches, which are shuffled by $S_e$ in the second circuit layer, and odd patches, which are shuffled by $S_o$ in the fifth layer.
We let 
\begin{equation}
    \PD^{\text{loc}(o)} = \bigotimes_{a \in \text{odd}} \PD^a \,\,\,\,\,\, \text{ and } \,\,\,\,\,\,
    \PD^{\text{loc}(e)} = \bigotimes_{a \in \text{even}} \PD^a
\end{equation}
denote the local distinct subspace projectors on the odd and even patches, respectively.
These will play an analogous role to $\PD^L$ and $\PD^R$ in our LRFC proof.
The subspaces have dimensions $\mathfrak{D}_{e} =\mathfrak{D}_{o} = 2^{(n/2)k} \cdot (\mathfrak{D}_a)^{m/2} \geq 2^{nk} (1-k^2/2^{\xi})^{m/2} \geq 2^{nk}(1 - \tilde \delta)$ where $\tilde \delta \equiv (m/2) k^2/2^\xi$, where $\mathfrak{D}_a = (2^\xi)!/(2^\xi-k+1)!$ is the dimension of each patch.



With this setup in hand, our proof proceeds nearly identically to  PFC and LRFC.
Let $U = S_o F_o F_e S_e C_e$ and $\bs{S}_{e,o} \equiv S_{e,o}^{\otimes k}$,  $\bs{F}_{e,o} \equiv F_{e,o}^{\otimes k}$, $\bs{C}_e \equiv C_e^{\otimes k}$.
The expected output is $\rho = \E_{C} \tr^{}_{\color{gray}XY} \! \big( \tilde B \, \bs{C} \dyad*{\Psi} \bs{C}^\dagger \big) $,
where
\begin{equation}
    \tilde B_{\color{gray}XY} \equiv  4^{nk} \E_{LRF} \big[ (\bs{F}_o \bs{F}_e \bs{S}_e)^\dagger_{\color{gray}X}
    (\bs{S}_o)^\dagger_{\color{gray}Y}
    \dyad{\Psi_{\text{Bell}}}_{\color{gray}XY} (\bs{F}_o \bs{F}_e \bs{S}_e)_{\color{gray}X} (\bs{S}_o)_{\color{gray}Y} \big],
\end{equation}
using a similar manipulation $(\bs{S}_e)_{\color{gray}X} \ket{\Psi_{\text{Bell}}} = (\bs{S}_e)_{\color{gray}Y} \ket{\Psi_{\text{Bell}}}$ as in the LRFC proof.
We define $B_{\color{gray} XY} \equiv \sum_{\pi \in S_k} \pi \otimes \pi$ as before.
Our five facts are modified as follows.

\vspace{3mm}
\noindent \text{(1''.)}
    $\tilde B$ and $B$ are equal on the local distinct subspace, $(\PD^{\text{loc}})_{\color{gray}XY} \equiv (\PD^{\text{loc}(o)})_{\color{gray}X} \otimes (\PD^{\text{loc}(e)})_{\color{gray}Y}$.
    To see this, we note that $(\PD^{\text{loc}(o)})_{\color{gray}X}$ commutes with $(\bs{F}_e \bs{F}_o \bs{S}_o)_{\color{gray}X}$ and 
    $(\PD^{\text{loc}(e)})_{\color{gray}Y}$ commutes with $(\bs{S}_e)_{\color{gray}Y}$.
    We can then expand
    $2^{nk} \dyad{\Psi_{\text{Bell}}} (\PD^{\text{loc}})_{\color{gray}XY} = \sum_{\tilde{x}} \sum_{x \in \text{loc-dist}} \dyad{\tilde x, \tilde x}{x, x}$, where the sum is over $x$ that are distinct on every patch.

    \vspace{3mm}
    \noindent Let us first analyze the twirl over $F_e$ and $F_o$. 
    The key result is that this is equivalent to the twirl over an $n$-qubit random function, when restricted to the local distinct subspace:
    \begin{equation} \label{eq: twirl F low}
        \E_F \Big[ \bs{F} \dyad{\tilde x, \tilde x}{x, x} \bs{F}^\dagger \Big]
        = \E_{f_{a,a+1} \, \forall a} \Big[ (-1)^{\sum_{j,a} f_{a,a+1}(\tilde{x}_{a,a+1}^{(j)})+\sum_{j,a} f_{a,a+1}(x_{a,a+1}^{(j)})} \dyad{\tilde x, \tilde x}{x, x} \Big] = \sum_\pi \delta_{\tilde{x} = \pi x} \cdot \dyad{\pi x, \pi x}{x, x}.
    \end{equation}
    For comparison, see Eq.~(\ref{eq: twirl F}) for the twirl over an $n$-qubit random function.
    Here, $x^{(j)}_{a,a+1}$ denotes the support of the $j$-th bitstring on patches $a,a+1$, and the sum runs over all patches $a = 1,\ldots,m$. 
    The random functions $f_{a,a+1}$ for even $a$ come from $F_e$ and for odd $a$ from $F_o$.
    The average over each random function forces $\tilde x$ and $x$ to be permutations of one another within each pair of nearest-neighbor patches, $\tilde x^{(j)}_{a,a+1} = \pi_{a} x^{(j)}_{a,a+1}$.
    This is equivalent to $\tilde x^{(j)}_a = \pi_{a} x^{(j)}_{a}$ and $\tilde x^{(j)}_{a+1} = \pi_{a} x^{(j)}_{a+1}$ for each $a$.
    However, since each $x^a$ is distinct, this can be true only if $\pi_{a-1} = \pi_{a}$ for every $a$.
    This implies that \emph{all} the permutations are in fact equal, $\pi_{a} = \pi_1 \equiv \pi$ for all $a$, which yields the final expression.

    \vspace{3mm}
    \noindent We can now turn to the twirl over $S_e$ and $S_o$. The key result is that these twirls send each local distinct bitstring $x$ to a uniformly random bitstring $y$.
    Together with Eq.~(\ref{eq: twirl F low}), this implies that the blocked LRF circuit acts identically to a global $n$-qubit LRF circuit, when restricted to the local distinct subspace.
    We begin with $S_o$:
    \begin{equation} \nonumber
         (\pi \otimes \pi) \cdot \E_{S_o} \Big[ (\bs{S}_o)^\dagger_{\color{gray}Y} \dyad{ x, x}{x, x} (\bs{S}_o)_{\color{gray}Y} \Big] = (1/2^{(n/2)k}) \sum_{y_o} \dyad{\pi x, \pi (y_o \lVert x_e)}{x, (y_o \lVert x_e)},
    \end{equation}
    since $y_o \equiv \bigotimes_{a\in \text{odd}} (x_a + h_{a,a+1}(x_{a+1}))$ is a uniformly random set of $k$ $\frac{n}{2}$-bit strings when each $x_{a+1}$ is distinct.
    Here, we let $x_e$ denote the concatenation of all bitstrings on even patches, and $x_o$ the concatenation on all odd patches,  $x \equiv x_o \lVert x_e$.
    We can now twirl over $S_e$ in an identical manner,
    \begin{equation} \nonumber
        (\pi \otimes \pi) \cdot \E_{S_e} \Big[ (\bs{S}_e)^\dagger_{\color{gray}X} \dyad{\pi x, \pi (y_o \lVert x_e)}{x, (y_o \lVert x_e)} (\bs{S}_e)_{\color{gray}X} \Big]  
        = (1/2^{(n/2)k}) \sum_{y_e} \dyad{\pi (x_o \lVert y_e), \pi (y_o \lVert x_e)}{(x_o \lVert y_e), (y_o \lVert x_e)},
    \end{equation}
    where $y_e \equiv \bigotimes_{a\in \text{even}} (x_a + h_{a,a+1}(x_{a+1}))$ is again uniformly random since each $x_{a+1}$ is distinct.
    In total, this shows $B''( \PD^{\text{loc}})_{\color{gray}XY} = \sum_\pi \sum_{x_o,x_e \in \text{loc-dist}} \sum_{y_o, y_e}  \dyad{\pi (x_o \lVert y_e), \pi (y_o \lVert x_e)}{(x_o \lVert y_e), (y_o \lVert x_e)} = B( \PD^{\text{loc-dist}})_{\color{gray}XY}$, as desired.
    
\vspace{3mm}
\noindent \text{(2''.)}
    The twirl of $(\PD^{\text{loc}(o)})_{\color{gray}X}$ over a tensor product of local 2-designs is near identity, $\big\lVert \mathbbm{1} - \E_{C} \big[ \bs{C}_e^\dagger \PD^{\text{loc}(o)} \bs{C}_e \big] \big\rVert_\infty \leq \tilde \delta$.
    This also holds for the twirl over any ensemble that is invariant under composition with local 2-designs (such as a Haar-random unitary on all $n$ qubits).
    
    \vspace{2mm}
    \noindent To show this, we can upper bound $\mathbbm{1} - \PD^{\text{loc}(o)} \leq \sum_{a \in \text{odd}} \sum_{1\leq i<j \leq k} \Pi^a_{ij}$ where $\Pi^a_{ij} = \sum_{z_a \in \{0,1\}^\xi} \dyad{z_a, z_a}$.
    The twirl obeys $0 \preceq \E_{C} [\bs{C} \Pi^a_{ij} \bs{C}^\dagger  ] \preceq 2/(2^\xi+1)$.
    Hence, the norm is less $\sum_a \sum_{i<j} \frac{2}{2^n+1} \leq (m/2)k^2/2^n = \tilde \delta$.

\vspace{3mm}
\noindent \text{(3''.)}
$B$ and $\tilde B$ commute with $(\PD^{\text{loc}(e)})_{\color{gray}Y}$ and $(\PD^{\text{loc}(o)})_{\color{gray}X}$. 
$B$ also commutes with $\bs{C}$.

\vspace{3mm}
\noindent \text{(4.)} For any positive semi-definite matrices $P,Q,R$ where $P$ and $Q$ commute, we have $\tr(PQR)  \leq \lVert Q \rVert_\infty \tr(P R)$.

\vspace{3mm}
\noindent \text{(5''.)} $\tr(B (\mathbbm{1}_{\color{gray}X} \otimes (\PD^{\text{loc}(e)})_{\color{gray}Y}) \dyad*{\Psi}) = \mathfrak{D}_e/2^{nk}$.
This again follows from an identical computation as before.

\vspace{4mm}
 The remainder of the proof is identical to the LRFC ensemble.
We insert $\mathbbm{1} = (\PD^{\text{loc}(o)})_{\color{gray}X} \otimes (\PD^{\text{loc}(e)})_{\color{gray}Y} + (\mathbbm{1}-(\PD^{\text{loc}(o)})_{\color{gray}X} \otimes (\PD^{\text{loc}(e)})_{\color{gray}Y})$ to decompose $\rho$ into two terms,
\begin{align*}
	\rho = \rho_{\text{dist}} + \delta\rho \equiv \E_{C} \Big[ \! \tr_{\color{gray} XY}( \tilde B (\PD^{\text{loc}(o)} \otimes \PD^{\text{loc}(e)}) \bs{C} \dyad*{\Psi} \bs{C}^\dagger ) \Big] + \E_{C} \Big[ \! \tr_{\color{gray} XY}( \tilde B (\mathbbm{1}-\PD^{\text{loc}(o)} \otimes \PD^{\text{loc}(e)}) \bs{C} \dyad*{\Psi} \bs{C}^\dagger ) \Big].
\end{align*}
\noindent We bound the second term via
\begin{align*}
    \lVert \delta \rho \rVert_1 = \tr( \delta \rho ) & = \E_{C} \big[ \text{tr}( B'' (\mathbbm{1}-\PD^{\text{loc}(o)} \otimes \PD^{\text{loc}(e)}) \bs{C} \dyad*{\Psi} \bs{C}^\dagger ) \big] && \text{(since $\delta \rho$ is positive)} \\
    & = 1 - \tr \big( B \cdot \E_{C} \big[ \bs{C}^\dagger \PD^{\text{loc}(o)} \bs{C} \big] \otimes \PD^{\text{loc}(e)} \cdot \dyad*{\Psi} \big) && \text{(since $\tr(\rho) = 1$ and (1''.))} \\
    & \leq 1 - (1-\tilde \delta) \tr \big( B (\mathbbm{1} \otimes \PD^{\text{loc}(e)}) \dyad*{\Psi} \big) && \text{(from (2''.), (3''.), (4.))} \\
    & = 1 - (1-\tilde \delta)(\mathfrak{D}_e/2^{nk}) \leq 2 \tilde \delta. && \text{(from (5''.))}
\end{align*}
Meanwhile, the first term is close to a fixed density matrix $\rho_a$,
\begin{align*}
    \rho_{\text{dist}} & = \tr_{\color{gray} XY} \!  \big( B \cdot \E_{C} \big[ \bs{C}^\dagger \PD^{\text{loc}(o)} \bs{C} \big] \otimes \PD^{\text{loc}(e)} \cdot \dyad*{\Psi} \big) && \text{(from (1''.))} \\
    & = \tr_{\color{gray} XY} \! \big( B (\mathbbm{1} \otimes \PD^{\text{loc}(e)}) \dyad*{\Psi} \big) + \Delta \equiv \rho_a + \Delta, && \text{(from (2''.), (3''.), (4.))}
\end{align*} 
where $\lVert \Delta \rVert_1 \leq \delta' \tr( B (\mathbbm{1} \otimes \PD^{\text{loc}(e)}) \dyad*{\Psi}) = \delta'(\mathfrak{D}_Re2^{nk}) \leq \tilde \delta$ from (2''.), (3''.), (4.), (5''.).

This completes the proof: We have $\lVert \rho - \rho_a \rVert_1 \leq 3\tilde \delta$ when $C$ is drawn from any 2-design.
Since both the Clifford and Haar ensembles are 2-designs, we have $\lVert \rho - \rho_H \rVert_1 \leq \lVert \rho_H - \rho_a \rVert_1 + \lVert \rho - \rho_a \rVert_1 \leq 6\tilde \delta = 3mk^2/2^\xi$. \qed

\subsection{Second approach to nearly optimal unitary designs via unitary-function gluing lemma}

It is possible to arrive at a similar construction and error bound by gluing together designs on different subsystems using only phase gates, similar to the approach in Sec.~\ref{sec:state designs}. This construction achieves the same scaling in $n$ and $k$, but uses a different analysis technique. 

We consider a modified low-depth variant of the LRFC ensemble in which left- and right-shuffle, phase, and Clifford operators are applied to each of the even patches, then an additional layer of random phases are applied to the odd patches:
\begin{definition}[Second blocked LRFC ensemble] \label{def:lower-depth-LRFC}
    Consider a system of $n$ qubits divided into nearest-neighbor patches of $\xi$ qubits each (Fig.~1). For each patch $a$, suppose $h_{a,a+1}$ is drawn uniformly randomly from functions on $\{0,1\}^{\xi} \rightarrow \{0,1\}^{\xi}$, and $f_{a,a+1}$ is drawn uniformly randomly from  binary functions on $\{0,1\}^{2\xi} \rightarrow \{0,1\}$.
    Also, for each odd $a$, suppose $C_a$ is drawn uniformly randomly from the Clifford group on $\xi$ qubits. 
    Then the low-depth Luby-Rackoff-Function-Clifford (LRFC) ensemble is given by the family of $n$-qubit random unitaries:
    \begin{equation}
        U = F_e \cdot S_o \cdot S_o \cdot F_o \cdot C_o,
    \end{equation}
    where $F_{o} = \bigotimes_{a \in \text{odd}} F_{a,a+1}$, $L_{o} = \bigotimes_{a \in \text{odd}} R_{a,a+1}$, $R_{o} = \bigotimes_{a \in \text{odd}} R_{a,a+1}$, and $C_o = \bigotimes_{a\in \text{odd}} C_a$ are tensor products of unitary operators acting on the odd pairs of patches, similar to Def. \ref{def:low-depth-LRFC}, while $F_{e}$ is a tensor product of phase operators acting on the even pairs of patches.
\end{definition}
We proceed to formulate a similar gluing bound as in the state case, in which we obtain an error which scales with the number of blocks, as well as inversely exponential in the number of qubits in the patch. Plugging in the circuit implementations of the $k$-wise independent functions again yields a depth which is doubly logarithmic in system size. 
\begin{lemma}[Gluing unitary designs with phase operators] \label{lemma:gluingunitariesf}
    Let $|A|, |B| \geq \xi$ and $k < 2^\xi$. Suppose $\mathcal{E}_A, \mathcal{E}_B$ are $\varepsilon_A$- and $\varepsilon_B$-approximate $k$-designs on $A$ and $B$, respectively, up to measurable error. Then the ensemble formed by applying a random phase operator $F_\alpha$ on $\xi$ qubits each of $A$ and of $B$ is an $\varepsilon$-approximate $k$-design with measurable error
    \begin{equation} \label{eq:gluing unitaries bound}
        \varepsilon = \varepsilon_A + \varepsilon_B + \mathcal{O}(k^2/2^\xi).
    \end{equation}
\end{lemma}
%
%
\proof Let $\mathcal{E}_{FAB}$ and $\mathcal{E}_F$ denote the unitary ensembles generated by applying random phase operators on overlapping blocks with unitaries sampled from $\mathcal{E}_A, \mathcal{E}_B$ and with Haar-random unitaries on $A, B$, respectively. Then for any choice of $W_1, W_2, \cdots, W_{k+1}$ we have that
\begin{equation}
    \begin{aligned}
        \norm{\rho_{\mathcal{E}_{FAB}} - \rho_H}_1 &\leq \norm{\rho_{\mathcal{E}_{FAB}} - \rho_{\mathcal{E_F}}}_1 + \norm{\rho_{\mathcal{E_F}} - \rho_H}_1 \\
        &\leq \varepsilon_A + \varepsilon_B + \norm{\rho_{\mathcal{E_F}} - \rho_H}_1,
    \end{aligned}
\end{equation}
where in the second step we have used both the assumption on $\mathcal{E}_A, \mathcal{E}_B$ as well as the fact that the action of any instance of a random phase operator can be absorbed into the following $W_i$. To analyze the second term, we rewrite the output state in the form
\begin{equation}
\begin{aligned}
    \rho_{\mathcal{E_F}} &= \expect_{U_A, U_B \sim H} \left[ \tr_{\color{gray} XY} \mathcal{F} (U_A \otimes U_B)^{\otimes k} |\Psi_W\rangle\langle\Psi_W|(U_A^\dagger \otimes U_B^\dagger)^{\otimes k} \right],
\end{aligned}
\end{equation}
where $\mathcal{F}$ is the unnormalized Bell state twirled over the random phases. Recall that the twirl over $F$ eliminates terms which are not of form $|\pi(xx')\rangle\langle xx'|_{\Xi_A \Xi_B}$. We again consider a projection onto both of the local distinct subspaces corresponding to $\Xi_A$ and $\Xi_B$, the regions which the phase operator acts on:
\begin{equation}
    \begin{aligned}
        \left(\rho_{\mathcal{E_F}}\right)_\dist &= \left[ \Pi_\dist^{\Xi_A} \otimes \Pi_\dist^{\Xi_B} \right] \Psi_{\mathcal{E_F}} \left[ \Pi_\dist^{\Xi_A} \otimes \Pi_\dist^{\Xi_B} \right] \\
        &= \expect_{U_A, U_B \sim H} \left[ \tr_{\color{gray} XY} \mathcal{F} \left[ \Pi_\dist^{\Xi_A} \otimes \Pi_\dist^{\Xi_B} \right] (U_A \otimes U_B)^{\otimes k} |\Psi_W\rangle\langle\Psi_W|(U_A^\dagger \otimes U_B^\dagger)^{\otimes k} \left[ \Pi_\dist^{\Xi_A} \otimes \Pi_\dist^{\Xi_B} \right] \right] \\
        &= \expect_{U_A, U_B \sim H} \left[ \tr_{\color{gray} XY} \mathcal{F} (U_A^\dagger \otimes U_B^\dagger)^{\otimes k} \left[ \Pi_\dist^{\Xi_A} \otimes \Pi_\dist^{\Xi_B} \right] (U_A \otimes U_B)^{\otimes k} |\Psi_W\rangle\langle\Psi_W| \right],
    \end{aligned}
\end{equation}
where we have used the fact that the restrictions of the symmetric subspace commute with each other and with the $U_A \otimes U_B$, similar to Sec.~\ref{sec: pf PFC}. We then observe that the twirl on the restricted subspace is the same as the restriction of the Haar twirl over the entire system:
\begin{equation}
    \expect_{U_A, U_B} \mathcal{F} (U_A^\dagger \otimes U_B^\dagger)^{\otimes k} \left[ \Pi_\dist^{\Xi_A} \otimes \Pi_\dist^{\Xi_B} \right] (U_A \otimes U_B)^{\otimes k} = \expect_{U_{AB} \sim H} \mathcal{F} (U_{AB}^\dagger)^{\otimes k} \left[ \Pi_\dist^{\Xi_A} \otimes \Pi_\dist^{\Xi_B} \right] U_{AB}^{\otimes k}.
\end{equation}
Applying the same argument as in Eq.~\ref{eq:gluing state local dist project} of the proof of Lemma \ref{lemma:gluingstates} to bound the dimension of the local distinct subspace, and using the fact that the trace norm is monotonic under physical operations, yields 
\begin{equation}
    \norm{\left(\rho_{\mathcal{E_F}}\right)_\dist - \rho_H}_1 \leq k^2/2^\xi.
\end{equation}
To finish the proof, we bound the difference between the terms
\begin{equation}
    \begin{aligned}
        \norm{\rho_{\mathcal{E_F}} - \left(\rho_{\mathcal{E_F}}\right)_\dist}_1 &= \expect_{U_A, U_B \sim H} \tr \left[ \mathcal{F} (U_A \otimes U_B)^{\otimes k}  \left[\mathbbm{1} - \left(\Pi_\dist^{\Xi_A} \otimes \Pi_\dist^{\Xi_B}\right) \right] (U_A^\dagger \otimes U_B^\dagger)^{\otimes k} |\Psi_W\rangle\langle\Psi_W| \right] \\
        &= 1- \tr \left( \rho_{\mathcal{E_F}} \right)_\dist = 1 - \tr \left(\Pi_\dist^{\Xi_A} \otimes \Pi_\dist^{\Xi_B}\right) \Psi_H \\
        &\leq k^2/2^\xi,
    \end{aligned}
\end{equation}
Collecting error terms up to leading order in $k, \xi$ yields the bound in Eq.~\ref{eq:gluing unitaries bound}. \qed
\\

Applying this gluing formula repeatedly for each additional block in the system yields the following bound for the blocked LRFC construction:
\begin{theorem}[Second blocked LRFC is a design] \label{thm:blocklfrc}
    Let $\mathcal{E}_{n,\xi}$ be the family of blocked $k$-wise independent LRFC circuits on $n$ qubits with blocks of size $\xi \leq n$. Then $\mathcal{E}_{n,\xi}$ forms an $\varepsilon$-approximate $k$-design with
    \begin{equation}
        \varepsilon = \mathcal{O}(nk^2/2^\xi \xi).
    \end{equation}
\end{theorem}
\proof From Lemma \ref{lemma:kwisephasestate}, applying the initial layers of Clifford, right shuffle, phase, and left shuffle gates yields approximate $k$-designs on each block, up to error $\mathcal{O}(k^2/2^\xi)$. The result follows by starting from the leftmost block and repeatedly applying the result of Lemma \ref{thm:blocklfrc} after adding of the additional $n/2\xi - 1$ blocks. \qed
\\

\noindent As a corollary, we immediately obtain the same scaling as the main result of this section:
\begin{corollary}[Unitary designs in $\log k \log \log nk/\varepsilon$ depth] \label{cor:lowdepthunitary}
    Let $\mathcal{E}_n$ be an ensemble of $k$-wise independent blocked LRFC circuits on $n$ qubits with block size $\xi = \mathcal{O}(\log n)$. Then $\mathcal{E}_n$ is an $\varepsilon$-approximate design with $\varepsilon = \mathcal{O}(nk^2/2^\xi \xi)$. In addition, each $U \sim \mathcal{E}_n$ can be implemented with a circuit of depth $\mathcal{O}(\log k \cdot \log \log nk/\varepsilon)$ with $k \cdot \tilde{\mathcal{O}}(n)$ ancillas. 
\end{corollary}
\proof This follows from substituting the constructions in Lemma \ref{lemma:kwisefns} and \cite{cleve2015near} for $k$-wise independent functions and exact unitary $2$-designs, respectively, into the result of Theorem \ref{thm:blocklfrc}, then choosing fixed $\varepsilon > 0$ and letting $\xi = \mathcal{O}(\log n)$ for large $n$. \qed

\subsection{Third approach to nearly optimal unitary designs via relative error gluing lemma}

Here, we present a formal proof of the alternative construction of unitary $k$-designs discussed in the main text, which builds heavily on the existing results of~\cite{schuster2024random} and~\cite{SRU2025}.
The unitary ensemble corresponds to a two-layer circuit in which each small random unitary is drawn randomly from the LRFC ensemble on $2\xi$ qubits. 
\begin{definition}[Third blocked LRFC ensemble]
    Consider a system of $n$ qubits divided into  patches of size $\xi$. Then the third LRFC ensemble is given by the family of $n$-qubit random unitaries:
    \begin{equation}
        U = (LRFC)_e \cdot (LRFC)_o,
    \end{equation}
    where $(LRFC)_e$ and $(LRFC)_o$ are tensor products of unitaries drawn from the LRFC ensemble on even and odd pairs of patches, respectively.
\end{definition}
Since the unitaries applied in each of the patches for both layers is drawn from an approximate design, it is possible to show that this construction also forms a design by applying the relative error gluing result (Theorem~1 of~\cite{schuster2024random}) to neighboring patches:
\begin{theorem}[Third blocked LRFC is a design]\label{eq: uniform two layer lrfc bound}
    The third blocked LRFC ensemble on $n$ qubits forms an approximate unitary $k$-design with measurable error $\mathcal{O}(nk^2/2^\xi)$.
\end{theorem}
\proof From Theorem~\ref{thm:LRFC-designs}, the LRFC ensemble on each individual patch forms an approximate $k$-design up to measurable error $\mathcal{O}(k^2/2^{\xi})$. Hence, the third blocked LRFC ensemble is indistinguishable from the blocked Haar-random circuit up to measurable error  $nk^2/2^\xi \xi$. From Theorem~1 of~\cite{schuster2024random}), the blocked Haar-random circuit is a unitary $k$-design with relative error $nk^2/2^\xi$.
This implies that it is a design with measurable error $2nk^2/2^\xi$ (Lemma~\ref{lemma:rel error to measurable}). \qed

\subsection{Application to pseudorandom unitaries} \label{sec: PRU}

As previously discussed in the case of state designs, our simple constructions and proofs of random unitary ensembles with small measurable error also have application to the study of pseudorandom unitaries (PRUs).
A PRU is a unitary ensemble that is indistinguishable from Haar-random by any bounded-time quantum experiment.
We refer again to Refs.~\cite{schuster2024random,ma2024construct} for a pedagogical introduction.
%

The existence of pseudorandom unitaries has remained an open question until very recently.
This  was resolved by~\cite{ma2024construct}, which proved that the PFC ensemble forms a PRU whenever $P$ and $F$ are instantiated with a pseudorandom permutation (PRP)~\cite{zhandry2016note} and pseudorandom function (PRF)~\cite{zhandry2021PRF}.
This built upon earlier work, which introduced the PFC ensemble and proved that it formed a PRU with a weaker parallel form of security~\cite{metger2024simple}.
Following this progress, the LRFC ensemble was introduced and also proven to form a PRU when each function is pseudorandom~\cite{SRU2025}.

So far, both proofs of the existence of PRUs (for the PFC ensemble~\cite{ma2024construct} and the LRFC ensemble~\cite{SRU2025}) work within the so-called path-recording framework introduced by~\cite{ma2024construct}.
Once instantiated, this framework is quite powerful.
For example, it allows one to prove that (slight variants of) the PFC and LRFC ensembles are secure even against adversaries that can query both a random unitary $U$ and its inverse $U^\dagger$ many times.
Such experiments are much more powerful than standard experiments that only query $U$~\cite{schuster2023learning,cotler2023information,schuster2024random}.
At the same time, however, the path-recording framework can yield somewhat lengthy proofs.
This arises both from a large amount of initial notation, and the repeated manipulation of large sums of bitstrings throughout the analysis.
Given the fundamental importance of PRUs, one might wonder whether a more direct proof is possible.

This is in fact provided by our work.
Our proofs that the PFC, LRFC, and blocked LRFC ensembles have small measurable error immediately imply that each ensemble is a PRU (when its functions and permutations are instantied pseudorandomly).
By definition, a PRP and PRF are indistinguishable from a  truly random permutation and function by any bounded-time  experiment. 
Our analysis then proves that the truly random ensembles are indistinguishable from Haar-random unitaries. 
Together, this implies that the pseudorandom ensembles are each PRUs. 
Our proofs for each ensemble are extremely short and direct (see Sections~\ref{sec: pf PFC},~\ref{sec: pf LRFC}, and~\ref{sec: pf blocked LRFC}).
They are also self-contained, relying only on basic properties of unitary 2-designs and the permutation group. 
We hope that this alternative approach can complement existing methods for future work on PRUs and unitary designs.




\section{Unitary designs with small relative error} \label{sec:relerrordesigns}

In this section, we present a variant of the blocked LRFC circuit that forms an approximate unitary $k$-design with small relative error.
The circuit has depth $\mathcal{O}(k \log k \log \log nk/\varepsilon)$, which is larger by only a factor of $k$ compared to our earlier results. 
As aforementioned, the best previous construction of unitary designs with relative error required circuit depths $\mathcal{O}(k \poly \log k \log n k / \varepsilon)$~\cite{schuster2024random,laracuente2024approximate}.
This features an exponentially worse scaling in the number of qubits~$n$ and the inverse error $1/\varepsilon$ compared to our new results.
In what follows, we first describe our construction and then provide a proof that it achieves small relative error.

\subsection{Blocked amplified LRFC circuits}

To describe our relative error unitary ensemble, consider the circuit formed by applying $p$ independently random LRFC unitaries on the same $2\xi$ qubits in succession.
We term this the (LRFC)$^p$ ensemble.
Our relative error unitary ensemble corresponds to the two-layer circuit in which each block is drawn independently randomly from the (LRFC)$^p$ ensemble on $2\xi$ qubits.
\begin{definition}[Blocked amplified LRFC circuit] \label{def:amplified lrfc}
    Consider a system of $n$ qubits divided into  patches of size $\xi$. The $p$-layer blocked amplified LRFC circuit is given by the family of $n$-qubit random unitaries:
    \begin{equation}
        U = (LRFC)^p_e (LRFC)^p_o,
    \end{equation}
    where $(LRFC)^p_e$ is a tensor product of (LRFC)$^p$ random unitaries on every even patch of qubits, and $(LRFC)^p_o$ is a tensor product of (LRFC)$^p$ random unitaries on every old patch of qubits.
\end{definition}
The ``amplification'' in the name of our ensemble comes from the \emph{amplification of approximation error} result for unitary designs \cite{brandao2016local}, which states that the additive error decreases multiplicatively when designs are composed. 
For completeness, we reproduce the result in the following subsection, before using it to show that the blocked amplified LRFC ensemble is a relative error design for $p = \mathcal{O}(k)$ and $\xi \leq n/2$:
\begin{theorem}[Blocked amplified LRFC circuits are relative error designs] \label{thm:low-depth-repeated-LRFC-designs}
    For any $k \leq 2^{\xi/4}$ and $p \geq 8k+1$.
    The (LRFC)$^p$ ensemble forms an approximate unitary $k$-design on $n$ qubits with relative error $\varepsilon = \mathcal{O}(nk^2/2^{\xi})$.
\end{theorem}
The circuit depth of this construction is $p$ times the depth of the original blocked LRFC construction from Fact~\ref{fact:blocked-LRFC-resources}:
\begin{fact}[Circuit depth of blocked amplified LRFC circuit]
    The blocked amplified LRFC circuit can be implemented in depth $\mathcal{O}(p \cdot \log k \cdot \log \xi)$.
\end{fact}
Hence, by taking $p = 8k + 1$ and $\xi = \mathcal{O}(n k^2 / \varepsilon)$, we can achieve approximate unitary $k$-designs with relative error $\varepsilon$ in depth $\mathcal{O}(k \cdot \log k \cdot \log \log (n k / \varepsilon))$. This improves the scalings with $n$ and $\varepsilon$ compared to known results~\cite{schuster2024random, laracuente2024approximate}.

\subsection{Proof of Theorem~\ref{thm:low-depth-repeated-LRFC-designs}: The blocked amplified LRFC ensemble has small relative error}

Our proof is extremely short.
From Lemma~\ref{lemma: add error square} and Lemma~\ref{lemma:amplified lrfc block} below, the (LRFC)$^p$ circuit on $2\xi$ qubits forms a design with relative error $\varepsilon = 2k^2/2^\xi$ for $p \geq 8k+1$. 
 We then apply the relative error gluing lemma (Theorem~1 of~\cite{schuster2024random}) to obtain a design on $n$ qubits with relative error $(n/\xi)2k^2/2^\xi + nk^2/2^\xi$. This completes the proof. \qed

\begin{lemma}[Error decays] \label{lemma: add error square}
    If $\Phi_{\mathcal{E}}$ has additive error $\varepsilon$, then $\Phi_{\mathcal{E}} \circ \Phi_{\mathcal{E}}$ has additive error at most $\varepsilon^2$.
\end{lemma}

\begin{proof}
    We write $\Phi_{\mathcal{E}} = \Phi_H + \delta \Phi$, where $\Phi_H \circ \delta \Phi = \delta \Phi \circ \Phi_H = 0$ by definition.
    Squaring gives, $\Phi_{\mathcal{E}} \circ \Phi_{\mathcal{E}} = \Phi_H + \delta \Phi \circ \delta \Phi$.
    By assumption, we have $\lVert \delta \Phi(\rho) \rVert_1 \leq \varepsilon$ for any $\rho$.
    We can break $\delta \Phi(\rho) =  \rho_+ -  \rho_-$ into a positive and negative part, which have normalizations $\varepsilon_\pm \equiv \lVert \rho_\pm \rVert_1 = \tr(\rho_\pm)$.
    Since $\rho_+$ and $\rho_-$ have support on orthogonal subspaces, we have $\varepsilon_+ + \varepsilon_- = \lVert \delta \Phi(\rho) \rVert_1 \leq \varepsilon$ from our previous bound.
    To proceed, we write $\delta \Phi( \delta \Phi (\rho) ) = \delta \Phi(\rho_+ ) -  \delta \Phi( \rho_-)$.
    Since $\rho_\pm$ are positive density matrices with traces $\varepsilon_\pm$, we have $\lVert \delta \Phi(\rho_\pm) \rVert_1 \leq \varepsilon \lVert \rho_\pm \rVert_1 = \varepsilon \varepsilon_\pm$.
    This yields $\lVert \delta \Phi(\delta \Phi(\rho)) \rVert_1 \leq \varepsilon (\varepsilon_+ + \varepsilon_-) \leq \varepsilon^2$, which completes the proof.
\end{proof}

\begin{lemma}[Amplification] \label{lemma:amplified lrfc block}
    For any $k \leq 2^{n/8}$ and $p \geq 8k+1$. The $(\text{\emph{LRFC}})^p$ ensemble is an approximate unitary $k$-design with relative error $\varepsilon = 2k^2/\sqrt{2^{n}}$.
\end{lemma}
\begin{proof}
    From Theorem~\ref{thm:LRFC-designs}, the LRFC ensemble forms an approximate unitary $k$-design with additive error $4 k^2/2^{n/2}$.
    Iterating Lemma~\ref{lemma: add error square} $p-1$ times, we find that the $p$-th composition of the LRFC ensemble with itself forms an approximate unitary $k$-design with additive error $\varepsilon_{\text{add}} = (k^2/2^{n/2})^p$.
    From Lemma~\ref{lemma:add to rel error}, any approximate unitary $k$-design with additive error $\varepsilon_{\text{add}}$ is also an approximate unitary $k$-design with \emph{relative} error $\varepsilon = (4^{nk}/k!)(1+k^2/2^n) \varepsilon_{\text{add}}$. Setting $p \geq 8k+1$ and using $k! \geq 1$ yields $\varepsilon \leq 4^{nk}(1+k^2/2^n) (k^{16k+2}/4^{2nk+n/4}) = (k^2/\sqrt{2^n})(1+k^2/2^n)(k^{8}/2^{n})^{2k}$.
    When $k \leq 2^{n/8}$, the final term is less than one  and the second term is less than two.
\end{proof}

\section{Improved lower bounds on state and unitary designs}

In this section, we provide the proof of our improved lower bounds on the circuit depths required for state and unitary designs.
The full statement of our theorem is as follows.
\begin{theorem}\label{thm: lower bound design}
    {\emph{(Lower bound for state and unitary designs with additive error)}}
    For any $1 < k = o(2^{n})$. Any random state or unitary ensemble over $n$ qubits that forms a $k$-design with additive error $\varepsilon$ requires circuit depth at least:
    \begin{itemize}
    \item $d = \Omega ( k + \log n / \varepsilon )$, for 1D circuits with $\mathcal{O}(n)$  ancilla qubits and $d = o(n)$,
    \item $d = \Omega (  k + \log \log n / \varepsilon )$, for all-to-all circuits with $\mathcal{O}(n)$ ancilla qubits and $d = o(\log n)$.
    \item $d = \Omega ( \log k + \log \log n / \varepsilon )$, for all-to-all circuits with any number of ancilla qubits and $d = o(\log n)$. 
    \end{itemize}
\end{theorem}
\noindent We assume that every all-to-all circuit in the ensemble has the same  architecture. The restriction to $k = o(2^n)$ is necessary since the a Haar-random unitary forms an exact $k$-design for any $k$ and can be compiled in depth $\mathcal{O}(2^{2n})$ in any circuit geometry~\cite{brandao2016local}.
The restriction to $d = o(n)$ (in 1D) and $d = o(\log n)$ (in all-to-all systems) is necessary since the Clifford group forms an exact unitary 3-design and can be compiled in depth $\mathcal{O}(n)$ in 1D and $\mathcal{O}(\log n)$ in all-to-all circuits.


Before providing the proof of Theorem~\ref{thm: lower bound design}, let us first review the best known previous lower bounds for state and unitary designs.
For clarity, we provide separate discussions regarding the dependence on each of $k$, $n$, and $\varepsilon$.
\begin{itemize}
    \item For the $k$-dependence, the best existing lower bound is Proposition~8 of Ref.~\cite{brandao2016local}, which states that any additive error state or unitary design requires a circuit of size at least $\Omega(nk/(\log nk))$. This lower bound is widely believed to be optimal, and has been proven to be optimal up to poly-logarithmic factors~\cite{chen2024incompressibility}. 
    
    Our lower bound on the $k$-dependence simply re-phrases this known lower bound on the circuit size in terms of the circuit depth.
    We provide full details in our proof below.
    
    \item For the $n$-dependence, the only existing lower bound is Proposition~1 of Ref.~\cite{schuster2024random}. This applies  to state/unitary designs with relative error, and yields depth lower bounds of $d = \Omega(\log n)$ in 1D circuits and $d = \Omega(\log \log n)$ in all-to-all connected circuits. No lower bound existed for state/unitary designs with additive error.

    To this end, the primary achievement of our lower bounds is to show that these same circuit depths are also needed for designs with additive error, and not just relative error. This extension is highly non-trivial, since many features of relative error designs are incredibly fine-grained and not efficiently measurable in experiments.
    
    \item For the $\varepsilon$-dependence, a basic argument involving a SWAP test on half the system leads to a lower bound of $d = \Omega(\log 1 / \varepsilon)$ in 1D circuits.
    This applies to both state and unitary designs with additive  error.
    In all-to-all circuit geometries, the only existing lower bound is Proposition~2 of Ref.~\cite{schuster2024random}. This applies only to unitary designs, and yields a depth lower bound of $d = \Omega(\log \log 1/\varepsilon)$ for any unitary design with additive error.

    Our Theorem~\ref{thm: lower bound design} extends the lower bound for all-to-all-connected circuits to state designs.
\end{itemize}
\noindent We note that more restricted circuit classes may feature stronger lower bounds that are not reflected above. For example, any unitary ensemble in which each unitary is composed of independent Haar-random gates requires circuit depth at least $d = \Omega(\log n)$ to form relative error designs \footnote{The lower bound follows because random circuits require at least logarithmic depth to anti-concentrate~\cite{aharonov2023polynomial}. Anti-concentration is necessary in order to form a design with small relative error. We expect that this lower bound can be extended to apply to additive error designs using the techniques we introduce in the proof of Theorem~\ref{thm: lower bound design}.} and $d = \Omega(\log 1/ \varepsilon)$ to form additive error designs~\cite{fefferman2024anti}, even in all-to-all connected circuit geometries.

Let us now turn to the proof of Theorem~\ref{thm: lower bound design}. As discussed above, the dependence on $k$ follows from Proposition 8 in Ref.~\cite{brandao2016local}. Hence, the bulk of our proof will focus on the $n$ and $\varepsilon$ dependence.
\begin{proof}[Proof of Theorem~\ref{thm: lower bound design}, $k$-dependence]
From the analysis in Proposition 8 of \cite{brandao2016local}, any ensemble $\mathcal{S}$ that forms an  state $k$-design with additive error $\varepsilon < 1/4$ must contains at least $\Omega(nk/\log(nk))$ gates. If the number of ancilla qubits is bounded by $\mathcal{O}(n)$, then the circuit size can grow at most linearly in the circuit depth $d$, $\mathcal{O}(n \cdot d)$. Hence, we require $d = \Omega(k/\log(nk))$. If the number of ancilla qubits is unbounded, on the other hand, the light-cone can grow exponentially in the circuit depth (for circuits with arbitrary long-range two-qubit gates).
Hence, the total number of gates can grow as $\mathcal{O}(n \cdot \exp(d))$, which implies that the circuit depth must be at least $\Omega(\log k)$.
\end{proof}

\begin{proof}[Proof of Theorem~\ref{thm: lower bound design}, $n$- and $\varepsilon$-dependence]
    Let $\mathcal{S}$ denote the state ensemble.
    We write $\ket{\psi} = U \ket{0^n}$ for each $\ket{\psi} \sim \mathcal{S}$, where $U$ is a circuit of depth at most $d$.
    We let $L$ denote the maximum size of any light-cone in $U$.
    We have $L \leq 2^d$ for all-to-all connected circuits and $L \leq 2d$ for 1D circuits.

    To set up our proof, let us first divide the $n$ qubits into isolated patches as follows.
    Each patch will correspond to the forward light-cone of a given qubit.
    We additionally will demand that the backwards light-cones of any two different patches do not overlap.
    Mechanistically, we place the patches as follows.
    We set the first patch by choosing any qubit and letting the first patch consist of all $L$ qubits in its light-cone.
    We then trace the backwards light-cone of the first patch (which has size less than $L^2$), then trace the forwards light-cone of this (which has size less than $L^3$), and then finally trace the backwards light-cone of this (which has size less than $L^4$).
    We then set the second patch by choosing any qubit that is outside this backwards-forwards-backwards light-cone.
    There are at least $n - L^4$ possible choices.
    The second patch is equal to the forwards light-cone of this qubit.
    We then set the third patch by tracing out the backwards-forwards-backwards light-cone of the union of the first and second patch, and setting the third patch to be the forwards light-cone of any qubit outside of this.
    There are at least $n - 2L^4$ possible such choices.
    Iterating in this manner, we can place at least $M \equiv n/L^4$ patches across the entire system whose backwards light-cones do not overlap.
    Each patch contains $L$ qubits and the forward light-cone of at least one qubit.



    With this setup in hand, we can now prove our lower bound.
    We consider the following quantum experiment.
    We draw two copies of a state $\ket{\psi}$ from either the ensemble $\mathcal{S}$ or the Haar ensemble $H$.
    We then perform random single-qubit rotations, $v = \bigotimes_{i=1}^n v_i$, on each qubit of the system.
    Averaging over these rotations yields the state,
    \begin{equation} \label{eq: dephased}
        \E_v \left[ \big( v \dyad{\psi} v^\dagger \big)^{\otimes 2} \right] 
    \end{equation}
    We then measure (i.e.~dephase) the state in the computational basis.
    This yields the state,
    \begin{equation}
        \sum_{x,y \in \{0,1\}^n} \bra{x \otimes y} \E_v \left[ \big( v \dyad{\psi} v^\dagger \big)^{\otimes 2} \right]  \ket{x \otimes y} \cdot \dyad{x \otimes y}.
    \end{equation}
    Note that one could place the expectation over $v$ either inside or outside the inner product in the first term.
    We then forget all information about the measurement outcomes except for whether the bitstring $x$ and the bitstring $y$ agreed on all of their bits in each individual patch.
    Let $z(x,y) \in \{0,1\}^M$ denote the resulting bitstring, where $z_a$ indicates where $x$ and $y$ are equal on all bits in patch $a$.
    That is, $z_a = 1$ if $x_i = y_i$ for all qubits $i$ in patch $a$, and $z_a = 0$ if $x_i \neq y_i$ for any qubit $i$ in patch $a$.
    We can represent the resulting state as,
    \begin{equation} \label{eq: z state}
        \chi_\psi \equiv \sum_{z \in \{0,1\}^M} p_\psi(z) \dyad{z},
    \end{equation}
    where $p_\psi(z) = \sum_{x,y : z(x,y) = z} p_\psi(x,y)$, and $p_\psi(x,y)$ denotes the first term in the sum in Eq.~(\ref{eq: dephased}).
    Finally, we conclude the experiment by counting whether the total number of collisions, $s = |z|$, is greater or less than a fixed threshold $s^*$.
    That is, we measure the observable,
    \begin{equation} \label{eq: A}
        A = \sum_{z : |z| > s^*} \dyad{z} - \sum_{z : |z| \leq s^*} \dyad{z},
    \end{equation}
    where a positive measurement outcome will be associated to the ensemble $\mathcal{S}$, and a negative measurement outcome to the Haar ensemble.
    We will set the variable $s^*$ later on in the proof.

    The state $\chi_\psi$ in Eq.~(\ref{eq: z state}) was obtained from $\ket{\psi}$ by a sequence of quantum channels.
    Hence, if the ensemble $\mathcal{S}$ forms an approximate 2-design with additive error $\varepsilon$, we must have
    \begin{equation} \label{eq: chi eps}
        \left\lVert \chi_{\mathcal{S}} - \chi_H \right\rVert_1 \leq \varepsilon,
    \end{equation}
    where $\chi_{\mathcal{S}} = \E_{\psi \sim \mathcal{S}} \left[ \chi_\psi \right]$ and $\chi_H = \E_{\psi \sim H} \left[ \chi_\psi \right]$.
    By Holder's inequality, we must also have,
    \begin{equation}
        \left| \tr( A \chi_\mathcal{S} ) - \tr( A \chi_H ) \right| \leq \varepsilon.
    \end{equation}
    The strategy of our proof is to lower bound the left hand side above.
    This immediately lower bounds the achievable additive error $\varepsilon$ of any state or unitary design. 

    Let us first consider the Haar ensemble.
    The second moment of the Haar ensemble is given by,
    \begin{equation}
        \E_{\psi \sim H} \left[ \dyad{\psi}^{\otimes 2} \right] = \frac{1}{2^n(2^n-1)} \mathbbm{1} + \frac{1}{2^n(2^n-1)} \mathcal{S},
    \end{equation}
    where $\mathcal{S}$ is the SWAP operator.
    This leads to an (expected) output distribution
    \begin{equation}
        \E_{\psi \sim H} \big[ p_\psi(x,y) \big] = \frac{1}{2^n(2^n-1)}  + \frac{1}{2^n(2^n-1)} \delta_{x = y}.
    \end{equation}
    The total variational (i.e.~1-norm) distance between this distribution and the maximally mixed distribution is,
    \begin{equation}
        \sum_{x,y} \left| \E_{\psi \sim H} \big[ p_\psi(x,y) \big] - \frac{1}{4^n} \right| = 2(4^n-2^n)\left( \frac{1}{4^n} - \frac{1}{2^n(2^n+1)} \right) = \frac{2}{2^n} \frac{1-1/2^n}{1+1/2^n} \leq \frac{2}{2^n},
    \end{equation}
    which is exponentially small in $n$.
    Hence, for the purposes of this proof we can treat $\E_{\psi \sim H} \big[ p_\psi(x,y) \big]$ as maximally mixed.
    This yields a collision distribution,
    \begin{equation}
        \tilde{p}_H(z) = (1-1/2^L)^{|z|} (1/2^L)^{M-|z|},
    \end{equation}
    in which each patch features a collision with an independent probability $q_H = 1/2^L$.
    Here, we let $\tilde{p}_H(z)$ denote the collision distribution drawn from the flat distribution over $x$ and $y$, which is close to the exact Haar collision distribution $p_H(z)$ up to total variation distance $\sum_z |\tilde{p}_H(z)-p_H(z)| \leq 2/2^n$.

    We can now turn to the ensemble $\mathcal{S}$.
    Let us first notice a key fact: For each fixed state $\ket{\psi}$, the distribution $p_\psi(z)$ is a tensor product distribution.
    This follows because the patches are backwards-light-cone-separated, and in going from $p_\psi(x,y)$ to $p_\psi(z)$ we have traced out all qubits that are outside every patch.
    (We note that this does not necessarily hold after averaging over $\ket{\psi}$, since the ensemble $\mathcal{S}$ might feature long-range correlations between the random gates in different patches.)
    Hence, we can write
    \begin{equation}
        p_\psi(z) = \prod_{a = 1}^M ( 1 - q_\psi^a )^{z_a} ( q_\psi^a )^{1 - z_a},
    \end{equation}
    where $q_\psi^a$ is the probability to observe a collision in patch $a$.

    Our next step is to lower bound each $q_\psi^a$.
    To do so, we follow a similar approach to the proof of Proposition~1 in Ref.~\cite{schuster2024random}.
    Let 
    \begin{equation}
        \rho_a = \tr_{\bar a}(U \dyad{0^n} U^\dagger)
    \end{equation}
    denote the reduced density matrix of the system on patch $a$ after $U$ is applied. 
    The quantity $q_\psi^a$ measures the collision probability when $\rho_a$ is measured in a random single-qubit basis, 
    \begin{equation}
        q^a_\psi = \E_{v} \left[ \sum_{x_a \in \{0,1\}^L} \bra{x_a} v_a \rho_a v_a^\dagger \ket{x_a}^2 \right],
    \end{equation}
     where $v_a$ is the restriction of the random single-qubit rotations $v$ to patch $a$.
    A standard formula for the twirl over a tensor product of single-qubit 2-designs~\cite{elben2023randomized} allows one to compute
    \begin{equation}
    	q^a_\psi = \sum_{k = 0}^n   3^{-k} \tr_a \big( \rho_a \cdot \mathcal{P}_{k} \big\{ \rho_a \big\}  \big),
    \end{equation}
    where $\mathcal{P}_{k}$ is a super-operator that projects onto Pauli operators with weight $k$.

    Now that we have an explicit formula for $q^a_\psi$, we can proceed with our lower bound.
    For each patch $a$, we re-write the initial state of the system as a sum of three operators,
    \begin{equation} \label{eq: decompose state}
    	\dyad{0^n} = \frac{1}{2^n} \mathbbm{1} + \rho^a_{1} + \rho^a_{>},
    \end{equation}
    where $\rho^a_1$ consists solely the single-qubit $Z$-stabilizer whose light-cone is in $a$, and $\rho^a_{>}$ consists of all other non-identity stabilizers.
    Note that $U \rho^a_1 U^\dagger$ contains operators at most of weight $L$, all of which have support entirely in patch $a$.
    We then consider the total support of $\rho_a$ on all weights between 1 and $L$, 
    \begin{equation}
    	P^a_L \equiv \sum_{k=1}^L \tr_a( \rho_a \cdot \mathcal{P}_k \{ \rho_a \} ),
    \end{equation}
    where again, $\rho_a = \tr_{\bar a}( U \dyad{0^n} U^\dagger )$.
    We can now insert the decomposition Eq.~(\ref{eq: decompose state}) into this expression. 
    Since $\rho^a_{1}$ and $\rho^a_{>}$ are orthogonal before the application of $U$, i.e.~$\tr( \rho^a_1 \rho^a_> ) = 0$, they are also orthogonal after the application of $U$.
    Moreover, they remain orthogonal after tracing out $\bar a$ and projecting to weights $k \leq L$, since $\rho_1$ only has support on $a$ and such weights.
    This allows us to lower bound $P_L$ by focusing solely on the $\rho^a_1$ term,
    \begin{equation}
    	P^a_L =  \tr_a( (\rho^a_1)^2 ) + 
    	\sum_{k=0}^L \tr_a( U \rho^a_> U^\dagger \cdot \mathcal{P}_k \{ U \rho^a_> U^\dagger \} ) \geq \frac{1}{2^L},
    \end{equation}
    since the second term is non-negative and the first term is equal to $1/2^L$, since $\rho^a_1$ contains a single stabilizer. 
    Note that the identity term in Eq.~(\ref{eq: decompose state}) does not contribute to $P^a_L$, since it does not have weight between $1$ and $L$.
    This lower bound on $P^a_L$ in turn lower bounds the expected collision probability, 
    \begin{equation}
    	q^a_\psi = \sum_{k = 0}^n   3^{-k} \tr_a \big( \rho_a \cdot \mathcal{P}_{k} \big\{ \rho_a \big\}  \big) \geq \frac{1}{2^L} + 3^{-L} P^a_L \geq \frac{1}{2^L} \left( 1 + \frac{1}{3^L} \right),
    \end{equation}
    where the first bound follows because every term in the sum is positive.
    This completes our lower bound on $q^a_\psi$.

    We can now specify the observable $A$ and analyze its behavior.
    Consider a fictional Bernoulli process $p_f(z)$ in which each patch has an independent probability of collision, $q_f = (1/2^L) ( 1 + 1/3^L )$.
    This is equal to our lower bound on the collision probabilities $q_\psi^a$.
    We will construct $A$ as the optimal observable to distinguish this fictional Bernoulli process from the Haar-random collision distribution, which corresponds to a Bernoulli process with collision probabilities, $q_H = 1/2^L$.
    The optimal observable is given by,
    \begin{equation}
        A = \sum_{z \in \{0,1\}^M} \text{sgn}\big( p_f(z) - \tilde{p}_H(z) \big) \cdot \dyad{z}.
    \end{equation}
    To determine the value of the sign function, we can compute the ratio of the two probability distributions,
    \begin{equation}
        \frac{p_f(z)}{\tilde{p}_H(z)} = \left( \frac{1-(1/2^L)(1+1/3^L)}{1-1/2^L} \right)^{|z|} \left( \frac{(1/2^L)(1+1/3^L)}{1/2^L} \right)^{M-|z|}.
    \end{equation}
    We define $s^*$ as the value of $|z|$ at which this ratio is equal to one.
    For $|z| > s*$, we have $p_f(z) > \tilde{p}_H(z)$, and for $|z| \leq s*$, we have $p_f(z) \leq \tilde{p}_H(z)$.
    This confirms our formula for $A$ in Eq.~(\ref{eq: A}).

    We will now show that $A$ allows one to distinguish $p_f(z)$ and $p_H(z)$. 
    In the next paragraph, we will use this result to show that $A$ also allows one to distinguish $p_\mathcal{S}(z)$ and $\tilde{p}_H(z)$ for any sufficiently low-depth ensemble $\mathcal E$.
    We have
    \begin{equation}
        \langle A \rangle_f - \langle A \rangle_H \geq \sum_{z \in \{0,1\}^M} \big| p_f(z) - \tilde{p}_H(z) \big| - 2/2^n,
    \end{equation}
    where the latter term accounts for the difference between the exact Haar distribution $p_H(z)$ and $\tilde{p}_H(z)$, which is small.
    The first term is equal to the total variation distance $\text{TVD}(p_f,\tilde{p}_H)$ between the two distributions.
    Since both distributions are a tensor power of the same distribution across all patches, we can apply the standard inequalities (see e.g.~Section 13 of Ref.~\cite{o2021learning}),
    \begin{equation} \label{eq: TVD inequality}
        \text{TVD}\big(p_1^{\otimes M},p_2^{\otimes M}\big) \geq \frac{1}{2} d_H\big(p_1^{\otimes M},p_2^{\otimes M}\big)^2 =  1 - \left(1-\frac{1}{2} d_H(p_1,p_2)^2 \right)^M \geq  1 - \left(1-\frac{1}{2}\text{TVD}(p_1,p_2)^2 \right)^M,
    \end{equation}
    where $d_H(p_1,p_2)$ is the Hellinger distance. 
    From this, we find
    \begin{equation}
        \langle A \rangle_f - \langle A \rangle_H +2/2^n \geq \text{TVD}(p_f,\tilde{p}_H) \geq 1 - (1 - |q_f-q_H|^2)^M = 1 - \left( 1 - \left( \frac{1}{2^L 3^L} \right)^2 \right)^M.
    \end{equation}
    Let $x \equiv (1/(2^L 3^L))^2$. We can bound $(1-x)^M \leq e^{-Mx} \leq 1 - M x/(2\log 2)$, where the second inequality holds for $xM \leq \log 2$.
    This yields,
    \begin{equation} \label{eq: fH inequality}
        \langle A \rangle_f - \langle A \rangle_H \geq \frac{M}{\log 2} \left( \frac{1}{2^L 3^L} \right)^2  - 2/2^n.
    \end{equation}
    This is our final lower bound on $\langle A \rangle_f - \langle A \rangle_H$.

    Let us now turn to the ensemble $\mathcal E$.
    We will show that $\langle A \rangle_\mathcal{S} \geq \langle A \rangle_f$, hence $\langle A \rangle_f - \langle A \rangle_H$ obeys the same lower bound as in Eq.~(\ref{eq: fH inequality}). 
    To see this, consider $\langle A \rangle_\psi$ for an individual state $\ket{\psi} \sim \mathcal{S}$.
    As discussed, the distribution $p_\psi(z)$ is a tensor product of binary variables with probabilities at least $q^a_\psi$ in each patch.
    Consider the expectation value $\langle A \rangle(\{q^a_\psi\})$ as a function of these probabilities.
    We have
    \begin{equation}
    \begin{split}
        \partial_{q^a_\psi} \langle A \rangle(\{q^a_\psi\}) & = \partial_{q^a_\psi} \left( 2 \sum_{z:|z| > s^*} p_\psi(z) - 1 \right) \\
        & = 2 \partial_{q^a_\psi} \left( (1-q^a_\psi) \sum_{|z_{\bar a}| > s^*} p_\psi(z_{\bar a}) + q^a_\psi \sum_{|z_{\bar a}| > s^*-1} p_\psi(z_{\bar a}) \right) \\
        & = 2 \left( \sum_{|z_{\bar a}| > s^*-1} p_\psi(z_{\bar a}) - \sum_{|z_{\bar a}| > s^*} p_\psi(z_{\bar a}) \right) \geq 0. \\
    \end{split}
    \end{equation}
    Hence, $\langle A \rangle(\{q^a_\psi\})$ is monotonically increasing as a function of each $q^a_\psi$.
    The inequality $q^a_\psi \geq q_f$ then implies $\langle A \rangle_\mathcal{S} = \E_{\psi \sim \mathcal E}[ \langle A \rangle_\psi ] \geq \E_{\psi \sim \mathcal E}[ \langle A \rangle_f ]  = \langle A \rangle_f$, as desired.

    This completes our proof.
    From all of the above, we have
    \begin{equation} \label{eq: EH inequality}
        \varepsilon \geq \langle A \rangle_\mathcal{S} - \langle A \rangle_H \geq \frac{M}{\log 2} \left( \frac{1}{2^L 3^L} \right)^2  - 2/2^n = \frac{1}{\log 2} \frac{n/L^4}{36^L}  - 2/2^n,
    \end{equation}
    for any approximate unitary 2-design with additive error $\varepsilon$.
    We have assumed that $d = o(n)$ in 1D circuits and $d = o(\log n)$ in all-to-all connected circuits.
    Since $L \leq 2d$ in 1D and $L \leq 2^d$ in all-to-all circuits, these assumptions imply $L = o(n)$.
    Hence, the negative term in Eq.~(\ref{eq: EH inequality}) is small compared to the positive term and can be neglected.
    Thus, the inequality can be satisfied if an only if $L = \Omega(\log(n/\varepsilon))$.
    In 1D circuits, this yields $d = \Omega(\log(n/\varepsilon))$. 
    In all-to-all connected circuits, this yields $d = \Omega(\log \log(n/\varepsilon))$ as desired.
\end{proof}

\end{document}